\documentclass[12pt]{amsart}

\input{yuming.sty}

\title{Robust self-testing for nonlocal games with robust game algebras}
\author[Yuming Zhao]{Yuming Zhao}
\date{}
\address{QMATH, Department of Mathematical Sciences, University of Copenhagen}
\email{yuming@math.ku.dk}

\begin{document}

\begin{abstract}
    We give an operator-algebraic formulation of robust self-testing in terms of states on $C^*$-algebras. We show that a quantum correlation $p$ is a robust self-test only if among all (abstract) states, there is a unique one achieving $p$. We show that the ``if" direction of this statement also holds, provided that $p$ is optimal/perfect for a nonlocal game that has a robust game algebra. This last condition applies to many nonlocal games of interest, including all XOR games, synchronous games, and boolean constrained system (BCS) games.

    For those nonlocal games with robust game algebras, we prove that self-testing is equivalent to the uniqueness of finite-dimensional tracial states on the associated game algebra, and robust self-testing is equivalent to the uniqueness of amenable tracial states. Applying this tracial-state characterization of self-testing to parallel repetition, we show that a synchronous game is a  self-test for perfect quantum strategies if and only if its parallel repeated version is a self-test for perfect quantum strategies.

    As a proof approach, we give the first quantitative Gower-Hatami theorem that is applicable to $C^*$-algebras. Here ``quantitative" means there is a constructive bound on the distance between the approximate representations and exact representations. We also demonstrate how this quantitative Gowers-Hatami theorem can be used to calculate the explicit robustness function of a self-test. 
\end{abstract}

\maketitle
\tableofcontents

\section{Introduction}

Suppose we have a physical system consisting of two separate labs, each capable of making several different measurements. If these two labs are entangled, then the measurement statistics $p$ (which are referred to as a \textbf{correlation}) can be correlated in surprising ways that any classical theory cannot explain. In this scenario, the behavior of the system is described by a quantum state and measurement operators (this collection is called a \textbf{model} for $p$). In general, a given correlation can be realized by many different models. However, some correlations have a unique\footnote{In more rigorous terms, the model is unique up to changing of bases of local systems and tensoring auxiliary systems.} underlying model. A correlation with this property is called a \textbf{self-test}. 

In essence, self-testing allows us to infer the exact quantum state and measurements solely from the observed correlations. Tsirelson, around the 1980s, already observed such phenomena when studying Bell-type inequalities~\cite{Tsi93}, which led to the  ``nonlocality birth" of self-testing~\cite{youtub}. The term ``self-testing" and its ``cryptography birth" was given by Mayers and Yao in \cite{MY03}. Since then, self-testing has been applied in various areas of quantum information. In device-independent quantum cryptography, self-testing is arguably the most effective way to prove the trustworthiness of quantum devices~\cite{BSCA18a,BSCA18b}. In delegated quantum computations, self-testing can be used to verify the correctness of computation performed by a remote quantum server~\cite{CGJV19,Gri19,BMZ23}. Self-testing has also been used to certify entanglement and separate quantum correlation sets~\cite{CS18}. In quantum complexity theory, self-testing is an important tool for dealing with malicious quantum provers and proving the soundness of an interactive proof system. Notably, self-testing is one of the key techniques in the recent breakthrough $\text{MIP}^*=\text{RE}$~\cite{MIPRE}. This complexity-theoretical result has resolved several long-standing open problems in operator algebras, including Connes' embedding problem and Kitchberg's QWEP conjecture.

Historically, many self-testing results were derived in an ad-hoc manner. Recently, approximate representation theory has emerged as a powerful tool for establishing self-testing results. One notable example is the use of the Gowers-Hatami theorem for finite groups \cite{GH17,CS18,Vid18}. However, this framework is limited to correlations and nonlocal games exhibiting certain finite-group structures. The first analogue of the Gowers-Hatami theorem that is applicable to certain $C^*$-algebras was introduced in \cite{MPS21}, but it lacks a quantitative bound on the ``distance" between approximate and exact representations, so the resulting self-test has no constructive robustness. 

In parallel, self-testing has played a crucial role in resolving significant conjectures in operator algebras. Yet, none of the aforementioned frameworks offers a comprehensive operator-algebraic interpretation of self-testing.
On this front, Paddock, Slofstra, Zhao, and Zhou have made substantial progress by giving the following operator-algebraic formulation of self-testing in terms of states on $C^*$-algebras: 

\begin{theorem}[Informal version of Theorem 3.5 in \cite{PSZZ23}]\label{thm}
An extreme quantum correlation $p$ is a self-test for quantum models if and only if there is a unique finite-dimensional state on $\mintensor$ for $p$.
\end{theorem}

Here a model $S$ is called a \textbf{quantum model} if the quantum state employed in $S$ is \textbf{finite-dimensional}. Correlations that have quantum models are called \textbf{quantum correlations}. The set $C_q$ of quantum correlations is a convex set, but not closed~\cite{Slof19b}. A quantum correlation is said to be \textbf{extreme} if it is an extreme point of $C_q$. The closure of $C_q$ is denoted $C_{qa}$ and is referred to as the set of \textbf{quantum approximate correlations}. 

Given a finite set $X$ of measurement settings and a finite set $A$ of measurement outcomes, the \textbf{POVM algebra} $\POVM^{X,A}$ is the universal $C^*$-algebra generated by ``abstract POVMs" $\{e^x_a:a\in A\},x\in X$ in the sense that $e^x_a$'s are positive contractions satisfying the algebraic relations $\sum_{a\in A}e^x_a=\Id$ for all $x\in X$. An (abstract) state $f$ on a $C^*$-algebra $\mcA$ is a unital positive linear functional, and $f$ is said to be finite-dimensional if its GNS representation is finite-dimensional. We say that a state $f$ on $\mintensor$ achieves a correlation $p$ if $f(e^x_a\otimes e^y_b)=p(a,b|x,y)$ for all $x,y,a,b$. It is well-known that (see e.g., \cite{Fri12}) $C_{qa}$ consists of correlations that can be achieved by states (including infinite-dimensional ones) on $\mintensor$, and $C_q$ consists of correlations that can be achieved by \textbf{finite-dimensional} states on $\mintensor$.

\Cref{thm} only deals with \textbf{exact} self-test, meaning that it only studies models that can achieve the given correlation perfectly. However, due to the presence of noise in the environment, one may never observe the ideal correlation exactly. For practical applications, it is crucial to study the \textbf{robustness} of a self-test. A correlation $p$ is said to be a \textbf{robust self-test} if it has a unique model $\wtd{S}$ and any model that achieves a correlation close to $p$ must be close to $\wtd{S}$ in some suitable sense. One of the main goals of this paper is to generalize \Cref{thm} to robust self-testing. Thus, we ask:

\begin{question}\label{qu1}
    What is an operator-algebraic formulation of \textbf{robust self-testing} in terms of (abstract) states on $C^*$-algebras?
\end{question}

In \cite{PSZZ23}, if there is a unique state in a set of states $\mcS$ that can achieve a correlation $p$, then $p$ is said to be an \textbf{abstract state self-test} for $\mcS$. Using this language, \Cref{thm} asserts that self-testing for quantum models is equivalent to abstract state self-testing for finite-dimensional states on $\mintensor$. The study of finite-dimensional states versus infinite-dimensional states on $\mintensor$ is closely related to understanding the structure of quantum correlation sets. One could also ask what is an abstract state self-test for infinite-dimensional states.

\begin{question}\label{qu2}
    Suppose a quantum correlation $p$ can be realized by a unique finite-dimensional state. Are there any \textbf{infinite-dimensional states} that can achieve $p$?
\end{question}

The first main contribution of this paper is to illustrate that \Cref{qu1} and \Cref{qu2} are indeed the same question. Suppose a quantum correlation $p$ is a self-test for quantum models, and let $f$ be the unique finite-dimensional state achieving $p$. If in addition, $p$ is a robust self-test, we prove that there is \textbf{no infinite-dimensional state} that can achieve $p$, so $f$ is the unique state for $p$ among \textbf{all states} on $\mintensor$ (see \Cref{thm:robustuniquestate} for details). In other words, robust self-test for quantum models implies abstract state self-test for all states on $\mintensor$. 

Our second main contribution is to show that robust self-testing for quantum models and abstract state self-testing for all states are indeed equivalent, provided that the correlation is optimal/perfect for a nonlocal game that has a \textbf{robust game algebra} (see \Cref{thm:robustcorr}). Here the game algebra associated with a nonlocal game $\mcG$ is the quotient of Bob's POVM algebra $\POVM^{Y,B}$ by some relations $\mcR$ such that optimal strategies for $\mcG$ correspond to tracial states on $\POVM^{Y,B}\slash\ang{\mcR}$. This game algebra is said to be robust if in addition, near-optimal strategies correspond to states on $\POVM^{Y,B}$ that are approximately tracial and approximately respect the relations in $\mcR$. As shown in \cite{Pad22}, all XOR games~\cite{Tsi87,Slof11}, synchronous games~\cite{PSSTW16}, and boolean constrained system (BCS) games~\cite{Ji13,CM14} have robust game algebras. So our abstract-state characterization of robust self-testing applies to optimal/perfect quantum correlations of these games.

Our third main contribution is to show that, for those nonlocal games with robust game algebras, self-testing for optimal quantum strategies is equivalent to the uniqueness of \textbf{finite-dimensional tracial states} on the associated game algebra (see \Cref{thm:gametrace}); and robust self-testing is equivalent to the uniqueness of \textbf{amenable tracial states} (see \Cref{thm:robustgame}).

Notably, the above tracial-state characterization of (robust) self-testing has several interesting implementations. First, it provides \textbf{one-line proofs} for many robust self-testing results. For example, the result that CHSH game is a robust self-test follows immediately from the fact that its associated game algebra $Cl_2$, the Clifford algebra of rank 2, has a unique tracial state; the Mermin-Peres magic square game is a robust self-test because the tensor product of two $Cl_2$'s has a unique tracial state; all games that are based on presentations of Pauli groups $P_n$ --- including Pauli braiding test~\cite{NV17}, low-weight Pauli braiding test~\cite{BMZ23}, Pauli basis test~\cite{MIPRE}, quantum low-degree test~\cite{NV18}, and so on --- are robust self-tests because the full group $C^*$-algebra of $P_n$ has a unique tracial state.

Furthermore, this tracial-state characterization illustrates the necessary and sufficient conditions for \textbf{non-robust self-testing}: a nonlocal game is a self-test but not a robust self-test if and only if its associated game algebra has a unique finite-dimensional tracial state but has multiple amenable tracial states. The first example of such a game was constructed by Man\v{c}inska and Schmidt in \cite{MS23}. Given two games $\mcG_1$ and $\mcG_2$, they consider the $(\mcG_1\vee \mcG_2)$-game in a rough sense that the players win if they choose the same game ($\mcG_1$ or $\mcG_2$) and win it. Taking $\mcG_1$ to be a game that separates $C_q$ and $C_{qa}$ (e.g., \cite{Slof19b}) and taking $\mcG_2$ to be the magic square game, the game algebra $C^*(\mcG_1\vee \mcG_2)$ has a unique finite-dimensional tracial state given by the perfect strategy for $\mcG_2$, so $\mcG_1\vee \mcG_2$ is a self-test. However, $C^*(\mcG_1\vee \mcG_2)$ also has other infinite-dimensional amenable tracial states given by perfect strategies for $\mcG_1$, thus this self-test is not robust. 

This tracial-state characterization of self-testing applies to parallel repeated games. Parallel repetition is an important technique for gap amplification in complexity theory. Parallel repeated synchronous games are key ingredients in the $\text{MIP}^*=\text{RE}$ proof~\cite{MIPRE}. Parallel repeated CHSH game and parallel repeated magic square game have been shown to be robust self-tests in \cite{McK16,McK17, CN16}. In this paper, based on an algebraic formulation of products of synchronous games given in \cite{productsyn}, we prove that a synchronous game is a self-test for its perfect strategies if and only if its parallel repeated version is a self-test for perfect strategies (see \Cref{cor:parallelselftest}). 

Having a characterization for robust self-testing in terms of tracial states on $C^*$-algebras also raises the prospect of constructing new robust self-tests from algebras which are known to have unique tracial states (e.g., full matrix algebras, Clifford algebras of even rank, etc.). Efficient (fewer generators and relations) presentations of those algebras could result in efficient (fewer questions and answers) robust self-tests, which, in turn, can be used to improve soundness analysis in quantum interactive proofs (see e.g. \cite{CVY23} for recent progress). One of the potential approaches to tackling the quantum PCP conjecture (game version) is to construct a family of self-tests that are efficient ($O(log(n))$-bit question and $O(1)$-bit answer) and are highly robust (robustness is independent of $n$).

In the final part of this paper, we investigate how the stability of a game algebra encodes the robustness function of a self-test. To this end, we state and prove a \textbf{quantitative} Gowers-Hatami theorem for game algebras (see \Cref{thm:algGH}). To the best of our knowledge, this is the first Gowers-Hatami theorem for $C^*$-algebras that has a quantitative bound on the distance between approximate representations and exact representations. This theorem also identifies all factors that determine how robust a self-test is (see \Cref{rmk:GH}) and provides a systematic approach to computing the robustness function. Using this theorem, we give a new and succinct proof for the result that the CHSH game is a robust self-test with robustness function $O(\sqrt{\eps})$.

\begin{comment}
\subsection{Open problems and future directions}

non-robust self-test for correlations

equivalents between ....

robust self-testing in parallel 

applations of quantatiove algebraic GH

design high robust, effient self-test, quantum PCP \cite{CN24}

\subsection{Organization of the paper}
\end{comment}

The rest of this paper is organized as follows. In \Cref{sec:pre}, we review some background concepts on algebras and nonlocal correlations. In \Cref{sec:operatorselftest}, we state and prove our operator-algebraic formulation of robust self-testing. In \Cref{sec:gameselftest}, we discuss robust self-testing in the context of nonlocal games. In \Cref{sec:trace}, we discuss some properties of tracial states, with a focus on amenable tracial states and the convex structure of finite-dimensional tracial states. In \Cref{sec:traceselftest,sec:tracerobustselftest}, we study self-testing and robust self-testing using tracial states. In \Cref{sec:parallel}, we apply the tracial-state characterization of self-testing to parallel repeated synchronous games. In \Cref{sec:GH}, we state and prove our quantitative Gowers-Hatami theorem for game algebras and show how it can be used to calculate the robustness function of the CHSH self-test.

\subsection{Acknowledgements}
Part of this work was completed as the author's PhD thesis~\cite{YZ} at the University of Waterloo. The author thanks his committee members Michael Brannan, Richard Cleve, Matthew Kennedy, William Slofstra, and Jurij Vol\v{c}i\v{c} for valuable feedback. We thank Laura Man\v{c}inska for pointing out the reference \cite{productsyn} and suggesting the connections to parallel repetition. While preparing this manuscript, the author became aware of contemporaneous work \cite{KM24} which
establishes similar results.

\section{Preliminaries}\label{sec:pre}
In this section, we establish an algebraic framework for nonlocal games and correlations.

\subsection{Algebras and representations}
 In this paper, a $*$-algebra refers to a unital associative $\C$-algebra $\mcA$ equipped with an antilinear involution $ \mcA\arr\mcA: a \mapsto
a^*$
such that $(a b)^* = b^* a^*$ for all $a,b \in \mcA$. We use $1$ for the identity in any $*$-algebra. We also study the presentations of $*$-algebras. Given a set $\mcX$, we use $\C^*\ang{\mcX}$ to denote the free $*$-algebra with generating set $\mcX$. In other words, elements in $\C^*\ang{\mcX}$ are noncommutative $*$-polynomials over $\mcX$. For any $\mcR \subseteq \C^*\ang{\mcX}$, $\C^*\ang{\mcX : \mcR}$ denotes the
quotient of $\C^*\ang{\mcS}$ by the two-sided $*$-ideal generated by $\mcR$. 

If a $*$-algebra $\mcA$ has a submultiplicative Banach norm $\norm{\cdot}_{\mcA}$ such that the \emph{$C^*$-identity} $\norm{a^*a}=\norm{a}^2$ holds for all $a\in \mcA$ and $\mcA$ is closed with respect to $\norm{\cdot}_\mcA$, then $\mcA$ is said to be a $C^*$-algebra. For any Hilbert space $\mcH$, we use $\msB(\mcH)$ to denote the $C^*$-algebra of bounded operators on $\mcH$. All Hilbert spaces in this paper are assumed to be separable. We denote by $\Id_\mcH$ the identity operator in $\msB(\mcH)$ and simply write $\Id$ if $\mcH$ is clear from the context. The \emph{commutant} of a subset $X\subseteq\msB(\mcH)$ is $X':=\{T\in\msB(\mcH):TS=ST \text{ for all }S\in X \}$. 

The set of $d\times d$ matrices $M_{d}(\C)$ equipped with the operator norm is a $C^*$-algebra called \emph{full matrix algebra} and denoted $\M_d$. If we think of $M_{d}(\C)$ as a vector space, then equipped with the inner product $\ang{A,B}:=\frac{1}{n}\Tr(A^*B)$, $M_d(\C)$ is a $d^2$-dimensional Hilbert space.

 A \emph{$*$-homomorphism}
$\phi: \mcA \to \mcB$ between $*$-algebras $\mcA$ and $\mcB$ is an algebra homomorphism 
such that $\phi(a^*) = \phi(a)^*$. A \emph{representation} $\pi$ of a $*$-algebra $\mcA$ is a $*$-homomorphism from $\mcA\arr\msB(\mcH)$ for some Hilbert space $\mcH$.
If $\mcH$ is finite-dimensional then we say $\pi$ is \emph{finite-dimensional}. A \emph{subrepresentation} of $\pi$ is a closed subspace $\mcK\subseteq\mcH$ such that 
\begin{equation*}
   \pi(\mcA)\mcK:=\text{span}\{\pi(a)v:a\in\mcA,v\in\mcK \}=\mcK. 
\end{equation*}
A representation is \emph{irreducible} if it contains no proper non-zero subrepresentations. A vector $\ket{v}\in\mcH$ is \emph{cyclic} for a representation $\pi:\mcA\arr\msB(\mcH)$ if $\pi(\mcA)\ket{v}$ is dense in $\mcH$. A \emph{cyclic representation} of $\mcA$ is a tuple $(H,\pi,\ket{v})$, where $\pi$ is a representation of $\mcA$ on $\mcH$, and $\ket{v}\in \mcH$ is cyclic for $\pi$. Two representations $\pi:\mcA\arr\msB(\mcH)$ and $\varphi:\mcA\arr\msB(\mcK)$ of $\mcA$ are \emph{(unitarily) equivalent}, denoted $\pi\cong\varphi$, if there is a unitary $U:\mcH\arr\mcK$ such that $U\pi(a)U^*=\varphi(a)$ for all $a\in\mcA$. Two cyclic representations $(\mcH,\pi,\ket{v})$ and $(\mcK,\varphi,\ket{w})$ are (unitarily) equivalent if there is such a unitary with $U\ket{v}=\ket{w}$. 

A state on a $*$-algebra is a linear function $f : \mcA
\arr \C$ such that
\begin{enumerate}[(i)]
    \item $f(1)=1$, $f(a^*) = \overline{f(a)}$, $f(a^* a) \geq 0$ for all $a\in\mcA$, and
    \item $f$ is bounded in the sense that
    \begin{align}\label{eq:bounded}
    \sup\left\{\frac{f(b^* a^* a b)}{f(b^* b)}: b \in \mcA, f(b^* b) \neq 0\right\}<\infty
\end{align}
for all $a\in\mcA$.
\end{enumerate}
 Every (bounded) state on a $*$-algebra $\mcA$ has a unique (up to equivalence) cyclic representation $(\mcH,\pi,\ket{v})$ where $\ket{v}$ is a unit vector in $\mcH$ and $f(a) = \bra{\psi}\pi(a)\ket{\psi}$ for all $a \in \mcA$
\cite[Theorem 4.38]{Sch20}. Every $*$-algebra $\mcA$ has an operator norm
$\norm{\cdotp}_\mcA:\mcA\arr\R_{\geq 0}\cup \{\infty\}$ given by the supremum over all states on $\mcA$. In other words,  
\begin{equation*}
    \norm{a}_\mcA:=\sup\{\sqrt{f(a^*a)}: f \text{ a state on } \mcA\}.
\end{equation*}
For any Hilbert space $\mcH$, we denote by $\norm{\cdot}_{op}$ the operator norm on $\msB(\mcH)$. A state $f$ is said to be \emph{tracial} if $f(ab)=f(ba)$ for all $a,b\in\mcA$. Every full matrix algebra $\M_d$ has a unique tracial state $\tr_d$ given by $\tr_d(T):=\frac{1}{d}\Tr(T)$.

%A state $\psi$ on a $*$-algebra $\mcA$ induces a semi-norm $\norm{a}_\psi:=\sqrt{\psi(a^*a)}$ on $\mcA$ satisfying $\abs{\psi(ab)}\leq \norm{a}_\psi\norm{b}_{\psi}$ for all $a,b\in \mcA$. If $\psi$ is tracial, then $\norm{\cdotp}_{\psi}$ is unitarily invariant, meaning that $\norm{uav}_{\psi} = \norm{a}_{\psi}$ for all unitaries $u,v \in \mcA$ and elements $a \in \mcA$.  For an arbitrary state $\psi$, the seminorm $\norm{\cdotp}_{\psi}$ is left unitarily invariant, but not necessarily right unitarily invariant.

%Using the GNS theorem, it can be shown that
%\begin{equation*}
%    \norm{a}_\mcA=\sup\{\norm{\pi(a)}_{op}:\pi \text{ a }
%\text{representation} \text{ on a Hilbert space }\mcH\}
%\end{equation*}
%and
%$\psi(b^*a^*ab)\leq \norm{a}^2_{\mcA}\psi(b^*b)$ for all $a,b\in\mcA$ and
%states $\psi$ on $\mcA$. 

%A \emph{$*$-positive cone}  of a $*$-algebra $\mcA$ is a subset $\mcA_{+}\subset \mcA$ such that (i) all elements in $\mcA_+$ are hermitian, (ii) $\lambda\cdot 1\in\mcA_+ $ and $\lambda a+b\in \mcA_+$ for all $a,b\in\mcA_+$ and $\lambda\in \mcR_{\geq 0}$, and (iii) $x^*ax\in\mcA_{+}$ for all $a\in\mcA_{+}$ and $x\in\mcA$. We say $\mcA_+$ has the \emph{archimedean property} if for any $x\in\mcA$, there is a $\lambda>0$ such that $\lambda1-x^*x\in\mcA_+$. A $*$-algebra $\mcA$ together with a $*$-positive cone $\mcA_+$ that satisfies the archimedean property is called a \emph{semi-pre-$C^*$-algebra}.  

In this paper, we often work with finite-dimensional systems. A state $f$ on a $*$-algebra is \emph{finite-dimensional} if the Hilbert space $\mcH$ in the GNS representation $(H,\pi,\ket{v})$ of $f$ is finite-dimensional. If $\pi:\mcA\arr\msB(\mcH)$ is a finite-dimensional representation, then by the \emph{structure theory for finite-dimensional $C^*$-algebras}, $H$ is isomorphic to $\bigoplus_{i=1}^\ell \C^{n_i} \otimes \C^{m_i}$ for some integers $n_i$, $m_i$, $i=1,\ldots,\ell$, and 
\begin{equation}
    \pi(A)\cong \bigoplus_{i=1}^\ell M_{n_i}(\C)\otimes \Id_{m_i}\,,\text{ and}\quad \pi(A)'\cong \bigoplus_{i=1}^\ell \Id_{n_i}\otimes M_{m_i}(\C).
\end{equation}
In particular, $\pi(A)$ and $\pi(A)'$ are direct sums of matrix algebras.

%\subsection{More on $C^*$-algebras}
Let $\mcA$ be a $C^*$-algebra. The Gelfand–Naimark theorem states that $\mcA$ can be represented as a $C^*$-subalgebra of some $\msB(\mcH)$. In other words, there is a faithful representation $\pi:\mcA\arr\msB(\mcH)$ for some Hilbert space $\mcH$. For every $a\in \mcA$, $\norm{a}^2-a^*a$ is positive in $\mcA$, so any linear functional $f:\mcA\arr\C$ satisfying $f(1)=1$ and $f(a^*a)\geq 0$ for all $a\in\mcA$ is automatically bounded in the sense of \Cref{eq:bounded}. The set of all states on $\mcA$, denoted $\msS(\mcA)$, is called the state space of $\mcA$. $\msS(\mcA)$ is convex and weak*-compact. The extreme points of $\msS(\mcA)$ are called \emph{pure states}. Let $f\in\msS(\mcA)$ and let $(\mcH,\pi,\ket{v})$ be a GNS representation of $f$. Then $f$ is a pure state if and only if $\pi$ is an irreducible representation.

Given two $C^*$-algebras $\mcA$ and $\mcB$, a linear mapping $\varphi:\mcA\arr\mcB$ is said to be \emph{positive} if $\varphi(x)$ is positive in $\mcB$ for any positive element $x\in\mcA$. A linear mapping $\Phi:\mcA\arr\mcB$ is said to be \emph{unital completely positive (ucp)} if 
\begin{enumerate}[(i)]
    \item $\Phi(1_\mcA)=1_\mcB$, and
    \item for any $n\in \N$, the mapping
    \begin{align*}
        \Phi_n:M_n(\mcA)&\arr M_n(\mcB)\\
        \begin{pmatrix}
            a_{11} & \cdots & a_{1n}\\
            \vdots & \ddots & \vdots\\
            a_{n1} & \cdots & a_{nn}
        \end{pmatrix} &\mapsto         \begin{pmatrix}
           \Phi( a_{11}) & \cdots & \Phi(a_{1n})\\
            \vdots & \ddots & \vdots\\
            \Phi(a_{n1}) & \cdots & \Phi(a_{nn})
        \end{pmatrix}
    \end{align*}
    is positive.
\end{enumerate}
The Stinespring dilation theorem states that for any ucp map $\Phi:\mcA\arr\msB(\mcH)$, there exists a Hilbert space $\mcK$, a representation $\pi:\mcA\arr\msB(\mcK)$, and an isometry $V:\mcH\arr\mcK$ such that
\begin{equation*}
    \Phi(a)=V^*\pi(a)V
\end{equation*}
for all $a\in\mcA$. Moreover, if $\mcH$ is finite-dimensional, then $\mcK$ can be taken as finite-dimensional.

Given two $C^*$-algebras $\mcA$ and $\mcB$, their algebraic tensor product $\mcA\otimes \mcB$ is a $*$-algebra. There is more than one way to make $\mcA\otimes\mcB$ into a $C^*$-algebra. The first is to define the max-norm $\norm{\cdot}_{max}$ on $\mcA\otimes\mcB$ by
\begin{equation}
    \norm{a}_{max}:=\sup\{\norm{\pi(a)}_{op}:\pi\text{ a representation of }\mcA\otimes\mcB \text{ on a Hilbert space }\mcH\} \label{eq:maxtensor}
\end{equation}
The \emph{max tensor product} of $\mcA$ and $\mcB$, denoted $\mcA\otimes_{max}\mcB$, is the closure of $\mcA\otimes\mcB$ with respect to $\norm{\cdot}_{max}$. For any representation $\pi:\mcA\otimes_{max}\mcB\arr\msB(\mcH)$ there are pairs of representations $\pi_A:\mcA\arr\msB(\mcH)$ and $\pi_B:\mcB\arr\msB(\mcH)$ such that $\pi(a\otimes b)=\pi_A(a)\pi_B(b)=\pi_B(b)\pi_A(a)$ for all $a\in\mcA,b\in\mcB$. One can also define the min-norm. Fix faithful representations $\mcA\subset \msB(\mcH)$ and $\mcB\subset \msB(\mcK)$, then $\mcA\otimes \mcB\subset \msB(\mcH\otimes \mcK)$ as a $*$-subalgebra. Restricting the operator norm on $\msB(\mcH\otimes \mcK)$ to $\mcA\otimes \mcB$ defines the min-norm $\norm{\cdot}_{min}$. The $C^*$-algebra $\mcA\otimes_{min}\mcB$ is the closure of $\mcA\otimes\mcB$ with respect to $\norm{\cdot}_{min}$. Note that $\mcA\otimes_{min}\mcB$ is independent of the choice of faithful representations $\mcA\subset \msB(\mcH)$ and $\mcB\subset \msB(\mcK)$.

Given two $C^*$-algebras $\mcA$ and $\mcB$, we say a state $f:\mcA\otimes \mcB\arr\C$ is \emph{max-continuous} (resp. \emph{min-continuous}) if it is continuous with respective to $\norm{\cdot}_{max}$ (resp. $\norm{\cdot}_{min}$). Equivalent, $f$ is max-continuous (resp. min-continuous) if it extends to a state on $\mcA\otimes_{max}\mcB$ (resp. $\mcA\otimes_{min}\mcB$). By our definition of states, every state on $\mcA\otimes \mcB$ is max-continuous.

\subsection{Vectors and measurements}
we use the bra-ket notation for vectors in Hilbert spaces. 
A \emph{vector state} is a unit vector in some Hilbert space. Given two vectors $\ket{v}$ and $\ket{w}$, we write $\ket{v}\approx_{\eps}\ket{w}$ to denote $\norm{\ket{v}-\ket{w}}\leq \eps$.

Given two vector states $\ket{\alpha}\in \mcH$ and $\ket{\beta}\in \mcK$, we often write $\ket{\alpha,\beta}$ to denote the product state $\ket{\alpha}\otimes\ket{\beta}\in \mcH\otimes\mcK$. We use $\{\ket{i}\}_{i=1}^d$ to denote the standard basis for the Euclidean space $\C^d$. Given two Hilbert space $\mcH_A$ and $\mcH_B$, every bipartite vector state $\ket{\psi}\in\mcH_A\otimes \mcH_B$ has a \emph{Schmidt decomposition} 
\begin{equation*}
    \ket{\psi}=\sum_{i\in\mcI}\lambda_i\ket{\alpha_i}\otimes\ket{\beta_i},
\end{equation*}
where the index set $\mcI$ is either finite or countable, the \emph{Schmidt coefficients} $\lambda_i$ are strictly positive, and $\{\ket{\alpha_i}\}_{i\in\mcI}$ and $\{\ket{\beta_i}\}_{i\in\mcI}$  are orthonormal subsets of $\mcH_A$ and $\mcH_B$ respectively. The cardinality of $\mcI$ is called the \emph{Schmidt rank} of $\ket{\psi}$.  
%For a derivation of the Schmidt decomposition in infinite-dimensional spaces, see e.g. \cite[Theorem A.5]{CLP17}. 

Given a Hilbert space space $\mcH$, we call
\begin{equation*}
    C_1(\mcH):=\{T\in\msB(\mcH):\sum_{n=1}^\infty s_n(T)< \infty\}
\end{equation*}
the \emph{trace class operators}, where $s_1(T)\geq s_2(T)\geq \cdots \geq 0$ are the \emph{singular values} of $T$. For any $T\in C_1(\mcH)$ and any orthonormal basis $\{\ket{\alpha_i}\}_{i\in\mcI}$ for $\mcH$,
\begin{equation}
    \sum_{i\in \mcI}\bra{\alpha_i}T\ket{\alpha_i}\label{eq:trace}
\end{equation}
converges and the value is independent of the choice of the orthonormal basis. We denote by $\Tr(T)$ the value of \Cref{eq:trace} and call it the \emph{trace of $T$}. An operator $\rho\in C_1(\mcH)$ is called a \emph{density operator} (or a quantum state) on $\mcH$ if $\rho$ is a positive operator with $\Tr(\rho)=1$. We say that a density operator $\rho$ on $\mcH$ is a \emph{pure quantum state} if $\rho=\ket{\psi}\bra{\psi}$ for some vector state $\ket{\psi}\in\mcH$, and we say that $\rho$ is a \emph{mixed quantum state} if it is not a pure quantum state. Every density operator $\rho$ on $\mcH$ induces a semi-norm $\norm{X}_\rho:=\sqrt{\Tr\big(X^*X\rho\big)}$ on $\msB(\mcH)$ which is called the $\rho$-norm. This norm is left unitarily invariant in the sense that 
\begin{equation*}
    \norm{UX}_\rho=\norm{X}_\rho
\end{equation*}
for any $X\in\msB(\mcH)$ and unitary operator $U$ on $\mcH$.

A \emph{positive operator-valued measure} (abbrev. POVM) on a Hilbert space $\mcH$ 
with finite index set $\mcI$ is a collection of positive operators $\{M_i:i\in \mcI\}$ on $\mcH$
such that $\sum_{i\in I}M_i=\Id$. A POVM $\{M_i\}_{i \in \mcI}$ is said to be a
\emph{projection-valued measure} (abbrev. PVM) if in addition, $M_i^*=M_i=M_i^2$ for all $i \in \mcI$. 
%\emph{Naimark's dilation theorem}states that for any POVM $\{M_i\}_{i \in \mcI}$ on a Hilbert space $\mcH$, there is a Hilbert space $K$, an isometry $V : \mcH \to \mcK$, and a PVM $\{P_i\}_{i \in \mcI}$ such that $M_i=V^*P_i V$.

Given a PVM $\{M_1,\ldots,M_m\}\subset\msB(\mcH)$, the corresponding \emph{observable} is the unitary operator $A\in\msB(\mcH)$ of order $m$ defined by
\begin{equation*}
    A:=\sum_{k=1}^m\exp\left(\frac{2\pi \sqrt{-1}}{m}k\right)M_k.
\end{equation*}
If $m=2$, then $A=M_1-M_2$ is hermitian and we call it a \emph{binary observable}.

\subsection{Nonlocal games and correlations}

Let $X$, $Y$, $A$, and $B$ be finite sets. A \emph{correlation} $p$ is a collection of conditional probabilities $\{p(a,b|x,y)\} \in \R_{\geq 0}^{A \times B \times X \times Y}$ such that 
\begin{equation*}
    \sum_{(a,b) \in A \times B} p(a,b|x,y) = 1
\end{equation*}
for all $x \in X$, $y \in Y$.  A \emph{quantum model} $\models$ consists of 
\begin{enumerate}
    \item finite-dimensional Hilbert spaces $\mcH_A$ and $\mcH_B$,
    \item POVMs $\{M^x_a\}_a\in A,x\in X $ on $\mcH_A$ and POVMs $\{N^y_b\}_b\in B,y\in Y$ on $\mcH_B$, and 
    \item a unit vector $\ket{\psi}\in\mcH_A\otimes\mcH_B$.
\end{enumerate}
If $\{M^x_a:a\in A\},x\in X$ and $\{N^y_b,b\in B\},y\in Y$ are
PVMs then we say $S$ is a \emph{projective}
quantum model. A
quantum model is \emph{full-rank} if $\dim \mcH_A = \dim \mcH_B$, and the Schmidt
rank of $\ket{\psi}$ is $\dim \mcH_A$. We say $\models$ is a model for a correlation $p$ or $S$ achieves $p$ if 
\begin{equation*}
    p(a,b|x,y) = \braket{\psi| M^x_a \otimes N^y_b|\psi}
\end{equation*}
for all $(a,b,x,y) \in A \times B \times X \times Y$. We often denote by $p_S$ the correlation achieved by $S$. 

We use $C_{q}(X,Y,A,B)$ to denote the set of correlations in $\R_{\geq 0}^{A
\times B \times X \times Y}$ that can be achieved by quantum models. We use $C_q$ to denote the union of $C_q(X,Y,A,B)$ over all finite sets $X,Y,A,B$. The set
$C_q$ is not closed \cite{Slof19b}, and the closure of $C_q$ is denoted by $C_{qa}$.

In addition to the above models that exhibit a tensor product structure, Algebraic Quantum Field Theory suggests a commuting-operator framework for correlations. A \emph{commuting operator model} $S=\big(\mcH,\{M^x_a\},\{N^y_b\}\big)$ for a correlation $p$ consists of
\begin{enumerate}[]
    \item a Hilbert space $\mcH$, 
    \item POVMs $\{M^x_a : a \in A\},x \in X$ and $\{N^y_b : b \in B \},y \in Y$ on $\mcH$ such that
    \begin{equation*}
        M^x_a N^y_b = N^y_b M^x_a
    \end{equation*}
     for all $(a,b,x,y) \in A \times B \times X \times Y$, and
    \item a unit vector $\ket{\psi} \in \mcH$ 
\end{enumerate}
such that
\begin{equation*}
    p(a,b|x,y) = \braket{\psi| M^x_a \cdot N^y_b|\psi}
\end{equation*}
for all $(a,b,x,y) \in A \times B \times X \times Y$. We use $C_{qc}$ to denote the set of
correlations that can be achieved by commuting operator models. The set $C_{qc}$ is
closed and convex.

We refer to the correlations in $C_q,C_{qa}$, and $C_{qc}$ as quantum correlations, quantum approximate correlations, and quantum commuting correlations respectively. 

Operationally, one can think that correlations arise from an interactive game: a referee samples a question pair $(x,y)$ according to a distribution $\mu$ on $X\times Y$, sends $x$ to the player Alice, and sends $y$ to the player Bob; Alice and Bob then return $a\in A$ and $b\in B$ respectively. Their behavior is captured by a correlation $p\in\R^{A\times B\times X\times Y}_{\geq 0}$ where $p(a,b|x,y)$ indicates the probability that Alice and Bob return $a$ and $b$ upon receiving $x$ and $y$. Based on a predicate $V:A\times B\times X\times Y \arr\{0,1\}$, the referee then determines whether the players win ($V(a,b|x,y)=1$) or lose ($V(a,b|x,y)=0$). We call the tuple $\mcG:=(X,Y,A,B,\mu, V)$ a \emph{nonlocal game}. Alice and Bob know the rules of $\mcG$ and can strategize together, but they are not allowed to communicate once the game begins. 

Given a nonlocal game $\mcG=(X,Y,A,B,\mu,V)$ and a correlation $p\in \R^{A\times B\times X\times Y}_{\geq 0}$, the winning probability of $p$ for $\mcG$ is
\begin{equation*}
    w(\mcG;p):=\sum_{a,b,x,y}\mu(x,y)V(a,b|x,y)p(a,b|x,y).
\end{equation*}
In the context of nonlocal games, we refer to models as strategies. Given a strategy $S$, we denote by $w(\mcG;S)$ (or simply $w(S)$ if $\mcG$ is clear from the context) the winning probability of $S$ for $\mcG$. That is, $w(\mcG;S)=w(\mcG;p)$ where $p$ is the correlation achieved by $S$. The \emph{quantum value} $w_q(\mcG)$ of $\mcG$ is defined to be the supremum of $w(\mcG;S)$ over all quantum strategies $S$ for $\mcG$. Since $C_{qa}$ is the closure of $C_q$,
\begin{align*}
    w_q(\mcG)=\sup\{w(\mcG;p):p\in C_q\}=\sup\{w(\mcG;p):p\in C_{qa}\}.
\end{align*}
In this paper, a correlation $p$ (resp. a strategy $S$) is said to be \emph{optimal}\footnote{If $w_q(\mcG)=1$, we often replace ``optimal" with ``perfect".} for $\mcG$ if $w(\mcG;p)=w_q(\mcG)$\footnote{One can also define the commuting-operator value $w_{co}(\mcG):=\sup \{w(\mcG;p):p\in C_{qc}\}$. In this paper, being optimal always means achieving the quantum value $w_q(\mcG)$, even when working with $p\in C_{qc}$ or infinite-dimensional commuting-operator strategies.} (resp. $w(\mcG;S)=w_q(\mcG)$).

%\subsection{Algebras for nonlocal correlations}
Quantum models for correlations can also be expressed as states on
$C^*$-algebras~\cite{PSZZ23}. Given finite sets $X$ and $A$, the \emph{POVM algebra}
$\POVM^{X,A}$ is the universal $C^*$-algebra generated by positive
contractions $e^x_a$, $x \in X$, $a \in A$, subject to the relations $\sum_{a
\in A} e^x_a = 1$ for all $x \in X$. By the universal property, a collection of operators $\{M^x_a,a\in A\},x\in X$ on a Hilbert space $\mcH$ are POVMs if and only if there is a representation $\pi : \POVM^{X,A} \to \msB(\mcH)$ sending $e^x_a\mapsto M^x_a$. When working with bipartite system $\POVM^{X,A}\otimes \POVM^{Y,B}$, we let $m^{x}_a := e^x_a \in \POVM^{X,A}$ and
$n^{y}_b := e^{y}_b \in \POVM^{Y,B}$. 
%We can also think of these generators as
%being inside the algebraic tensor product $\POVM^{X,A} \otimes_{alg} \POVM^{Y,B}$,
%where the product $m^{x}_a \cdot n^y_b := m^x_a \otimes n^y_b$. 

Given a quantum model $\models$, let $\pi_A:\POVM^{X,A}\arr\msB(\mcH_A)$ be the representation sending $m^x_a\mapsto M^x_a$ and let $\pi_B:\POVM^{Y,B}\arr\msB(\mcH_B)$ be the representation sending $n^y_b\mapsto N^y_b$, we refer to $\pi_A\otimes\pi_B$ as the \emph{associated representation of $S$}. The abstract state $f_S$ on $\POVM^{X,A} \otimes_{min} \POVM^{Y,B}$ defined by $f_S(x): = \braket{\psi|(\pi_A \otimes \pi_B)(x)|\psi}$ is finite-dimensional and achieves $p_S$ in the sense that 
\begin{equation}\label{eqn:state_to_correlation}
     f_S(m^x_a \otimes n^y_b)=\braket{\psi|\pi_A(m_a^x)\otimes \pi_B(n_b^y)|\psi}=p_S(a,b|x,y).
\end{equation}
We refer to $f_S$ as the \emph{abstract state defined by $S$}. Conversely, any finite-dimensional state on $\POVM^{X,A} \otimes_{min} \POVM^{Y,B}$ yields a quantum model $S$ such that $f = f_S$ \cite{SW08}. In other words, $C_q(X,Y,A,B)$ consists of correlations that can be achieved by finite-dimensional states on $\POVM^{X,A} \otimes_{min} \POVM^{Y,B}$~ \cite{Fri12,JNPPSW11}.  

We also work with PVM algebras $\PVM^{X,A}$, which is the quotient of $\POVM^{X,A}$ by the relations $(e^x_a)^2 = e^x_a$ for all $x \in X$, $a \in A$. Note that the resulting algebra $\PVM^{X,A}$ is isomorphic to the group $C^*$-algebra $C^* (\Z_{\abs{A}}^{*\abs{X}})$. We often work with unitary generators $\{a_x:x\in X\}$ and $\{b_y:y\in Y\}$ for $\PVM^{X,A}\otimes_{min}\PVM^{Y,B}$ where $a_x$'s and $b_y$'s are unitary elements with order $\abs{A}$ and $\abs{B}$ respectively.

\section{An operator-algebraic formulation of robust self-testing}\label{sec:operatorselftest}

In this section, we state and prove an operator-algebraic characterization of robust self-testing. To define a robust self-test, we need the following notion of local $\eps$-dilation.
\begin{definition}\label{def:roubstlocaldilation}
Given $\eps\geq 0$ and two quantum models 
    \begin{align*}
        &\models \text{ and } \\ 
        &\wtdmodels,
    \end{align*}
 we say $\wtd{S}$ is a local $\epsilon$-dilation of $S$, denoted $S\succeq_\epsilon\wtd{S}$, if there are isometries $I_A:\mcH_A\arr \wtd{\mcH}_A\otimes \mcH_A^{aux}$ and $I_B:\mcH_B\arr \wtd{\mcH}_B\otimes \mcH_B^{aux}$, and vector state $\ket{aux}\in \mcH_A^{aux}\otimes \mcH_B^{aux}$ such that 
    \begin{align}
     I_A\otimes I_B\big(M^x_a\otimes \Id\ket{\psi}  \big)&\approx_\epsilon \big(\wtd{M}^x_a\otimes\Id\ket{\wtd{\psi}}\big)\otimes\ket{aux}\label{eq:dilationA}\\
      I_A\otimes I_B\big(\Id\otimes N^y_b\ket{\psi}  \big)&\approx_\epsilon \big(\Id\otimes\wtd{N}^y_b\ket{\wtd{\psi}}\big)\otimes\ket{aux},\text{ and}\label{eq:dilationB}\\
        I_A\otimes I_B\ket{\psi}&\approx_{\epsilon} \ket{\wtd{\psi}}\otimes\ket{aux}.\label{eq:dilationstate} 
    \end{align}
    for all $(a,b,x,y)\in A\times B\times X\times Y$. We refer to a local $0$-dilation as a local dilation and write $S\succeq \wtd{S}$.
\end{definition}

\begin{remark}
    Note that, in the above definition, \Cref{eq:dilationA,eq:dilationB,eq:dilationstate} imply 
\begin{align}
I_A M^x_a I_A^*\otimes\Id\ket{\wtd{\psi},aux}&\approx_{2\epsilon}\big(\wtd{M}^x_a\otimes\Id\ket{\wtd{\psi}}\big)\otimes\ket{aux} \label{eq:a},\\
\Id\otimes I_B N^y_b I_B^*\ket{\wtd{\psi},aux}&\approx_{2\epsilon}\big(\Id\otimes\wtd{N}^y_b\ket{\wtd{\psi}}\big)\otimes\ket{aux} , \text{ and }\label{eq:b}\\
 I_A\otimes I_B\big(M^x_a\otimes N^y_b\ket{\psi}  \big)&\approx_{3\epsilon} \big(\wtd{M}^x_a\otimes\wtd{N}^y_b\ket{\wtd{\psi}}\big)\otimes\ket{aux}\label{eq:state}.
\end{align}
We often use \Cref{eq:state} (with $3\epsilon$ replaced by $\epsilon$) as the definition of $S\succeq_{\epsilon}\wtd{S}$. 
\end{remark}    

\begin{definition}\label{def:robustcorr}
    Let $\mcC$ be a class of quantum models. We say a quantum correlation $p\in C_q$ is a \textbf{robust self-test} for $\mcC$ if there is an ideal model $\wtd{S}\in\mcC$ such that the following holds. For any $\delta\geq0$, there is an $\epsilon\geq0$ such that $S\succeq_{\delta} \wtd{S}$ for any model $S\in \mcC$ satisfying $\norm{p_S-p}_1\leq \epsilon$.
\end{definition}

Here the $1$-norm of a vector $v\in\R^{A\times B\times X\times Y}$ is given by $\norm{v}_1:=\sum_{a,b,x,y}\abs{v(a,b|x,y)}$. It is easy to see that $S\succeq \wtd{S}$ implies $p_S=p$. Taking $\delta=0$, we can define an (exact) self-test: $p$ is a self-test for $\mcC$ if there is an ideal model $\wtd{S}\in\mcC$ such that $S\succeq \wtd{S}$ for any model $S\in\mcC$ for $p$. It is clear that every robust self-test is a self-test.  The $\delta$-$\eps$ dependence is called the robustness of a self-test and is a crucial concept in the applications of self-testing. If $\delta$ is a function of $\eps$, then we say $\delta(\eps)$ is the \emph{robustness function} of this self-test. For correlations, whether every self-test is robust is still an open problem. We also discuss robust self-testing for nonlocal games. In this context, we often refer to models as \emph{strategies}.

\begin{definition}\label{def:gameselftest}
    Let $\mcC$ be a class of quantum models. A nonlocal game $\mcG$ is a self-test for its optimal quantum strategies in $\mcC$ if there is an optimal quantum strategy $\wtd{S}\in \mcC$ for $\mcG$ such that $S\succeq \wtd{S}$ for any optimal quantum strategy $S\in\mcC$ for $\mcG$. 
\end{definition}

In \Cref{def:gameselftest}, $\wtd{S}$ is referred to as an \emph{ideal optimal quantum strategy for $\mcG$}. When $\mcC$ is the class of all quantum models, we simply say that $\mcG$ is a self-test for its optimal quantum strategies; when $\mcC$ is the class of all projective quantum models, we say that $\mcG$ is a self-test for its projective optimal quantum strategies. We also define robust self-testing for nonlocal games. We say a strategy $S$ is $\eps$-optimal for a nonlocal game $\mcG$ if $w(\mcG;S)\geq w_q(\mcG)-\eps$.

\begin{definition}
Let $\mcC$ be a class of quantum models. A nonlocal game $\mcG$ is a robust self-test for its optimal strategies in $\mcC$ if there is an ideal optimal quantum strategy $\wtd{S}\in\mcC$ for $\mcG$ such that the following statement holds. For any $\delta\geq 0$, there is an $\eps\geq 0$ such that $S\succeq_\delta\wtd{S}$ for any $\eps$-optimal strategy $S\in\mcC$ for $\mcG$.
\end{definition}

As we mentioned in the introduction, for nonlocal games, not every self-test is robust~\cite{MS23}.  We discuss robust self-testing for nonlocal games in more detail in \Cref{sec:gameselftest}.

One of the main purposes of this paper is to give an operator-algebraic formulation of robust self-testing in terms of states on $C^*$-algebras. We recall the following definition from \cite{PSZZ23}.

\begin{definition}[Definition 3.3 in \cite{PSZZ23}]\label{def:abstract_self-test}
    Let $\mcS$ be a subset of states on $\POVM^{X,A} \otimes_{min} \POVM^{Y,B}$. A correlation $p$ is an \textbf{abstract state self-test for $\mcS$} if there exists a unique state $f \in \mcS$ achieving $p$.
\end{definition}

In \cite{PSZZ23}, they show that when $\mcC$ is the class of all quantum models (or some other classes that have certain ``nice" properties) and $p$ is an extreme point in $C_q$, self-test in the standard sense is equivalent to abstract state self-test for finite-dimensional states.

\begin{theorem}[Part (a) of Corollary 3.6 in \cite{PSZZ23}] \label{thm:old}
Suppose $p\in C_q(X,Y,A,B)$ is an extreme point in $C_q$. Then $p$ is a self-test for the class of all quantum models if and only if $p$ is an abstract state self-test for finite-dimensional states on $\mintensor$.
\end{theorem}

The first main theorem of this paper is to show that \emph{robust self-test for quantum models} corresponds to \emph{abstract state self-test for all states}.

\begin{theorem}\label{thm:robustuniquestate}
    If $p\in C_q(X,Y,A,B)$ is a robust self-test for all quantum models, then $p$ is an abstract state self-test for all states on $\POVM^{X,A}\otimes_{\min} \POVM^{Y,B}$.
\end{theorem}

Compared to the ``only if" direction of \Cref{thm:old}, this theorem replaces ``self-test" and ``finite-dimensional states" with ``robust self-test" and ``all states".  Note that a quantum correlation $p\in C_q$ can always be achieved by \emph{some} finite-dimensional states. So if $f$ is the unique state for a quantum correlation, then $f$ must be finite-dimensional. This gives an alternative way to phrase \Cref{thm:robustuniquestate}:

\begin{theorem}[Restated from \Cref{thm:robustuniquestate}]\label{thm:robustuniquestate'}
    If  $p\in C_q(X,Y,A,B)$ is a robust self-test for all quantum models, then 
    \begin{enumerate}[(i)]
        \item there is a unique finite-dimensional state on $\POVM^{X,A}\otimes_{\min} \POVM^{Y,B}$ for $p$ and 
        \item there is no infinite-dimensional state on $\POVM^{X,A}\otimes_{\min} \POVM^{Y,B}$ that can achieve $p$.
    \end{enumerate}
\end{theorem}

We prove this theorem at the end of this section. Unlike \Cref{thm:old}, the above theorem is not an ``if-and-only-if" statement. Hence, we ask:

\begin{question}\label{question1}
Let $p\in C_q(X,Y,A,B)$ be an extreme point in $C_q$. Suppose $p$ is an abstract state self-test for all states on $\POVM^{X,A}\otimes_{\min} \POVM^{Y,B}$. Is $p$ a robust self-test for all quantum models?
\end{question}

Through \Cref{sec:gameselftest,sec:trace,sec:traceselftest,sec:tracerobustselftest}, we aim to give an affirmative answer to \Cref{question1} for a large class of quantum correlations.

\begin{theorem}\label{thm:mainrobust2}
Suppose $p\in C_q(X,Y,A,B)$ is the unique perfect quantum correlation for a synchronous game, or the unique optimal quantum correlation for an XOR game. Then $p$ is a robust self-test for all quantum models if and only if $p$ is an abstract state self-test for all states on $\POVM^{X,A}\otimes_{\min} \POVM^{Y,B}$. 
\end{theorem}

In the context of nonlocal games, the above theorem implies:
\begin{corollary}[Game version of \Cref{thm:mainrobust2}]
    Suppose $\mcG$ is a synchronous game with perfect quantum strategies or an XOR game. Then $\mcG$ is a robust self-test for its optimal quantum strategies if and only if there is a unique state on $\POVM^{X,A}\otimes_{\min} \POVM^{Y,B}$ that is optimal for $\mcG$.
\end{corollary}
Here a state $f$ on $\POVM^{X,A}\otimes_{\min} \POVM^{Y,B}$ is said to be optimal for $\mcG$ if the correlation achieved by $f$ is optimal for $\mcG$.

Our proof approach for \Cref{thm:mainrobust2} is based on algebraic characterizations of synchronous games and  XOR games established in \cite{Slof11,HMPS19,KPS18}. For any synchronous game (resp. XOR game) $\mcG$, one can associate a finitely-presented algebra $C^*(\mcG)$ such that perfect (resp. optimal) strategies for $\mcG$ correspond to tracial states on $C^*(\mcG)$. Furthermore, as shown in \cite{Pad22}, any near-perfect (resp. near-optimal) strategy $S$ for $\mcG$ defines an approximate state $\phi$ on $C^*(\mcG)$ which is also approximately tracial. Finally, using a lifting theorem for $C^*(\mcG)$, we lift $\phi$ to a ucp map $\theta$ on $C^*(\mcG)$ and construct the desired local dilation $I_A\otimes I_B$ from the Stinespring dilation of $\theta$. This last step is inspired by \cite{MPS21}.

The proof approach outlined above applies to any nonlocal game that exhibits structures similar to synchronous or XOR games. In particular, \Cref{thm:mainrobust2} also holds if $p$ is the unique perfect quantum correlation for a boolean constraint system (BCS) game. Indeed, we'll prove a general result (\Cref{thm:robustcorr}). Given a nonlocal game $\mcG$, we say an algebra $C^*(\mcG)$ is the associated game algebra for $\mcG$ if every optimal strategy for $\mcG$ corresponds to a tracial state on $C^*(\mcG)$. If in addition,
every near-optimal strategy for $\mcG$ corresponds to an approximate state on $C^*(\mcG)$ which is also approximately tracial, then we say $C^*(\mcG)$ is \emph{robust}.
In \Cref{sec:tracerobustselftest}, we show that if a game $\mcG$ has a robust game algebra $C^*(\mcG)$, then its self-testing (resp. robust self-testing) property can be characterized by finite-dimensional (resp. amenable) tracial states on $C^*(\mcG)$. These general results provide convenient ways to examine whether a nonlocal game is a (robust) self-test by studying the tracial states on the associated game algebras.

\begin{theorem}
[Restated from part of \Cref{thm:gametrace,cor:uniquetrace}, informal]\label{cor:mainrobust2} Suppose a nonlocal game $\mcG$ has an associated game algebra $C^*(\mcG)$. 
\begin{enumerate}[(a)]
    \item $\mcG$ is a self-test for its optimal quantum strategies if and only if $C^*(\mcG)$ has a unique finite-dimensional tracial state.
    \item Suppose in addition, $C^*(\mcG)$ is robust. If $C^*(\mcG)$ has a unique tracial state $\tau$ and $\tau$ is finite-dimensional, then $\mcG$ is a robust self-test for its optimal quantum strategies.
\end{enumerate}  
\end{theorem}

\subsection{Proof of Theorem \ref{thm:robustuniquestate}}

We prove \Cref{thm:robustuniquestate} in the rest of this section. Recall that for any residually finite-dimensional (RFD) $C^*$-algebras $\mcA$ and $\mcB$, the $C^*$-algebra $\mcA\otimes_{min}\mcB$ is also residually finite-dimensional, so the set of finite-dimensional states on $\mcA\otimes_{\min}\mcB$ forms a weak*-dense subset of the state space of $\mcA\otimes_{\min}\mcB$. Since the POVM algebra $\POVM^{X,A}$ is RFD for any finite sets $X$ and $A$, we have:
 \begin{restatable}{proposition}{AppendixA}\label{prop:dense}
      For every state $f$ on $\mintensor$ there is a sequence of quantum models $S_n,n\in\N$ such that $\lim_{n\arr \infty} f_{S_n}=f$ in the weak*-topology.
 \end{restatable}
 For the completeness of this manuscript, we present a proof of \Cref{prop:dense} in \Cref{AppendixA}. Before we prove \Cref{thm:robustuniquestate}, the following robust version of Proposition 4.8 in \cite{PSZZ23} is needed. We first recall the concept of \emph{centrally supported model} from \cite{PSZZ23}. Given a quantum model $\models$,
 the \emph{support} of $\ket{\psi}$ in $\mcH_A$ (resp. $\mcH_B$)
is the image of the reduced density matrix $\rho_A :=
\Tr_{\mcH_B}(\ket{\psi}\bra{\psi})$ (resp. $\rho_B =
\Tr_{\mcH_A}(\ket{\psi}\bra{\psi})$). The \emph{support projections} $\Pi_A$ and $\Pi_B$ of $\ket{\psi}$
are the self-adjoint projections onto the support of $\ket{\psi}$ in $\mcH_A$ and $\mcH_B$ respectively. We say that $S$ is \emph{centrally supported} if 
\begin{align*}
    [\Pi_A,M^x_a]=[\Pi_B,N^y_b]=0 
\end{align*}
for all $(a,b,x,y)\in A\times B\times X\times Y$. Any full-rank model is centrally supported since the
support projections are the identity operators. Conversely, if $S$ is centrally supported, then it is full-rank when restricted to its support. 

 \begin{proposition}\label{prop:diff}
    Let 
    \begin{align*}
        &\models \text{ and } \\ 
        &\wtdmodels
    \end{align*}
    be two quantum models with associated representations $\pi_A\otimes \pi_B$ and $\wtd{\pi}_A\otimes\wtd{\pi}_B$ respectively. If $\wtd{S}$ is centrally supported and $S\succeq_\epsilon\wtd{S}$ via local isometry $I_A\otimes I_B$ and vector state $\ket{aux}$, then for every $k,\ell\in \N$ and monomials $\alpha=m^{x_1}_{a_1}\cdots m_{x_k}^{a_k}\in \POVM^{X,A},\beta=n^{y_1}_{b_1}\cdots n^{y_\ell}_{b_\ell}\in\POVM^{Y,B}$ we have
    \begin{align}
        &\big(\wtd{\pi}_A(\alpha)\otimes\Id\ket{\wtd{\psi}}\big)\otimes\ket{aux}\approx_{(2k+1)\epsilon} I_A\pi_A(\alpha)I_A^*\otimes\Id\ket{\wtd{\psi},aux},\label{eq:ka}\\
        & \big(\Id\otimes\wtd{\pi}_B(\beta)\ket{\wtd{\psi}}\big)\otimes\ket{aux}\approx_{(2\ell+1)\epsilon} \Id\otimes I_B\pi_B(\beta)I_B^*\ket{\wtd{\psi},aux}\label{eq:lb}
    \end{align}
and $\abs{f_S(\alpha\otimes \beta)-f_{\wtd{S}}(\alpha\otimes\beta)}\leq 2(k+\ell+1)\epsilon$.
\end{proposition}

\begin{proof}
    We first prove \Cref{eq:ka} by induction on the monomial degree $k\in\N$. The base cases $k=0,1$ follow straight \Cref{eq:a}. Suppose \Cref{eq:ka} holds for all monomials in $\POVM^{X,A}$ of degree $k$. For any given monomial $\alpha=\alpha_1\cdots\alpha_k\alpha_{k+1}$ in $\POVM^{X,A}$ of degree $k+1$, let $\alpha':=\alpha_1\cdots\alpha_k$. Since $\wtd{S}$ is centrally supported, by \cite[Proposition 4.5]{PSZZ23} there is an operator $\wtd{F}\in\msB(\wtd{\mcH}_B)$ with $\norm{\wtd{F}}\leq\norm{\wtd{\pi}_A(\alpha_{k+1})}\leq 1$ such that $\wtd{\pi}_A(\alpha_{k+1})\otimes\Id_{\wtd{\mcH}_B}\ket{\wtd{\psi}}=\Id_{\wtd{\mcH}_A}\otimes\wtd{F}\ket{\wtd{\psi}}$. Note that $\norm{\wtd{\pi}_A(\alpha')}\leq 1$ and $\norm{I_A\pi_A(\alpha')I_A^*}\leq 1$. Thus by the inductive hypothesis,
    \begin{align*}
        \big(\wtd{\pi}_A(\alpha)\otimes\Id_{\wtd{\mcH}_B}\ket{\wtd{\psi}}\big)\otimes\ket{aux}&=\big(\wtd{\pi}_A(\alpha')\wtd{\pi}_A(\alpha_{k+1})\otimes\Id_{\wtd{\mcH}_B}\ket{\wtd{\psi}}\big)\otimes\ket{aux}\\
        &= \big(\wtd{\pi}_A(\alpha')\otimes\wtd{F}\ket{\wtd{\psi}}\big)\otimes\ket{aux}\\
        &\approx_{(2k+1)\epsilon} I_A\pi_A(\alpha')I_A^*\otimes \wtd{F}\otimes \Id_{\mcH_B^{aux}}\ket{\wtd{\psi},aux}\\
        &=\Big(I_A\pi_A(\alpha')I_A^*\big(\wtd{\pi}_A(\alpha_{k+1})\otimes\Id_{\mcH_A^{aux}} \big)   \Big)\otimes \Id_{\wtd{\mcH}_B\otimes \mcH_B^{aux}}\ket{\wtd{\psi},aux}\\
        &\approx_{2\epsilon} I_A\pi_A(\alpha')I_A^*I_A\pi_A(\alpha_{k+1})I_A^*\otimes \Id_{\wtd{\mcH}_B\otimes \mcH_B^{aux}}\ket{\wtd{\psi},aux}\\
        &=I_A\pi_A(\alpha)I_A^*\otimes \Id\ket{\wtd{\psi},aux}.
    \end{align*}
This implies  
\begin{equation*}
\big(\wtd{\pi}_A(\alpha)\otimes\Id\ket{\wtd{\psi}}\big)\otimes\ket{aux}\approx_{\big(2(k+1)+1\big)\epsilon}I_A\pi_A(\alpha)I_A^*\otimes \Id\ket{\wtd{\psi},aux}.
\end{equation*}
We conclude that \Cref{eq:ka} holds for all $k\in \N$ and monomials in $\POVM^{X,A}$ of degree $k$. The proof of \Cref{eq:lb} is similar. The rest follows from the following lemma by taking $\ket{\alpha}=I_A\pi_A(\alpha)I_A^*\otimes\Id\ket{\wtd{\psi},aux},\ket{\beta}=I_B\pi_B(\beta)I_B^*\ket{\wtd{\psi},aux},\ket{\wtd{\alpha}}=\big(\wtd{\pi}_A(\alpha)\otimes\Id\ket{\wtd{\psi}}\big)\otimes\ket{aux}$, and $\ket{\wtd{\beta}}=\big(\Id\otimes\wtd{\pi}_B(\beta)\ket{\wtd{\psi}}\big)\otimes\ket{aux}$.
\end{proof}

\begin{lemma}\label{lemma:diff}
    If $\ket{\alpha},\ket{\wtd{\alpha}},\ket{\beta},\ket{\wtd{\beta}}$ are four vectors with norm $\leq 1$ such that $\ket{\alpha}\approx_{\delta_1}\ket{\wtd{\alpha}}$ and $\ket{\beta}\approx_{\delta_2}\ket{\wtd{\beta}}$, then $\abs{\braket{\alpha|\beta}-\braket{\wtd{\alpha}|\wtd{\beta}}}\leq \delta_1+\delta_2$.
\end{lemma}
\begin{proof}
    Observe that $\braket{\alpha|\beta}-\braket{\wtd{\alpha}|\wtd{\beta}}=\big(\bra{\alpha}-\bra{\wtd{\alpha}}\big)\ket{\beta}+\bra{\wtd{\alpha}}\big(\ket{\beta}-\ket{\wtd{\beta}}\big)$. The lemma follows from the Cauchy-Schwartz inequality.
\end{proof}

Now we are ready to prove \Cref{thm:robustuniquestate}.

\begin{proof}[Proof of \Cref{thm:robustuniquestate}]
    Suppose $\wtd{p}\in C_q(X,Y,A,B)$ is a robust self-test for all quantum models, and let $\wtd{S}$ be an ideal model. By \Cref{thm:old}, there is a unique finite-dimensional state $\wtd{f}$ on $\mintensor$ achieving $\wtd{p}$. Assume for the sake of contradiction that there is an infinite-dimensional state $f$ on $\mintensor$ achieving $\wtd{p}$. By \Cref{prop:dense} there exists a sequence of quantum models $S_n,n\in\N$ such that $f=\lim\limits_{n\arr\infty}f_{S_n}$ in the weak*-topology. This implies 
    \begin{equation*}
       \wtd{p}(a,b|x,y)=f(m^x_a\otimes n^y_b)=\lim\limits_{n\arr\infty}f_{S_n}(m^x_a\otimes n^y_b)=\lim\limits_{n\arr\infty}p_{S_n}(a,b|x,y) 
    \end{equation*}
    for all $x,y,a,b$. By the definition of robust self-testing, there is a function $\eta:\N\arr\R_{\geq 0}$ with $\lim\limits_{n\arr\infty}\eta(n)=0$ such that $S_n\succeq_{\eta(n)} \wtd{S}$ for all $n\in\N$. Then for all monomials $\alpha\in\POVM^{X,A}$ and $\beta\in\POVM^{Y,B}$, by \Cref{prop:diff},
    \begin{equation*}
        \abs{f_{S_n}(\alpha\otimes \beta)-\wtd{f}(\alpha\otimes\beta)}\leq 2\big(\deg(\alpha)+\deg(\beta)+1\big)\eta(n)\arr 0 
    \end{equation*}
     as $n\arr \infty$. This means $\lim\limits_{n\arr\infty}f_{S_n}=\wtd{f}$ in the weak*-topology. Hence $\wtd{f}=f$, a contradiction. We conclude that there is no infinite-dimensional state on $\mintensor$ for $p$.
\end{proof}

Although \Cref{thm:robustuniquestate} is stated for the class of all quantum models, from the above proof, it is easy to see that similar results can be established for any class that is closed (in the sense of \cite{PSZZ23}) and contains a full-rank model. In particular, if $p\in C_q$ has a full-rank projective quantum model and $p$ is a robust self-test for projective quantum models, then $p$ is an abstract state self-test for all projective states on $\mintensor$ (or equivalent, all states on $\PVM^{X,A}\otimes_{min}\PVM^{Y,B}$).

\section{From correlations to nonlocal games}\label{sec:gameselftest}

In this section, we take a closer look at self-testing in the context of nonlocal games.
For any nonlocal game $\mcG=(X,Y,A,B,\mu,V)$, we define its \emph{game polynomial} to be the $*$-polynomial 
\begin{align*}
    \Phi_{\mcG}:=\sum_{x,y,a,b}\mu(x,y) V(a,b|x,y)m^x_a\otimes n^y_b
\end{align*}
in $\C^*\ang{m^x_a:x\in X,a\in A}\otimes \C^*\ang{n^y_b:y\in Y,b\in B}$. The optimal quantum value $w_q(\mcG)$ of the game $\mcG$ is equal to the supremum value of $f(\Phi_\mcG)$ over all states $f$ on $\POVM^{X,A}\otimes_{min}\POVM^{Y,B}$. The same conclusion holds if we replace $\POVM$ with $\PVM$.

We say a state $f$ on $\POVM^{X,A}\otimes_{min}\POVM^{Y,B}$ (or $\PVM^{X,A}\otimes_{min}\PVM^{Y,B}$) is optimal for $\mcG$  if $f(\Phi_\mcG)=w_q(\mcG)$. Equivalently, $f$ is optimal for $\mcG$ if the quantum approximate correlation achieved by $f$ is optimal for $\mcG$. Since $C_q$ consists of correlations that can be achieved by finite-dimensional states, $\mcG$ has an optimal quantum strategy if and only if there exists a finite-dimensional optimal state on $\POVM^{X,A}\otimes_{min}\POVM^{Y,B}$ for $\mcG$. The following lemma is straightforward from the GNS construction.
\begin{lemma}\label{lemma:spectrum}
    Let $\mcG=(X,Y,A,B,\mu,V)$ be a nonlocal game. For any $*$-representation $\pi$ of $\POVM^{X,A}\otimes_{min}\POVM^{Y,B}$, the operator $\pi(\Phi_{\mcG})$ has spectrum $\spec\big(\pi(\Phi_{\mcG}) \big)\subseteq [0,w_q(\mcG)]$.
\end{lemma}

 From \Cref{def:gameselftest}, we immediately see that:

\begin{lemma}\label{lemma:uniqueCq}
    A nonlocal game $\mcG$ is a self-test for its optimal quantum strategies (resp. projective optimal quantum strategies) if and only if $\mcG$ has a unique optimal quantum correlation $p$ and $p$ is a self-test for all quantum models (resp. projective quantum models).
\end{lemma}

The following lemma implies that if a nonlocal game is a self-test for its optimal quantum strategies (or optimal projective quantum strategies), then its unique optimal quantum correlation must be an extreme point in $C_q$. 

\begin{lemma}\label{lemma:uniqueextreme}
    If a nonlocal game $\mcG$ has a unique optimal quantum correlation $p$, then $p$ must be an extreme point in $C_q$. 
\end{lemma}

\begin{proof}
    Suppose $\mcG$ has a unique optimal quantum correlation $p$. Assume for the sake of contradiction that $p$ is not an extreme point in $C_q$. Then $p=\lambda p_1+(1-\lambda)p_2$ for some $0< \lambda<1$ and $p_1\neq p_2$ in $C_q$. Note that $p$ is optimal. So
    \begin{align*}
        w(\mcG;p)=\lambda w(\mcG;p_1)+(1-\lambda) w(\mcG;p_2)\leq \lambda w(\mcG;p)+(1-\lambda) w(\mcG;p)=w(\mcG;p).
    \end{align*}
It follows that $w(\mcG;p_1)=w(\mcG;p_2)=w(\mcG;p)$. The uniqueness of $p$ implies that $p=p_1=p_2$, which is a contradiction. Hence $p$ must be an extreme point in $C_q$.
\end{proof}

On the other hand, if a nonlocal game has a unique optimal state $f$, then $f$ is an extreme point in the state space:
\begin{lemma}\label{lemma:gameabstractstate}
    Let $\mcG=(X,Y,A,B,\mu,V)$ be a nonlocal game. Suppose that there is a unique finite-dimensional optimal state $f$ on $\POVM^{X,A}\otimes_{min}\POVM^{Y,B}$ (or $\PVM^{X,A}\otimes_{min}\PVM^{Y,B}$) for $\mcG$. Then $f$ is a pure state, and the correlation $p$ achieved by $f$ is the unique optimal quantum correlation for $\mcG$. 
\end{lemma}
\begin{proof}
    Let $p$ be the correlation achieved by $f$. Since $w(\mcG;p)=f(\Phi_\mcG)=w_q(\mcG)$, $p$ is optimal for $\mcG$. Assume $\mcG$ has another optimal quantum correlation $p'$. Then there is a finite-dimensional state $f'$ that can achieve $p'$. So $f'$ is optimal for $\mcG$, but $f'\neq f$, a contradiction. We conclude that $p$ is the unique optimal quantum correlation for $p$. By \Cref{lemma:uniqueCq}, $p$ is an extreme point in $C_q$. Since any state that can achieve $p$ is optimal for $\mcG$, $f$ must be the unique finite-dimensional state for $\mcG$. By \cite{PSZZ23}, the GNS representation of $f$ is irreducible. It follows that $f$ is a pure state. 
\end{proof}

\Cref{lemma:uniqueextreme,lemma:gameabstractstate} allows us to translate many self-testing results established for extremal quantum correlations in \cite{PSZZ23}  to self-testing for nonlocal games.

\begin{theorem}[Game version of Corollary 3.6 in \cite{PSZZ23}]\label{cor:mainresult2game}
Let $\mcG=(X,Y,A,B,\mu,V)$ be a nonlocal game.
\begin{enumerate}[(a)]
    \item $\mcG$ is a self-test for its optimal quantum strategies if and only if there is a unique finite-dimensional optimal state on $\POVM^{X,A}\otimes_{min}\POVM^{Y,B}$ for $\mcG$.

    \item If $\mcG$ has a full-rank projective optimal quantum strategy, then $\mcG$ is a self-test for its projective optimal strategies if and only if there is a unique finite-dimensional optimal state on $\PVM^{X,A}\otimes_{min}\PVM^{Y,B}$ for $\mcG$. 
\end{enumerate} 
\end{theorem}

\begin{theorem}[Restated from Theorem 4.1 and Theorem 4.3 in \cite{lifiting23}]\label{thm:mainresult3game}
    For any nonlocal game $\mcG$, the following statements hold.
    \begin{enumerate}[(a)]
        \item If $\mcG$ is a self-test for its optimal quantum strategies, then $\mcG$ has a full-rank projective ideal optimal strategy and $\mcG$ is a self-test for its projective optimal quantum strategies.
        \item If $\mcG$ is a self-test for its projective optimal quantum strategies and $\mcG$ has a full-rank projective optimal strategy, then $\mcG$ is a self-test for its optimal quantum strategies.
    \end{enumerate}
\end{theorem}

\begin{corollary}[Game version of Theorem 3.7 in \cite{PSZZ23}]\label{cor:newcorgame}
    If $\mcG$ has a full-rank projective optimal quantum strategy, then the following statements are equivalent.
    \begin{enumerate}[(1)]
    \item $\mcG$ is a self-test for its optimal quantum strategies.
    \item $\mcG$ is a self-test for its projective optimal quantum strategies.
    \item There is a unique finite-dimensional optimal state on $\POVM^{X,A}\otimes\POVM^{Y,B}$ for $\mcG$.
    \item There is a unique finite-dimensional optimal state on $\PVM^{X,A}\otimes\PVM^{Y,B}$ for $\mcG$.
\end{enumerate}
\end{corollary}

Recall from \cite{PSZZ23} that if there is a unique finite-dimensional state $f$ achieving an extreme quantum correlation $p$, then $f$ has a GNS representation $(\mcH,\pi,\ket{\psi})$ where $\pi=\pi_A\otimes\pi_B$ is a tensor product of two irreducible representations. Furthermore, $p$ is a self-test for quantum models and there is an ideal model with associated representation $\pi_A\otimes\pi_B$. We can sharpen this statement when working with self-testing for nonlocal games:

\begin{theorem}\label{thm:gameidealrep}
    Suppose a nonlocal game $\mcG=(X,Y,A,B,\mu,V)$ has a unique finite-dimensional optimal state $\wtd{f}$ on $\POVM^{X,A}\otimes_{min}\POVM^{Y,B}$. Then the following statements hold.
    \begin{enumerate}[(a)]
        \item $\wtd{f}$ has a GNS representation $(\wtd{\mcH}_A\otimes \wtd{\mcH}_B,\wtd{\pi}_A\otimes\wtd{\pi}_B,\ket{\wtd{\psi}})$ where $\wtd{\pi}_A$ and $\wtd{\pi}_B$ are irreducible $*$-representations of $\POVM^{X,A}$ and $\POVM^{Y,B}$ respectively.
        \item Let $\wtd{M}^x_a:=\wtd{\pi}_A(m^x_a)$ and $\wtd{N}^y_b:=\wtd{\pi}_B(n^y_b)$. Then $\mcG$ is a self-test for its optimal quantum strategies with an ideal optimal strategy
        \begin{align*}
            \wtdmodels.
        \end{align*}
        \item $\wtd{\pi}:=\wtd{\pi}_A\otimes\wtd{\pi}_B$ is the unique (up to unitary equivalence) finite-dimensional irreducible $*$-representation of $\POVM^{X,A}\otimes_{min}\POVM^{Y,B}$ with the property that the maximal eigenvalue of $\wtd{\pi}(\Phi_{\mcG})$ is $w_q(\mcG)$. Here $\Phi_\mcG$ is the game polynomial of $\mcG$.
        \item The $w_q(\mcG)$-eigenspace of $\wtd{\pi}(\Phi_{\mcG})$ is the one-dimensional space spanned by $\ket{\wtd{\psi}}$.
    \end{enumerate}
\end{theorem}

\begin{proof}
    Let $p$ be the correlation achieved by $\wtd{f}$. By \Cref{lemma:uniqueCq,lemma:gameabstractstate}, $p$ is an extreme point in $C_q$, $p$ is the unique optimal quantum correlation for $\mcG$, and $\wtd{f}$ is the unique finite-dimensional state on $\POVM^{X,A}\otimes_{min}\POVM^{Y,B}$ achieving $p$. This proves parts (a) and (b).
    
    Now suppose $\pi:\POVM^{X,A}\otimes_{min}\POVM^{Y,B}\arr\msB(\mcH)$ is a finite-dimensional irreducible $*$-representation such that the maximal eigenvalue of $\pi(\mcG)$ is $w_q(\mcG)$. Let $\ket{\psi}$ be a unit vector in the $w_q(\mcG)$-eigenspace of $\pi(\mcG)$. Since $\pi(\POVM^{X,A}\otimes_{min}\POVM^{Y,B})=\msB(\mcH)$, the vector $\ket{\psi}$ is cyclic for $\pi$. Hence $(\mcH,\pi,\ket{\psi})$ is a GNS representation of the state $f$ given by $f(\alpha):=\bra{\psi}\pi(\alpha)\ket{\psi},\alpha\in \POVM^{X,A}\otimes_{min}\POVM^{Y,B}$. This also means $f$ is a finite-dimensional state on $\POVM^{X,A}\otimes_{min}\POVM^{Y,B}$ such that $f(\Phi_\mcG)=\bra{\psi}\pi(\Phi_\mcG)\ket{\psi}=w_q(\mcG)$. By the uniqueness of $\wtd{f}$, we have $f=\wtd{f}$, and hence $(\mcH,\pi,\ket{\psi})$ is unitarily equivalent to $(\wtd{\mcH}_A\otimes \wtd{\mcH}_B,\wtd{\pi},\ket{\wtd{\psi}})$. This proves part (c).
    
    To see part (d), take an arbitrary unit vector $\ket{\psi}$ in the $w_q(\mcG)$-eigenspace of $\wtd{\pi}(\Phi_\mcG)$. Then the state $f$ defined by $f(\alpha):=\bra{\psi}\wtd{\pi}(\alpha)\ket{\psi},\alpha\in \POVM^{X,A}\otimes_{min}\POVM^{Y,B}$ is optimal for $\mcG$. So $f=\wtd{f}$. It follows that $(\wtd{\mcH}_A\otimes \wtd{\mcH}_B,\wtd{\pi},\ket{\psi})$ is a GNS representation for $\wtd{f}$. This means there is a unitary $U$ on $\wtd{\mcH}_A\otimes \wtd{\mcH}_B$ such that 
    \begin{align}
    U\ket{\psi}=\ket{\wtd{\psi}}, \text{ and }\label{eq:U}\\
        U\wtd{\pi}(\alpha)U^*=\wtd{\pi}(\alpha) \label{eq:Upi}
    \end{align}
    for all $\alpha\in\POVM^{X,A}\otimes_{min}\POVM^{Y,B}$. Note that $\wtd{\pi}$ is irreducible. \Cref{eq:Upi} implies $U=\lambda\Id$ for some $\lambda\in\C$. Then \Cref{eq:U} implies $\ket{\psi}$ is linearly dependent with $\ket{\wtd{\psi}}$. Since $\ket{\psi}$ was arbitrary, part (d) follows.
\end{proof}

Recall that the set of quantum approximate correlations $C_{qa}$ is the closure of $C_q$. The following lemma is a robust version of \Cref{lemma:uniqueCq}, although the proof is less straightforward. 

\begin{lemma}\label{lemma:robustunique}
    A nonlocal game $\mcG$ is a robust self-test for its optimal quantum strategies (resp. projective optimal quantum strategies) if and only if $\mcG$ has a unique quantum approximate correlation $p$ and $p$ is a robust self-test for all quantum models\footnote{This last condition automatically implies that $p$ is in $C_q$.} (resp. projective quantum models).
\end{lemma}

\begin{proof}
    For the ``only if" direction, suppose $\mcG=(X,Y,A,B,V)$ is a robust self-test for its optimal quantum strategies, and let $\wtd{S}$ be an ideal optimal strategy. By \Cref{lemma:uniqueCq}, the quantum correlation $\wtd{p}$ generated by $\wtd{S}$ is the unique optimal quantum correlation for $\mcG$. We claim that $\wtd{p}$ is indeed the unique optimal quantum approximate correlation for $\mcG$. Assume for the sake of contradiction that $\mcG$ has an optimal correlation $p\in C_{qa}$ while $p$ is not in $C_q$. Then there is an infinite-dimensional state $f$ achieving $p$. By \Cref{prop:dense}, there is a sequence of quantum models $S_n,n\in \N$ such that $\lim\limits_{n\arr\infty}f_{S_n}=f$ in the weak-$*$ topology. This implies $w(\mcG;S_n)=f_{S_n}(\Phi_\mcG)\arr f(\Phi_\mcG) =w_q(\mcG)$ as $n\arr \infty$. Since $\mcG$ is a robust self-test for its optimal quantum strategies, there is a function $\eta:\N\arr\R_{\geq 0}$ with $\lim\limits_{n\arr\infty}\eta(n)=0$ such that $S_n\succeq_{\eta(n)}\wtd{S}$. It follows that
 $p=\lim\limits_{n\arr\infty} p_{S_n}=p_{\wtd{S}}=\wtd{p}$, which contradicts to the assumption that $p$ is not in $C_q$. We conclude that $\wtd{p}$ is the unique quantum approximate correlation for $\mcG$. Note that for any quantum model $S$, $\norm{p_S-\wtd{p}}_1\leq \epsilon$ implies $w(\mcG;S)\geq w_q(\mcG)-\epsilon$. So $\mcG$ is a robust self-test as is $p$.

 For the ``if" direction, suppose that $\mcG$ has a unique optimal quantum approximate correlation $\wtd{p}$ and that $\wtd{p}$ is a robust self-test for quantum models, and let $\wtd{S}$ be an ideal model for $\wtd{p}$. Then $\wtd{p}$ must be in $C_q$. Assume for the sake of contradiction that $\mcG$ is not a robust self-test. This means there exists a $\delta>0$ and a sequence of quantum models $S_n,n\in\N$ with $\epsilon_n:=w_q(\mcG)-w(\mcG;S_n)\arr 0$ as $n\arr \infty$ such that for all $n\in\N$, $S_n\succeq_\delta \wtd{S}$ does not hold. Let $p$ be an accumulation point of $\{p_{S_n},n\in\N\}$. Then $p$ is in $C_{qa}$ and $p$ is optimal for $\mcG$. The uniqueness of $\wtd{p}$ implies that $p=\wtd{p}$. By passing a subsequence, we may assume without loss of generality that $\lim\limits_{n\arr\infty}p_{S_n}=\wtd{p}$. Since $\wtd{p}$ is a robust self-test for quantum models, there is a function $\eta:\N\arr\R_{\geq 0}$ with $\lim\limits_{n\arr\infty}\eta(n)=0$ such that $S_n\succeq_{\eta(n)}\wtd{S}$ for all $n\in\N$. Take a large enough $N$ such that $\eta(N)\leq \delta$, we see that $S_N\succeq_{\delta} \wtd{S}$, a contradiction. We conclude that $\mcG$ is a robust self-test for its optimal quantum strategies.
\end{proof}

We can also state a game version of \Cref{thm:robustuniquestate} (or equivalently \Cref{thm:robustuniquestate'}).

\begin{theorem}[Game version of \Cref{thm:robustuniquestate'}]\label{thm:robustuniquestate'game} If $\mcG=(X,Y,A,B,\mu,V)$ is a robust self-test for its optimal quantum strategies, then 
\begin{enumerate}[(i)]
    \item there is a unique finite-dimensional optimal state on $\POVM^{X,A}\otimes_{min}\POVM^{Y,B}$ for $\mcG$ and
    \item there is no infinite-dimensional optimal state for $\mcG$. 
\end{enumerate}

\end{theorem}

\subsection{Spectral gap of a self-test} Suppose $\mcG=(X,Y,A,B,V)$ is a self-test for its optimal strategies (or equivalently, $\mcG$ has a unique finite-dimensional optimal state $f$ on $\mintensor$). Let $\wtdmodels$ be the ideal optimal strategy with associated representation $\wtd{\pi}=\wtd{\pi}_A\otimes\wtd{\pi}_B$, and let $\Phi_\mcG$ be the game polynomial. Part (c) of \Cref{thm:gameidealrep} states that $w_q(\mcG)$ is the largest eigenvalue of $\wtd{\pi}(\Phi_\mcG)$ and $\ket{\psi}$ is the unique $w_q(\mcG)$-eigenstate. Let $\lambda$ be the second largest eigenvalue of $\wtd{\pi}(\Phi_\mcG)$. We refer to $\Delta:=w_q(\mcG)-\lambda$ as the \textbf{spectral gap} of $\mcG$. Note that the spectral gap is only defined for games that are self-tests. 

In many cases (for instance, the use of Gowers-Hatami theorem for proving robust self-testing), given an $\eps$-optimal strategy $\models$ for $\mcG$, one can first show that $\norm{I_A^*(\wtd{M}^x_a\otimes \Id_{\mcK_A}) I_A - M^x_a }_{\rho_A}\text{ and }\norm{I_B^*(\wtd{N}^y_b\otimes \Id_{\mcK_B}) I_B - N^y_b }_{\rho_B}$ are bounded by some $\delta>0$ for some isometries $I_A,I_B$. Here $\rho_A$ and $\rho_B$ are reduced density matrix of $\ket
{\psi}$ on $\mcH_A$ and $\mcH_B$ respectively. The following proposition illustrates that the ``distance" between $S$ and $\wtd{S}$ under the local isometry $I_A\otimes I_B$ is completely determined by $\eps,\delta$, and the spectral gap $\Delta$. Similar results have also been established in \cite{lifiting23,MPS21}.

\begin{proposition}\label{prop:gameidealrep}
 Suppose a nonlocal game $\mcG=(X,Y,A,B,\mu,V)$ is a self-test for its optimal quantum strategies. Let $\wtd{f}$ be the unique finite-dimensional optimal state $\wtd{f}$ on $\POVM^{X,A}\otimes_{min}\POVM^{Y,B}$ for $p$, let $\wtdmodels$ and  $\wtd{\pi}=\wtd{\pi}_A\otimes\wtd{\pi}_B$ be as in \Cref{thm:gameidealrep}, and  let $\Delta$ be the spectral gap of $\mcG$. Suppose
    \begin{align*}
        \models
    \end{align*}
    is an $\epsilon$-optimal quantum strategy for $\mcG$. Let $I_A:\mcH_A\arr \wtd{\mcH}_A\otimes \mcK_A$ and $I_B:\mcH_B\arr \wtd{\mcH}_B\otimes \mcK_B$ be two isometries. For every $(x,y,a,b)\in X\times Y\times A\times B$, let
    \begin{align*}
    \delta_{x,a}^A:= \norm{I_A^*(\wtd{M}^x_a\otimes \Id_{\mcK_A}) I_A - M^x_a }_{\rho_A}\text{ and } \delta_{y,b}^B:=\norm{I_B^*(\wtd{N}^y_b\otimes \Id_{\mcK_B}) I_B - N^y_b }_{\rho_B}
    \end{align*}
    where $\rho_A:=\Tr_{\mcH_B}\big(\ket{\psi}\bra{\psi}\big),\rho_B:=\Tr_{\mcH_A}\big(\ket{\psi}\bra{\psi}\big)$, and let 
    \begin{align*}
        \delta:= \sum_{x,y,a,b}\mu(x,y)V(a,b|x,y)\left( \delta_{x,a}^A + \delta_{y,b}^B\right).
    \end{align*}
     Then there exists a unit vector $\ket{aux}\in \mcK_A\otimes \mcK_B$ 
    such that
    \begin{align}
        \norm{I_A\otimes I_B \ket{\psi} - \ket{\wtd{\psi}}\otimes \ket{aux}}\leq \frac{\sqrt{2}\big(\delta+\sqrt{\eps}\big)}{\Delta}.\label{eq:difference}
    \end{align}
Moreover, for every $(x,y,a,b)\in X\times Y\times A\times B$,
\begin{align}
   \norm{I_A\otimes I_B \big(M^x_a\otimes N^y_b \ket{\psi}\big)- \big(\wtd{M}^x_a\otimes \wtd{N}^y_b\ket{\wtd{\psi}}\big)\otimes\ket{aux}}\leq \delta^A_{x,a}+\delta^B_{y,b}+ \frac{\sqrt{2}\big(\delta+\sqrt{\eps}\big)}{\Delta}.\label{eq:differencemore}
\end{align}
\end{proposition}

\begin{proof}
    Let $\pi:=\pi_A\otimes\pi_B$ be the associated representation of $S$, let $I=I_A\otimes I_B$, and let $\mcK:=\mcK_A\otimes\mcK_B$.
    By \Cref{lemma:diff},
\begin{align*}
    \norm{ \Big(I^*\big(\wtd{\pi}(m^x_a\otimes &n^y_b)\otimes\Id_\mcK\big)I-\pi(m^x_a\otimes n^y_b)\Big)\ket{\psi}}\\
   &\leq\norm{ I_A^*\big(\wtd{\pi}_A(m^x_a)\otimes\Id_{\mcK_A}\big)I_A-\pi_A(m^x_a)}_{\rho_A}+\norm{ I_B^*\big(\wtd{\pi}_B(n^y_b)\otimes\Id_{\mcK_B}\big)I_B-\pi_B(n^y_b)}_{\rho_B}\\
   &= \delta_{x,a}^A+\delta_{y,b}^B.
\end{align*}
This implies 
\begin{align*}
    \norm{\Big(I^*\big(\wtd{\pi}(\Phi_\mcG)\otimes\Id_\mcK\big)I-\pi(\Phi_\mcG)\Big)\ket{\psi}}\leq  \sum_{x,y,a,b}\mu(x,y)V(a,b|x,y)\left( \delta_{x,a}^A + \delta_{y,b}^B\right) =\delta.
\end{align*}
Since $S$ is an $\eps$-optimal strategy for $\mcG$, by \Cref{lemma:spectrum},  $w_q(\mcG)\Id-\pi(\Phi_\mcG)$ is a positive operator such that $\norm{w_q(\mcG)\Id-\pi(\Phi_\mcG)}_{op}\leq  1$ and $\bra{\psi}\big( w_q(\mcG)\Id-\pi(\Phi_\mcG)\big)\ket{\psi}\leq\eps$. So
\begin{align*}
    \norm{\big( w_q(\mcG)\Id-\pi(\Phi_\mcG)\big)\ket{\psi}}&\leq\norm{\big( w_q(\mcG)\Id-\pi(\Phi_\mcG)\big)^{1/2}}_{op}\cdot \norm{\big( w_q(\mcG)\Id-\pi(\Phi_\mcG)\big)^{1/2}\ket{\psi}}\\
    &\leq  \sqrt{\bra{\psi}\big( w_q(\mcG)\Id-\pi(\Phi_\mcG)\big)\ket{\psi}}\leq\sqrt{\eps}.
\end{align*}
Let $\wtd{\Phi}:=w_q(\mcG)\Id-\wtd{\pi}(\Phi_\mcG)$. It follows that
\begin{align*}
    \norm{I^*(\wtd{\Phi}\otimes\Id)I\ket{\psi}}\leq \norm{\Big(I^*\big(\wtd{\pi}(\Phi_\mcG)\otimes\Id\big)I-\pi(\Phi_\mcG)\Big)\ket{\psi}}+\norm{\big( w_q(\mcG)\Id-\pi(\Phi_\mcG)\big)\ket{\psi}}\leq 
    \delta+\sqrt{\eps}.
\end{align*}
Now we fix orthonormal bases $\{\ket{\alpha_i}\}_i$ and $\{\ket{\beta_j} \}_j$ for $\mcK_A$ and $\mcK_B$ respectively. We can write $I\ket{\psi}=\sum_{i,j}\lambda_{ij}\ket{\psi_{ij}}\otimes\ket{\alpha_i,\beta_j}$ for some unit vectors $\ket{\psi_{ij}}\in \wtd{\mcH}_A\otimes\wtd{\mcH}_B$ and some $\lambda_{ij}\in \C$ where $\sum_{ij}\abs{\lambda_{ij}}^2=1$. By absorbing phase of $\ket{\psi_{ij}}$ into $\lambda_{ij}$, we may assume without loss of generality that every $\delta_{ij}:=\braket{\psi_{ij}|\wtd{\psi}}\geq 0$. So we can write $\ket{\psi_{ij}}=\delta_{ij}\ket{\wtd{\psi}}+\sqrt{1-\delta_{ij}^2}\ket{\kappa_{ij}}$  for some unit vector $\ket{\kappa_{ij}}\in \wtd{\mcH}_A\otimes\wtd{\mcH}_B$ that is orthogonal to $\ket{\wtd{\psi}}$. \Cref{thm:gameidealrep} implies $\norm{\wtd{\Phi}\ket{\wtd{\psi}}}=0$ and $\norm{\wtd{\Phi}\ket{\kappa_{ij}}}\geq \Delta$ for all $i,j$. Hence
\begin{align*}
    (\delta+\sqrt{\eps})^2&\geq \norm{I^*(\wtd{\Phi}\otimes\Id)I\ket{\psi}}^2=\sum_{i,j}\abs{\lambda_{ij}}^2\norm{\wtd{\Phi}\ket{\psi_{ij}}  }^2\\
    &=\sum_{i,j}\abs{\lambda_{ij}}^2(1-\delta_{ij}^2)\norm{\wtd{\Phi}\ket{\kappa_{ij}}  }^2\geq \Delta^2 \sum_{i,j}\abs{\lambda_{ij}}^2(1-\delta_{ij}^2)\\
    &\geq \Delta^2 \sum_{i,j}\abs{\lambda_{ij}}^2(1-\delta_{ij})=\Delta^2\left(1-\sum_{i,j}\abs{\lambda_{ij}}^2\delta_{ij}\right).
\end{align*}
The last inequality uses the fact that $\delta_{ij}\leq 1$ for all $i,j$. Let $\ket{\kappa}:=\sum_{i,j}\lambda\ket{\alpha_i,\beta_j}$. Then
\begin{align*}
    \norm{I\ket{\psi}-\ket{\wtd{\psi}}\otimes\ket{\kappa}}^2=2\left(1-\sum_{i,j}\abs{\lambda_{ij}}^2\delta_{ij}\right)\leq \frac{2(\delta+\sqrt{\eps})^2}{\Delta^2}
\end{align*}
So \Cref{eq:difference} follows. For any $(x,y,a,b)\in X\times Y\times A\times B$, observe that
\begin{align*}
    I\big(M^x_a\otimes N^y_b \ket{\psi}\big)- \big(\wtd{M}^x_a\otimes \wtd{N}^y_b\ket{\wtd{\psi}}\big)\otimes\ket{aux}&=  \big( I_A M^x_a  -(\wtd{M}^x_a\otimes\Id_{\mcK_A})I_A \big)\otimes I_B N^y_b \ket{\psi}\\
    &+(\wtd{M}^x_A\otimes\Id_{\mcK_A})\otimes \big(  
 I_B N^y_b - (\wtd{N^y_b}\otimes \Id_{\mcK_B})I_B\big)\ket{\psi}\\
 &+\wtd{M}^x_a\otimes\wtd{N}^y_b\otimes \Id
\big(I \ket{\psi} - \ket{\wtd{\psi}}\otimes \ket{aux}\big).
\end{align*}
Since $I_BN^y_b,\wtd{M}^x_a$, and $\wtd{M}^x_a\otimes\wtd{N}^y_b$ all have operator norms $\leq 1$, we conclude that
\begin{align*}
    \norm{I \big(M^x_a\otimes N^y_b \ket{\psi}\big)- \big(\wtd{M}^x_a\otimes \wtd{N}^y_b\ket{\wtd{\psi}}\big)\otimes\ket{aux}}&\leq \norm{I_A^*(\wtd{M}^x_a\otimes \Id_{\mcK_A}) I_A - M^x_a }_{\rho_A}\\
    &+\norm{I_B^*(\wtd{N}^y_b\otimes \Id_{\mcK_B}) I_B - N^y_b }_{\rho_B}\\
    &+\norm{I \ket{\psi} - \ket{\wtd{\psi}}\otimes \ket{aux}}.
\end{align*}
So \Cref{eq:differencemore} follows.
\end{proof}
 
\begin{remark}\label{rmk:gameidealrep}
    We remark that \Cref{thm:gameidealrep} and \Cref{prop:gameidealrep} also hold if all $\POVM$ are replaced by $\PVM$ and ``optimal quantum strategies" is replaced by ``projective optimal strategies". The proofs follow similarly. 
\end{remark}

\section{Tracial states}\label{sec:trace}
%In this section, we show that for XOR games, synchronous games, and --- more generally --- nonlocal games that have associated game algebras, their self-testing properties can be characterized by tracial states on game algebras. 

%\subsection{Tracial states}
In this section, we establish some properties of tracial states that will be used later. Recall that, in this paper, all $C^*$-algebras are assumed to be unital and separable. A tracial state $\tau$ on a $C^*$-algebra $\mcA$ is a state satisfying $\tau(ab)=\tau(ba)$ for all $a,b\in\mcA$. We denote by $T(\mcA)$ the set of tracial states on $\mcA$. It is a convex set and is compact and closed in the weak*-topology. We use $\partial_eT(\mcA)$ to denote the extreme points in $T(\mcA)$.  We say a tracial state $\tau$ on $\mcA$ \emph{factors through} another $C^*$-algebra $\mcB$ if there are some $*$-homomorphism $\varphi:\mcA\arr\mcB$ and some tracial state $\wtd{\tau}$ on $\mcB$ such that $\tau=\wtd{\tau}\circ \varphi$. If in addition, $\varphi$ is surjective, we say $\tau$ \emph{factors surjectively through} $\mcB$. 

For any $\tau\in T(\mcA)$, the \emph{kernel} of $\tau$, defined by 
\begin{align}
    \mcI_{\tau}:=\{a\in\mcA: \tau(a^*a)=0\},\label{eq:tracialideal}
\end{align}
is a two-sided closed ideal. So $\ang{a+\mcI_\tau,b+\mcI_{\tau}}:=\tau(a^*b),a,b\in\mcA$ defines an inner product on $\mcA\slash\mcI_\tau$. Completing $\mcA\slash\mcI_\tau$ with respect to this inner product gives the Hilbert space $L^2(\mcA,\tau)$. For any $a\in \mcA$, we use $\widehat{a}$ to denote the canonical image of $a$ in $L^2(\mcA,\tau)$. Let $\pi_\tau:\mcA\arr\msB(L^2(\mcA,\tau))$ be the $*$-representation of $\mcA$ given by 
\begin{align*}
    \pi_{\tau}(a)\widehat{b}:=\widehat{ab} \text{ for all }a,b\in\mcA,
\end{align*}
and let $\ket{\tau}:=\widehat{1}$. Then the triple $\big(L^2(\mcA,\tau), \pi_\tau,\ket{\tau} 
 \big)$ is a GNS representation for $\tau$. We refer to $\pi_\tau$ as the \emph{left regular representatin} of $\tau$, and refer to $\big(L^2(\mcA,\tau), \pi_\tau,\ket{\tau} 
 \big)$ as the \emph{standard GNS} for $\tau$. One can also define the \emph{right regular representation} $\pi_{\tau}^{\op}$, which is the $*$-representation from $\mcA^{\op}\arr \msB(L^2(\mcA,\tau))$ given by 
 \begin{align*}
    \pi_{\tau}^{\op}(a^{\op})\widehat{b}:=\widehat{ba} \text{ for all } a,b\in\mcA.
\end{align*}
Here $\mcA^{\op}$ denotes the \emph{opposite algebra} of $\mcA$. That is, $\mcA^{\op}=\mcA$ as a Banach space, but the multiplication on $\mcA^{\op}$ is reversed: if we use $a^{\op}$ to denote any $a\in\mcA$ in $\mcA^{\op}$, then $a^{\op}\cdot b^{\op}=(ba)^{\op}$. It is clear that $\pi_{\tau}(\mcA)'\subset \pi_{\tau}^{\op}(\mcA^{op})$. So the left and right regular representations commute.

The following characterization of extreme tracial states is well-known (see e.g., \cite[Theorem 6.7.3]{Dixmier})

\begin{lemma}\label{lemma:extremetrace}
    Suppose $\tau$ is a tracial state on a $C^*$-algebra $\mcA$. Then $\tau$ is an extreme point in $T(\mcA)$ if and only if the enveloping von Neumann algebra $\pi_\tau(\mcA)''$ is a factor.
\end{lemma}

Note that in the ``if" direction of the above lemma, the hypothesis that $\pi_\tau(\mcA)''$ is a factor does not necessarily mean $\pi_\tau$ is irreducible. In fact, if an extreme tracial state $\tau$ has an irreducible GNS representation $\pi_\tau$, then $\pi_\tau$ must be a one-dimensional representation. In other words, for any $\tau\in \partial_e T(\mcA)$, $\tau$ is a pure state if and only if $\tau$ is a character.

\subsection{Finite-dimensional tracial states}
When working with quantum correlations (i.e., correlations can be realized by finite-dimensional models), we are particularly interested in finite-dimensional tracial states and their convex structure. A tracial state is said to be finite-dimensional if its GNS representation is finite-dimensional. The set of finite-dimensional tracial states on $\mcA$ is denoted by $T_{\fin}(\mcA)$. The following characterization of finite-dimensional tracial states is well-known (see e.g., \cite{MR21}).
\begin{lemma}\label{lemma:fdtrace}
A tracial state $\tau$ on a $C^*$-algebra $\mcA$ is finite-dimensional if and only if one of the following equivalent statements holds.
    \begin{enumerate}[(a)]
        \item $\mcA\slash\mcI_\tau$ is finite-dimensional.
        \item $\tau$ factors through a finite-dimensional $C^*$-algebras.
        \item $\pi_\tau(\mcA)''$ is finite-dimensional. 
        \item The linear functional $\phi:\mcA^{\op}\otimes\mcA\arr\C$ defined by $\phi(a^{\op}\otimes b)=\tau(ab),a,b\in\mcA$ is a finite-dimensional state.
    \end{enumerate}
\end{lemma}

Later in \Cref{prop:fdtracial} we'll show that $T_{\fin}(\mcA)$ is a convex face of $T(\mcA)$. We first need to establish the following lemmas.

\begin{lemma}
    For any $C^*$-algebra $\mcA$, the set $T_{\fin}(\mcA)$ is convex.
\end{lemma}
\begin{proof}
   Let $\tau_1$ and $\tau_2$ be two finite-dimensional tracial states on a $C^*$-algebra $\mcA$. Then  there are finite-dimensional $C^*$-algebras $\mcB_1$ and $\mcB_2$, $*$-homomorphisms $\varphi_1:\mcA\arr\mcB_1$ and $\varphi_2:\mcA\arr\mcB_2$, and tracial states $\wtd{\tau_1}$ and $\wtd{\tau_2}$ on $\mcB_1$ and $\mcB_2$ respectively such that $\tau_1=\wtd{\tau}_1\circ \varphi_1$ and $\tau_2=\wtd{\tau}_1\circ \varphi_2$. For any $0<\lambda <1$, let $\wtd{\tau}$ be the tracial state on $\mcB_1\oplus\mcB_2$ defined by 
\begin{align*}
    \wtd{\tau}(\beta_1,\beta_2)=\lambda\wtd{\tau}_1(\beta_1)+(1-\lambda)\wtd{\tau}_2(\beta_2)
\end{align*}
for all $(\beta_1,\beta_2)\in \mcB_1\oplus\mcB_2$. Then $\lambda\tau_1+(1-\lambda)\tau_2=\wtd{\tau}\circ \varphi$ where $\varphi:=\lambda\varphi_1\oplus (1-\lambda)\varphi_2$ is a $*$-homomorphism from $\mcA\arr \mcB_1\oplus\mcB_2$. This means the tracial state $\lambda\tau_1+(1-\lambda)\tau_2$ factors through the finite-dimensional $C^*$-algebra $\mcB_1\oplus\mcB_2$, so $\lambda\tau_1+(1-\lambda)\tau_2\in T_{\fin}(\mcA)$. Since $\tau_1,\tau_2$, and $\lambda$ were arbitrary, we conclude that $T_{\fin}(\mcA)$ is a convex set.
\end{proof}

We use $\partial_e T_{\fin}(\mcA)$ to denote the extreme points in $T_{\fin}(\mcA)$.

\begin{lemma}\label{lemma:extremefdtrace}
    Let $\tau$ be a finite-dimensional tracial state on a $C^*$-algebra $\mcA$. Suppose $\tau=\lambda \tau_1+(1-\lambda)\tau_2$ for some $\tau_1,\tau_2\in T(\mcA)$ and $0<\lambda<1$. Then $\tau_1$ and $\tau_2$ must be finite-dimensional.
\end{lemma}
\begin{proof}
Since $\tau\in T_{\fin}(\mcA)$, the algebra $\mcA\slash\mcI_{\tau}$ is finite-dimensional. Observe that for any $a\in\mcA$, $\tau(a^*a)=0$ if and only if $\tau_1(a^*a)=\tau_2(a^*a)=0$. This means $\mcI_{\tau}=\mcI_{\tau_1}\cap\mcI_{\tau_2}$. Hence $\mcA\slash\mcI_{\tau_1}$ and $\mcA\slash\mcI_{\tau_2}$ are both $C^*$-subalgebras of $\mcA\slash\mcI_{\tau}$, so they are finite-dimensional. It follows that $\tau_1,\tau_2\in T_{\fin}(\mcA)$. 
\end{proof}

\begin{lemma}\label{lemma:convexhull}
    Let $\mcA$ be a $C^*$-algebra. Then $T_{\fin}(\mcA)$
    is the convex hull of tracial states on $\mcA$ that factor surjectively through full matrix algebras. 
\end{lemma}
\begin{proof}
    Any $\tau\in T_{\fin}(\mcA)$ drops to a tracial state $\wtd{\tau}$ on $\mcA\slash \mcI_\tau$ so that $\tau=\wtd{\tau}\circ q$, where $q:\mcA \arr \mcA\slash \mcI_\tau$ is the quotient map. Since $\mcA\slash \mcI_\tau$ is finite-dimensional, we can write 
    \begin{align*}
        \mcA\slash \mcI_\tau=\M_{d_1}\oplus \cdots \oplus \M_{d_k}
    \end{align*}
    for some $d_1,\ldots,d_k\geq 1$. For every $1\leq i\leq k$, let $\Pi_i:{\mcA\slash \mcI_\tau} \arr \M_{d_i}$ be the canonical projection. Then every $\varphi_i:=\Pi_i\circ q$ is a surjective $*$-homomorphism from $\mcA\arr \M_{d_i}$, and hence every $\tau_i:=\tr_{d_i}\circ \varphi_i$ is a tracial state on $\mcA$ that factors surjectively through the matrix algebra $\M_{d_i}$. Since each $\wtd{\tau}|_{\M_{d_i}}$ is a tracial linear functional on $\M_{d_i}$, we must have $\wtd{\tau}|_{\M_{d_i}}=\lambda_i\tr_{d_i}$ where $\lambda_i=\wtd{\tau}(\Id_{d_i})$. This implies $\tau=\lambda_1\tau_1+\cdots \lambda_k\tau_k$. We conclude that every $\tau\in T_{\fin}(\mcA)$ is a convex combination of tracial states that factor surjectively through full matrix algebras.
\end{proof}

The following proposition states that $T_{\fin}(\mcA)$ is a convex face of $T(\mcA)$ and $\partial_e T_{\fin}(\mcA)$ consists of tracial states that factor surjectively through \emph{full matrix algebras}.

\begin{proposition}\label{prop:fdtracial}
    Suppose $\mcA$ is a $C^*$-algebra such that $T_{\fin}(\mcA)$ is non-empty.
    \begin{enumerate}[(a)]
    \item $\partial_e T_{\fin}(\mcA)=\partial_e T(\mcA)\cap T_{\fin}(\mcA)$.
    \item A tracial state $\tau$ is in $\partial_e T_{\fin}(\mcA)$ if and only if it factors surjectively through a full matrix algebra, and
    \item $T_{\fin}(\mcA)=conv\big(\partial_e T_{\fin}(\mcA)\big)$.
    \end{enumerate}
\end{proposition}

\begin{proof}
    For part (a),  it is clear that $\partial_e T(\mcA)\cap T_{\fin}(\mcA)\subseteq \partial_e T_{\fin}(\mcA)$, so it only left to show the inverse inclusion. Now assume that there is a $\tau\in \partial_e T_{\fin}(\mcA)$ which is not in $\partial_e T(\mcA)$. Then there are distinct $\tau_1,\tau_2\in T(\mcA)$ and $0<\lambda<1$ such that $\tau=\lambda \tau_1 +(1-\lambda)\tau_2$. By \Cref{lemma:extremefdtrace}, $\tau_1$ and $\tau_2$ are finite-dimensional. This implies $\tau$ is not an extreme point of $T_{\fin}(\mcA)$, a contradiction. Hence $\partial_e T_{\fin}(\mcA) \subseteq \partial_e T(\mcA)$. We conclude that $\partial_e T_{\fin}(\mcA)=\partial_e T(\mcA)\cap T_{\fin}(\mcA)$.
    
    For part (b), let $\tau\in T_{\fin}(\mcA)$. Then $\tau$ has a convex combination $\lambda_1\tau_1+\cdots \lambda_k\tau_k$ where every $\tau_i$ factors surjectively through a full matrix algebra, as shown in \Cref{lemma:convexhull}. For the ``only if" direction, suppose $\tau$ is in $\partial_e T_{\fin}(\mcA)$. Then we must have $\tau_1=\cdots=\tau_k$, and hence $\tau$ factors surjectively through a full matrix algebra. For the ``if" direction, suppose $\tau$ factors surjectively through a full matrix algebra. Let $\big(L^2(\mcA,\tau),\pi_\tau,\ket{\tau}\big)$ be the standard GNS of $\tau$.  Then $\tau=\tr_d\circ\varphi$ for some surjective $*$-homomorphism $\varphi:\mcA\arr \M_d$, and hence
    \begin{align*}
        \mcI_\tau=\{a\in A:\tr_d\big(\varphi(a^*a)  \big)=0\}=\{a\in\mcA:\varphi(a)=0\}=\ker(\varphi).
    \end{align*}
This implies $\pi_\tau(\mcA)\cong{\mcA\slash \mcI_\tau}={\mcA\slash \ker(\varphi)}\cong \M_d$ is a full matrix algebra. So $\pi_\tau(\mcA)''$ is a factor, and hence $\tau\in \partial_e T(\mcA)$. Then by part (a), $\tau\in \partial_e T_{\fin}(\mcA)$. We conclude that a tracial state $\tau$ is in $\partial_e T_{\fin}(\mcA)$ if and only if $\tau$ factors surjectively through a full matrix algebra. 

Part (c) follows straightforward from part (b) and \Cref{lemma:convexhull}.
\end{proof}

From the proofs of \Cref{prop:fdtracial,lemma:convexhull}, we immediately see that:
\begin{corollary}\label{cor:quotient}
    Let $\mcA$ be a $C^*$-algebra. A tracial state $\tau$ on $\mcA$ is in $\partial_e T_{\fin}(\mcA)$ if and only if ${\mcA\slash\mcI_\tau}$ is a full matrix algebra. 
\end{corollary}

We now study the GNS construction of states in $\partial_e T_{\fin}(\mcA)$. Given a $\tau\in \partial_e T_{\fin}(\mcA)$, let $\big(L^2(\mcA,\tau),\pi_\tau,\ket{\tau}\big)$ be the standard GNS for $\tau$. Since $L^2(\mcA,\tau)$ is finite-dimensional, ${\mcA\slash\mcI_\tau}= L^2(\mcA,\tau)$ as a Hilbert space. Let $q:\mcA\arr \mcA\slash\mcI_\tau$ be the quotient map. Then $\tau$ drops to a tracial state $\wtd{\tau}$ on $L^2(\mcA,\tau)$ so that $\tau=\wtd{\tau}\circ q$. By \Cref{cor:quotient}, we can identify $L^2(\mcA,\tau)$ with\footnote{Here we think of $\M_d$ as a Hilbert space with inner product $\ang{A,B}=\tr_d(A^*B)$} $\M_d$ for some $d\geq 1$. Under this identification, $\wtd{\tau}=\tr_d$ and $\ket{\tau}=\frac{1}{\sqrt{d}}\sum_{i=1}^d e_{ii}$, where $\{e_{ij}:1\leq i,j\leq d\}$ is the standard basis for $\M_d$. While all GNS representations of $\tau$ are unitarily equivalent to the standard one, some of them have a more ``concrete" form. The following proposition illustrates that any extreme finite-dimensional tracial state has a GNS representation that employs a maximally entangled state.

\begin{proposition}\label{prop:tracialGNS}
    Let $\mcA$ be a $C^*$-algebra. A state $\tau$ on $\mcA$ is an extreme point in $T_{\fin}(\mcA)$ if and only if $\tau$ has a GNS representation $\big(\C^d\otimes\C^d,\Id\otimes \pi(\cdot),\ket{\varphi_d}\big)$, where $d\geq 1$, $\pi$ is an irreducible $*$-representation of $\mcA$ on $\C^d$, and $\ket{\varphi_d}:=\frac{1}{\sqrt{d}}\sum_{i=1}^d\ket{i}\otimes \ket{i}$.
\end{proposition}

\begin{proof}
    For the ``if" direction, observe that $\pi$ is a surjective $*$-homomorphism from $\mcA\arr M_d(\C)$ such that 
    \begin{align*}
        \tau(a)=\bra{\varphi_d}\Id\otimes\pi(a)\ket{\varphi_d}=\tr_d\big(\pi(a)\big)
    \end{align*}
    for all $a\in \mcA$. This means $\tau$ is a tracial state that factors surjectively through the full matrix algebra $M_d(\C)\cong\M_d$. By part (b) of \Cref{prop:fdtracial}, we have $\tau\in \partial_e T_{\fin}(\mcA)$.

    Now we prove the ``only if" direction. Suppose $\tau\in \partial_e T_{\fin}(\mcA)$. Then the standard GNS of $\tau$ has the form $(\M_d,\pi_\tau,\ket{\tau})$, where $\ket{\tau}=\frac{1}{\sqrt{d}}\sum_{i=1}^d e_{ii}$ and $\{e_{ij}:1\leq i,j\leq d\}$ is the standard basis for $\M_d$. Let $U:\M_d\arr \C^d\otimes \C^d$ be the unitary map sending $e_{ij}\mapsto \ket{j}\otimes\ket{i}$. Then $U\ket{\tau}=\frac{1}{\sqrt{d}}\sum_{i=1}^d\ket{i}\otimes \ket{i}=:\ket{\varphi_d}$. Let $\Phi:\M_d\arr M_d(\C)$ be the isomorphism sending $e_{ij}\mapsto \ket{i}\bra{j}$ for all $1\leq i,j\leq d$. Then  $\pi:=\Phi\circ\pi_\tau$ is a surjective $*$-homomorphism from $\mcA\arr M_n(\C)$, and hence $\pi$ is an irreducible $*$-representation. For  any $a\in \mcA$, $\pi_\tau(a)$ can be written uniquely as $\sum_{i,j=1}^da_{ij}e_{ij}$ for some $a_{ij}\in\C$, so
    \begin{align*}
U\pi_\tau(a)U^*\ket{k,\ell}&=U\sum_{i,j=1}^da_{ij}e_{ij}e_{\ell k}=U\sum_{i=1}^da_{i\ell}e_{i k}=\sum_{i=1}^da_{i\ell}\ket{k,i}=\ket{k}\otimes\big(\sum_{i=1}^da_{i\ell}\ket{i} \big)\\
        &= \ket{k}\otimes\big(\sum_{i=1}^da_{i\ell}\ket{i}\braket{\ell|\ell} \big)=\ket{k}\otimes\big(\sum_{i,j=1}^da_{ij}\ket{i}\braket{j|\ell} \big)\\
        &=\ket{k}\otimes\Big(\Phi\big(\pi_\tau(a)\big)\ket{\ell}\Big)  =\ket{k}\otimes\big(\pi(a)\ket{\ell}\big)
    \end{align*}
    for all $\ket{k,\ell}\in\C^d\otimes \C^d$. This implies $U\pi_\tau(a)U^*=\Id\otimes\pi(a)$ for all $a\in \mcA$. It follows that $\big(\C^d\otimes\C^d,\Id\otimes\pi(\cdot),\ket{\varphi_d}\big)$ is unitarily equivalent to the standard GNS of $\tau$, and hence itself is a GNS representation of $\tau$.
\end{proof}

%\begin{remark}\label{rmk:right}
    The unitary operator $U$ in the proof of \Cref{prop:tracialGNS} is known as the transformation of the operator-vector correspondence. From the proof, we see that the left regular representation $\pi_\tau$ satisfies $\pi_{\tau}(a)=U^*\big( \Id\otimes\pi(a)\big)U$ for all $a\in\mcA$. Consequently, the right regular representation $\pi_\tau^{\op}$ is given by $\pi_\tau^{\op}(a)=U^*\big(\pi(a)^{T}\otimes\Id\big)U$.
%\end{remark}

For every $\tau\in \partial_e T_{\fin}(\mcA)$, the irreducible representation $\pi$ in the statement of \Cref{prop:tracialGNS} is unique up to unitary equivalence. This gives a one-to-one correspondence between extreme finite-dimensional tracial states on $\mcA$ and finite-dimensional irreducible representations of $\mcA$.

\begin{corollary}\label{cor:extremetrace-irrep}
    For any $C^*$-algebra $\mcA$, let $\Irr_{\fin}(\mcA)$ be the set of inequivalent finite-dimensional irreducible $*$-representations of $\mcA$. There is a bijective correspondence between $\Irr_{\fin}(\mcA) $ and $ \partial_e T_{\fin}(\mcA)$, sending a representation $\pi\in \Irr_{\fin}(\mcA)$  to the tracial state $\tau$ given by $\tau(a)=\tr_d\big(\pi(a)\big),a\in\mcA$, where $d$ is the dimension of $\pi$.
\end{corollary}

In particular, a $C^*$-algebra $\mcA$ has a unique finite-dimensional irreducible $*$-representation if and only if $\Irr_{\fin}(\mcA)$ has a unique extreme point. So \Cref{cor:extremetrace-irrep} further implies that:
\begin{corollary}\label{coro:uniqueirreptrace}
    A $C^*$-algebra $\mcA$ has a unique finite-dimensional tracial state $\tau$ if and only if $\mcA$  has a unique finite-dimensional irreducible $*$-representation $\pi$. In this case, $\tau(a)=\tr_d\big(\pi(a)\big),a\in\mcA$, where $d$ is the dimension of $\pi$. 
\end{corollary}

If a closed two-sided ideal $\mcJ$ is contained in $\mcI_{\tau}$ for a tracial state $\tau$ on $\mcA$, then it is clear that $\tau$ drops to a tracial state on $\mcA\slash \mcJ$. The following lemma asserts that the finite dimensionality of $\tau$ will also be preserved. 

\begin{proposition}\label{prop:tracequotientfd}
    Let $\mcA$ be a $C^*$-algebra. Suppose $\mcJ$ is a closed two-sided ideal of $\mcA$, and let $q:\mcA\arr \mcA\slash\mcJ$ be the quotient map. There is a bijective correspondence between tracial states on $\mcA\slash\mcJ$ and tracial states on $\mcA$ whose kernel contains $\mcJ$, sending a $\tau\in T(\mcA\slash \mcJ)$ to the tracial state $\wtd{\tau}:=\tau\circ q$. Moreover, $\tau\in T(\mcA\slash \mcJ)$ is finite-dimensional if and only if the corresponding $\wtd{\tau}$ is finite-dimensional.
\end{proposition}

\begin{proof}
   Observe that $L^2(\mcA,\wtd{\tau})\cong L^2(\mcA\slash \mcJ,\tau)$. So the lemma follows.
\end{proof}

\subsection{Amenable tracial states}
In the study of robust self-testing, we need to work with weak*-limit of finite-dimensional tracial states. This leads to the study of amenable tracial states.

\begin{definition}
    Let $\mcA$ be a $C^*$-algebra, and fix a faithful representation $\mcA\subset\msB(\mcH)$. We say a state $\tau$ on $\mcA$ is an \textbf{amenable tracial state} if there is a state $\phi$ on $\msB(\mcH)$ such that $\phi|_{\mcA}=\tau$ and $\phi(uTu^*)=\phi(T)$ for any $T\in\msB(\mcH)$ and unitary $u\in \mcA$.     
\end{definition}

    Note that $\phi|_{\mcA}$ is indeed a tracial state. This is because $\phi(uv)=\phi\big(u(vu)u^*\big)=\phi(vu)$ for all unitaries $u,v\in\mcA$ and unitaries in $\mcA$ span the entire $\mcA$. Another subtle point in this definition is that the amenability of $\tau$ is independent of the choice of embedding $\mcA\subset \msB(\mcH)$ (see \cite[Proposition 6.2.2]{BrownOzawa} for a proof).

The following characterization of amenable tracial states is from \cite[Theorem 6.2.7]{BrownOzawa}.
\begin{lemma}\label{lemma:amenable}
    A tracial state $\tau$ on a $C^*$-algebra $\mcA$ is amenable if and only if one of the following equivalent conditions holds.
    \begin{enumerate}[(a)]
        \item There exists a sequence of ucp maps $\varphi_n:\mcA\arr\M_{d_n},n\in\N$ such that $\tau(a)=\lim\limits_{n\arr\infty} \tr_{d_n}\big(\varphi_n(a)\big)$ and $\lim\limits_{n\arr \infty}\norm{\varphi_n(ab)-\varphi_n(a)\varphi_n(b)}_{hs}=0$ for all $a,b\in\mcA$.  
        \item The linear functional $\phi:\mcA^{\op}\otimes\mcA\arr\C$ defined by $\phi(a^{\op}\otimes b):=\tau(ab),a,b\in\mcA$ is a min-continuous state.
        \item The $*$-homomorphism $\pi:\mcA^{\op}\otimes\mcA\arr \msB(L^2(\mcA,\tau))$ defined by $\pi(a^{\op}\otimes b)=\pi_{\tau}^{\op}(a^{\op})\pi_\tau(b)$, $a,b\in\mcA$ is min-continuous.
        \item For any faithful representation $\mcA\subset \msB(\mcH)$ there is a ucp map $\Phi:\mcA(\mcH)\arr \pi_{\tau}(\mcA)''$ such that $\Phi(a)=\pi_\tau(a)$ for all $a\in\mcA$.
    \end{enumerate}
\end{lemma}

For any two-sided ideal $\mcJ$ of $\mcA$, in \Cref{prop:tracequotientfd}, we have seen the correspondence between tracial states on $\mcA\slash\mcJ$ and tracial states on $\mcA$ whose kernel contains $\mcJ$. The following proposition asserts that amenability is preserved in one direction of this correspondence.

\begin{proposition}\label{prop:tracequotientamenable}
Let $\mcA$ be a $C^*$-algebra. Suppose $\mcJ$ is a closed two-sided ideal of $\mcA$, and let $q:\mcA\arr \mcA\slash\mcJ$ be the quotient map. If a tracial state $\tau$ on $\mcA\slash\mcJ$ is amenable, then the tracial state $\wtd{\tau}:=\tau\circ q$ on $\mcA$ is amenable.
\end{proposition}

\begin{proof}
    Suppose $\tau$ is an amenable tracial state on $\mcA\slash\mcJ$, and let $\wtd{\tau}:=\tau\circ q$. Then there is a sequence of ucp maps $\varphi_n:\mcA\slash\mcJ\arr \M_{d_n},n\in\N$ such that $\tau(a)=\lim\limits_{n\arr\infty} \tr_{d_n}\big(\varphi_n(a)\big)$ and $\lim\limits_{n\arr \infty}\norm{\varphi_n(ab)-\varphi_n(a)\varphi_n(b)}_{hs}=0$ for all $a,b\in\mcA\slash\mcJ$. It follows that $\wtd{\varphi}_n:=\varphi_n\circ q:\mcA\arr\M_{d_n},n\in\N$ is a sequence of ucp maps such that $\lim\limits_{n\arr\infty} \tr_{d_n}\big(\wtd{\varphi}_n(a)\big)=\tau(q(a))=\wtd{\tau}(a)$ and $\lim\limits_{n\arr \infty}\norm{\wtd{\varphi}_n(ab)-\wtd{\varphi}_n(a)\wtd{\varphi}_n(b)}_{hs}=0$ for all $a,b\in\mcA$. So $\wtd{\tau}$ is amenable.
\end{proof}

\begin{remark}\label{rmk:tracequotientamenable}
    Unlike \Cref{prop:tracequotientfd}, the converse of \Cref{prop:tracequotientamenable} may not hold in general: there are examples where an amenable tracial state $\wtd{\tau}$ on $\mcA$ drops to a non-amenable tracial state $\tau$ on $\mcA\slash\mcJ$. However, as shown in \cite[Proposition 6.3.5]{BrownOzawa}, if $0\arr \mcJ\arr \mcA \arr \mcA\slash\mcJ\arr 0$ is a locally split extension, then any amenable tracial state on $\mcA$ whose kernel contains $\mcJ$ always drops to an amenable tracial state on $\mcA\slash\mcJ$.
\end{remark}

\section{A tracial-state characterization of self-testing}\label{sec:traceselftest}
%\subsection{General results}
In this section, we relate self-testing with tracial states on $C^*$-algebras. We first discuss the sufficient conditions under which optimal strategies of a nonlocal game can be characterized by tracial states on Bob's algebra $\PVM^{Y,B}$ only.  

\begin{definition}\label{def:pair}
    Let $\Gamma:=\{\gamma^x_a:x\in X,a\in A\}$ be a set of self-adjoint $*$-polynomials in $\PVM^{Y,B}$, and let $\mcR$ be a set of $*$-polynomials in $\PVM^{Y,B}$. We say $(\Gamma,\mcR)$ is a determining pair for a nonlocal game $\mcG=(X,Y,A,B,\mu, V)$ if the following conditions hold.
    \begin{enumerate}[(i)]
        \item For every $x\in X$, $\{q(\gamma^x_a)\}_{a\in A}$ is a PVM in the quotient $C^*(\mcG):=\PVM^{Y,B}\slash \ang{\mcR}$, where $\ang{\mcR}$ is the closed two-sided ideal generated by $\mcR$ and $q:\PVM^{Y,B}\arr C^*(\mcG)$ is the quotient map.
        \item A state\footnote{Since we assume all states are bounded in the sense of \Cref{eq:bounded}, every state $f$ on the algebraic tensor $\PVM^{X,A}\otimes\PVM^{Y,B}$ is max-continuous.} $f$ on $\PVM^{X,A}\otimes\PVM^{Y,B}$ is optimal for $\mcG$ if and only if
        \begin{enumerate}
            \item $f\big((m^x_a\otimes 1-1\otimes \gamma^x_a)^2\big)=0$ for all $x\in X, a\in A$, and
            \item the linear functional $\tau:=f|_{1\otimes \PVM^{Y,B}}$ is a tracial state on $\PVM^{Y,B}$ that satisfies $\tau(r^*r)=0$ for all $r\in\mcR$.
        \end{enumerate}
    \end{enumerate}
\end{definition}

In Condition (i),  $\{q(\gamma^x_a):a\in A\}$ is a PVM in $C^*(\mcG)$ means that 
\begin{itemize}
    \item $q\big(\sum_{a\in A}\gamma^x_a\big)=1$,
    \item $q\big((\gamma^x_a)^*\big)=q(\gamma^x_a)=q\big((\gamma^x_a)^2\big)$, and 
    \item $q\big(\gamma^x_a\gamma^x_{a'}\big)=0$
\end{itemize}
for all $x\in X$ and $a\neq a'$. We usually express $\gamma^x_a$ as a $*$-polynomial over $n^y_b$'s. Since the quotient map sends each generator $n^y_b$ for $\PVM^{Y,B}$ to the generator $n^y_b$ for $C^*(\mcG)$, we often just say $\{\gamma^x_a:a\in A\}$ is a PVM in $C^*(\mcG)$. We refer to the $C^*$-algebra $C^*(\mcG)$ defined in Condition (i) of \Cref{def:pair} as the \emph{associated game algebra} of $\mcG$.  

In Condition (ii), although we are working with max-continuous states on $\PVM^{X,A}\otimes\PVM^{Y,B}$, the optimal value of $\mcG$ still refers to the quantum value $w_q(\mcG)$. Here $\tau=f|_{1\otimes \PVM^{Y,B}}$ means $\tau(\alpha)=f(1\otimes\alpha)$ for all $\alpha\in \PVM^{Y,B}$. 

It is sometimes convenient to work with unitary observables $\{a_x:x\in X\},\{b_y:y\in Y\}$ as generators for $\PVM^{X,A}\otimes\PVM^{Y,B}$. In this case, we write $\Gamma=\{\gamma_x:x\in X\}$ where every $\gamma_x$ is a $*$-polynomials over $b_y$'s. Condition (i) becomes every $q(\gamma_x)$ is unitary of order $\abs{B}$ in $C^*(\mcG)$, and part (a) of Condition (ii) is replaced with $f\big((a_x\otimes\Id-\Id\otimes \gamma_x)^* (a_x\otimes\Id-\Id\otimes \gamma_x) \big)=0$ for all $x\in X$.

The above definition may appear technical, but it provides a natural and Hilbert-space-free framework for many nonlocal games of interest, including synchronous games, XOR games, BCS games, and mirror games\footnote{We remark that we do not know if the class of imitation games introduced in \cite{LMP20} fits into this framework. The reason is that we are not aware of any Hilber-space-free way to encode the relations between Alice and Bob's measurements which are employed in perfect strategies of imitation games}~\cite{LMP20}. Here we take a moment to provide more intuitions. The existence of polynomials $\Gamma$ and relations $\mcR$ in \Cref{def:pair} states that in any optimal strategy, the action of Alice's measurement $m^x_a$ is completely determined by the measurement $\gamma^x_a$ on Bob's side, and Bob's measurements must satisfy certain algebraic relations $\mcR$.

For instance, if $S=\big(\{M^x_a\},\{N^y_b\},\ket{\psi})$ is a perfect commuting-operator strategy for a synchronous game $\mcG=(X,A,V)$, then by \cite{PSSTW16}, 
\begin{align}
    M^x_a\ket{\psi}&=N^x_a\ket{\psi} \text{ for all } x\in X,a\in A,
\text{ and}\label{syngamma}\\
    \bra{\psi}N^x_aN^y_b\ket{\psi}&=0 \text{ whenever }V(a,b|x,y)=0.\label{synR}
\end{align}
\Cref{syngamma} implies that the state $f$ induced by $S$ satisfies $f\big((m^x_a\otimes 1-1\otimes n^x_a)^2\big)=0$ for all $x$ and $a$, so we can just take $\gamma^x_a=n^x_a$. This gives the set of $*$-polynomials $\Gamma$. \Cref{synR} suggests the the set of relations $\mcR=\{n^x_an^y_b:V(a,b|x,y)=0\}$. The game algebra $C^*(\mcG)$ is the quotient $\PVM^{Y,B}\slash\ang{\mcR}$. Bob's measurements $N^y_b$'s together with the vector state $\ket{\psi}$ define a tracial state on $\PVM^{Y,B}$ that respects all the relations in $\mcR$. As we'll show in \Cref{subsec:syn}, every synchronous game with commuting-operator value $1$ has such a determining pair $(\Gamma,\mcR)$.  

In \Cref{subsec:XOR}, we also show that all XOR games have determining pairs. Take the CHSH game as an example. If a commuting-operator strategy $S=\big(\{A_0,A_1\},\{B_0,B_1\},\ket{\psi}\big)$ is optimal for CHSH, then 
\begin{align*}
    A_0\ket{\psi}=\tfrac{B_0+B_1}{\sqrt{2}}\ket{\psi},\text{ and }
    A_1\ket{\psi}=\tfrac{B_0-B_1}{\sqrt{2}}\ket{\psi}.
\end{align*}
 Here we work with binary observables $a_0,a_1$ and $b_0,b_1$ as generators for $\PVM^{\Z_2,\Z_2}\otimes_{min}\PVM^{\Z_2,\Z_2}$. The above equations imply that the state $f$ induced by $S$ satisfies
\begin{equation*}
    f\big((a_0\otimes\Id-\Id\otimes \tfrac{b_0+b_1}{\sqrt{2}})^2  \big)= f\big((a_1\otimes\Id-\Id\otimes \tfrac{b_0-b_1}{\sqrt{2}})^2  \big)=0.
\end{equation*}
So we take $\Gamma=\{\tfrac{b_0+b_1}{\sqrt{2}},\tfrac{b_0-b_1}{\sqrt{2}}\}$. We also require $\tfrac{b_0+b_1}{\sqrt{2}}$ and $\tfrac{b_0-b_1}{\sqrt{2}}$ are binary observables as in Condition (i) of \Cref{def:pair}, which is equivalent to require $b_0$ and $b_1$ anticommute. This gives the relations set $\mcR=\{b_0b_1+b_1b_0\}$. The game algebra associated with CHSH is $\PVM^{\Z_2,\Z_2}\slash\ang{\mcR}\cong Cl_2$, the Clifford algebra of rank 2.

In the following, we present a tracial-state characterization of self-testing for nonlocal games that have determining pairs. For such a nonlocal game $\mcG$, we show that its self-testing property of $\mcG$ can be characterized by tracial states on the associated game algebra $C^*(\mcG)$. Then we illustrate that every synchronous game or XOR game has a determining pair. So their self-testing properties can be characterized by tracial states on the respect game algebras.

\subsection{General results for correlations}
We first prove the following general statement for extreme quantum correlations.
\begin{theorem}\label{thm:correlationtrace}
    Suppose a nonlocal game $\mcG$ has a determining pair $(\Gamma,\mcR)$. Let $C^*(\mcG)$ be the associated game algebra, and let $q:\PVM^{Y,B}\arr C^*(\mcG)$ be the quotient map.
    If an extreme quantum correlation $p$ is optimal for $\mcG$, then the following statements are equivalent.

\begin{enumerate}[(a)]
    \item $p$ is a self-test for all quantum models.
    \item $p$ is an abstract state self-test for finite-dimensional states on $\PVM^{X,A}\otimes_{min}\PVM^{Y,B}$.
    \item There is a unique finite-dimensional tracial state $\tau$ on $\PVM^{X,A}$ satisfying
    \begin{align}
        \tau(r^*r)&=0 \text{ for all }r\in\mcR, \text{ and}\label{eq:tau1}\\
        \tau(\gamma^x_an^y_b)&=p(a,b|x,y) \text{ for all }a,b,x,y.\label{eq:tau2}
    \end{align}

    \item There is a unique finite-dimensional irreducible $*$-representation $\pi:\PVM^{Y,B}\arr M_d(\C)$ such that the linear functional $\tau:=\tr_d\circ \pi $
    satisfies \Cref{eq:tau1,eq:tau2}. 
\end{enumerate}

\end{theorem}

The proof is established on a sequence of technical lemmas that characterize the nonlocal games that have determining pairs. We outline some proof ideas. Suppose $p\in \partial_eC_q$ is optimal for a nonlocal game $\mcG$ with determining pair $(\Gamma,\mcR)$. We will show that $p$ must have a full-rank projective model, so $p$ has a unique finite-dimensional state on the product of PVM algebras if and only if it has a unique finite-dimensional state on the product of POVM algebras. The latter condition is equivalent to part (a), as shown in \Cref{thm:old}. So $(a)\Leftrightarrow(b)$. For $(b)\Leftrightarrow(c)$, we demonstrate a one-to-one correspondence between states on $\PVM^{X,A}\otimes_{min}\PVM^{Y,B}$ for $p$ and tracial states on $\PVM^{Y,B}$ satisfying \Cref{eq:tau1,eq:tau2}. For $(c)\Leftrightarrow(d)$, we use the one-to-one correspondence between finite-dimensional irreducible representations and finite-dimensional extreme tracial states that was established in \Cref{cor:extremetrace-irrep}.

For notational convenience, we denote by $T^{(p)}(\Gamma,\mcR)$ the set of tracial states on $\PVM^{Y,B}$ that satisfy \Cref{eq:tau1,eq:tau2}. We use $T_{\fin}^{(p)}(\Gamma,\mcR)$ to denote the set of finite-dimensional tracial states in $T^{(p)}(\Gamma,\mcR)$. From \Cref{eq:tau1,eq:tau2}, it is clear that both $T^{(p)}(\Gamma,\mcR)$ and $T_{\fin}^{(p)}(\Gamma,\mcR)$ are convex sets. We denote by $\partial_eT^{(p)}(\Gamma,\mcR)$ and $\partial_e T_{\fin}^{(p)}(\Gamma,\mcR)$ the extreme points in $T^{(p)}(\Gamma,\mcR)$ and $T_{\fin}^{(p)}(\Gamma,\mcR)$ respectively. Later, in \Cref{lemma:extremeMES}, we will see that $T_{\fin}^{(p)}(\Gamma,\mcR)$ is a convex face of $T_{\fin}(\PVM^{Y,B})$ whenever $p$ is an extreme point in $C_q$. 

Recall that a state $f$ on the $*$-algebra $\PVM^{X,A}\otimes\PVM^{Y,B}$ is said to be optimal for a nonlocal game $\mcG=(X,Y,A,B,\mu,V)$ if $f(\Phi_\mcG)=w_q(\mcG)$, where $\Phi_\mcG$ is the game polynomial.

\begin{lemma}\label{lemma:extendtrace}
    Suppose a nonlocal game $\mcG=(X,Y,A,B,\mu,V)$ has a determining pair $(\Gamma,\mcR)$. If a tracial state $\tau$ on $\PVM^{Y,B}$ satisfies $\tau(r^*r)=0$ for all $r\in\mcR$, then $\tau$ extends uniquely to a state $f$ on $\PVM^{X,A}\otimes\PVM^{Y,B}$ which is optimal for $\mcG$ and satisfies
    \begin{align}
        f(m^{x_1}_{a_1}\cdots m^{x_k}_{a_k}\otimes n^{y_1}_{b_1}\cdots n^{y_\ell}_{b_\ell})=\tau(\gamma^{x_k}_{a_k}\cdots \gamma^{x_1}_{a_1} n^{y_1}_{b_1}\cdots n^{y_\ell}_{b_\ell})\label{eq:extendtrace}
    \end{align}
    for all $k,\ell\in\N$ and $x_1,\ldots,x_k\in X$, $a_1,\ldots,a_k\in A$, $y_1,\ldots,y_\ell\in Y$, $b_1,\ldots,b_\ell\in B$.
    If in addition, $\tau$ is finite-dimensional (resp. amenable), then $f$ is finite-dimensional (resp. min-continuous).
\end{lemma}
\begin{proof}
    Suppose $\tau$ is a tracial state on $\PVM^{Y,B}$ such that $\tau(r^*r)=0$ for all $r\in\mcR$. Then $\tau$ drops to a tracial state $\wtd{\tau}$ on the quotient $C^*(\mcG)$ such that $\tau=\wtd{\tau}\circ q$, where $q:\PVM^{Y,B}\arr C^*(\mcG)$ is the quotient map. Let $(\mcH_{\wtd{\tau}},\pi_{\wtd{\tau}},\ket{\wtd{\tau}})$ be the standard GNS for $\wtd{\tau}$. Here $\pi_{\wtd{\tau}}$ is the left regular representation for $\wtd{\tau}$. Then $\pi_\tau:=\pi_{\wtd{\tau}}\circ q$ is the left regular representation for $\tau$, and $(\mcH_{\wtd{\tau}},\pi_{\tau},\ket{\wtd{\tau}})$ is the standard GNS for $\tau$.
    
    Note that for every $x\in X$, $\{q(\gamma^x_a):a\in A\}$ is a PVM in the quotient $C^*(\mcG)$. This means every $\{\pi^{\op}_{\wtd{\tau}}\big(q(\gamma^x_a)\big):a\in A  \}$ is a PVM on $\mcH_{\wtd{\tau}}$. It follows that every $\{\pi_\tau^{\op}(\gamma^x_a):a\in A\}$ is a PVM on $\mcH_{\wtd{\tau}}$. Hence there is a $*$-representation $\pi:\PVM^{X,A}\otimes\PVM^{Y,B}\arr \msB(\mcH_{\wtd{\tau}})$ sending $m^x_a\otimes 1\mapsto \pi^{op}_{\tau}(\gamma^x_a)$ and $1\otimes n^y_b\mapsto \pi_{\tau}(n^y_b)$ for all $a,b,x,y$. Let $f$ be the state on $\PVM^{X,A}\otimes\PVM^{Y,B}$ defined by $f(\alpha):=\bra{\wtd{\tau}}\pi(\alpha)\ket{\wtd{\tau}},\alpha\in \PVM^{X,A}\otimes\PVM^{Y,B}$. Then 
    \begin{align*}
        f(m^{x_1}_{a_1}\cdots m^{x_k}_{a_k}\otimes n^{y_1}_{b_1}\cdots n^{y_\ell}_{b_\ell})=& \bra{\wtd{\tau}}\pi_\tau^{\op}(\gamma^{x_k}_{a_k}\cdots \gamma^{x_1}_{a_1})\pi_\tau(n^{y_1}_{b_1}\cdots n^{y_\ell}_{b_\ell})\ket{\wtd{\tau}} \\
        &=\bra{\wtd{\tau}}\pi_\tau(\gamma^{x_k}_{a_k}\cdots \gamma^{x_1}_{a_1}n^{y_1}_{b_1}\cdots n^{y_\ell}_{b_\ell})\ket{\wtd{\tau}} \\
        &=\tau(\gamma^{x_k}_{a_k}\cdots \gamma^{x_1}_{a_1} n^{y_1}_{b_1}\cdots n^{y_\ell}_{b_\ell})
    \end{align*}
    for all $k,\ell\in\N$ and $x_1,\ldots,x_k\in X$, $a_1,\ldots,a_k\in A$, $y_1,\ldots,y_\ell\in Y$, $b_1,\ldots,b_\ell\in B$. This proves \Cref{eq:extendtrace}. Meanwhile, 
    \begin{align*}
        \pi(m^x_a\otimes 1)\ket{\wtd{\tau}}=\pi^{\op}_{\tau}(\gamma^x_a)\ket{\wtd{\tau}}=\pi_{\tau}(\gamma^x_a)\ket{\wtd{\tau}}=\pi(1\otimes \gamma^x_a)\ket{\wtd{\tau}}.
    \end{align*}
    This implies $f\big((m^x_a\otimes 1-1\otimes\gamma^x_a)^2 \big)=0$. So $f$ satisfies part (a) of condition (ii) in \Cref{def:pair}. Part (b) of condition (ii) follows from \Cref{eq:extendtrace}. We conclude that $f$ is optimal for $\mcG$. Since monomials of the form $m^{x_1}_{a_1}\cdots m^{x_k}_{a_k}\otimes n^{y_1}_{b_1}\cdots n^{y_\ell}_{b_\ell}$ span a dense subset of $\PVM^{X,A}\otimes\PVM^{Y,B}$. \Cref{eq:extendtrace} implies that $f$ is uniquely determined by $\tau$.

    If $\tau$ is amenable, then $\pi^{\op}_{\tau}\times\pi_{\tau}$ is min-continuous, and hence $f$ is min-continuous. If $\tau$ is finite-dimensional, then $\mcH_{\wtd{\tau}}$ is finite-dimensional, and hence $f$ is finite-dimensional.
\end{proof}

\begin{proposition}\label{prop:extendtrace}
    Suppose $\mcG$ has a determining pair $(\Gamma,\mcR)$ and $p\in C_{qc}$ is optimal for $\mcG$. There is a bijective correspondence between states on $\PVM^{X,A}\otimes\PVM^{Y,B}$ for $p$ and tracial states in $T^{(p)}(\Gamma,\mcR)$, sending an $f$ to the tracial state $\tau:=f|_{1\otimes\PVM^{Y,B}}$. Moreover, 
    \begin{enumerate}[(a)]
        \item $f$ is finite-dimensional if and only if the corresponding $\tau$ is finite-dimensional, and 
        \item $f$ is min-continuous if the corresponding $\tau$ is amenable.
    \end{enumerate}
\end{proposition}

\begin{proof}
    Let $f$ be a state on $\PVM^{X,A}\otimes\PVM^{Y,B}$ for $p$. Since $p$ is optimal for $\mcG$, $f$ is an optimal state for $\mcG$. So 
    \begin{align}
        f\big((m^x_a\otimes 1-1\otimes \gamma^x_a)^2  \big)=0\label{eq:f}
    \end{align}
    for all $x\in X,a\in A$, and
    $\tau:=f|_{1\otimes\PVM^{Y,B}}$ is a tracial state on $\PVM^{Y,B}$ satisfying $\tau(r^*r)=0$ for all $r\in\mcR$. Let $(\mcH,\pi,\ket{\psi})$ be a GNS representation for $f$. \Cref{eq:f} implies
    \begin{align*}
        \pi(m^x_a\otimes 1)\ket{\psi}=\pi(1\otimes \gamma^x_a)\ket{\psi}
    \end{align*}
    for all $x\in X,a\in A$. It follows that
    \begin{align*}
    \tau(\gamma^x_a n^y_b)=\bra{\psi}\pi(1\otimes\gamma^x_an^y_b)\ket{\psi}=\bra{\psi}\pi(m^x_a\otimes n^y_b)\ket{\psi}=f(m^x_a\otimes n^y_b)=p(a,b|x,y)
\end{align*}
for all $a,b,x,y$.
So $\tau$ is a tracial state in $T^{(p)}(\Gamma,\mcR)$. Moreover, if $f$ is finite-dimensional, then $\tau$ factors through the finite-dimensional $C^*$-algebra $\pi(1\otimes\PVM^{Y,B})$, and hence $\tau$ is finite-dimensional. The rest of the proof follows straightforwardly from \Cref{lemma:extendtrace}.
\end{proof}

\begin{lemma}\label{lemma:fullrankprojective}
    Suppose $\mcG=(X,Y,A,B,\mu,V)$ has a determining pair $(\Gamma,\mcR)$. If a quantum correlation $p\in C_q(X,Y,A,B)$ is optimal for $\mcG$, then $p$ has a full-rank projective model.
\end{lemma}

\begin{proof}
    Let $f$ be a finite-dimensional state on $\PVM^{X,A}\otimes_{min}\PVM^{Y,B}$ for $p$, and let $\tau$ be the corresponding tracial state in $T_{\fin}^{(p)}(\Gamma,\mcR)$ given by \Cref{prop:extendtrace}. Since $\tau$ is a finite-dimensional tracial state on $\PVM^{Y,B}$, by \Cref{prop:fdtracial}, $\tau$ has a convex combination $\tau=\lambda_1\tau_1+\cdots+\lambda_k\tau_k$ of tracial states $\tau_1,\ldots,\tau_k\in\partial_e T_{\fin}(\PVM^{X,A})$. By \Cref{prop:tracialGNS}, every $\tau_i$ has a GNS representation $\big(\C^{d_i}\otimes \C^{d_i},\Id\otimes\pi_i(\cdot),\ket{\varphi_{d_i}}  \big)$, where $d_i\geq 1$, $\pi_i:\mcA\arr M_{d_i}(\C)$ is an irreducible $*$-representation, and $\ket{\varphi_{d_i}}=\frac{1}{\sqrt{d_i}}\sum_{j=1}^{d_i}\ket{j}\otimes\ket{j}$. Note that $\lambda_1,\ldots,\lambda_k>0$, so $\tau(r^*r)=0$ for all $r\in\mcR$ implies $\tau_i(r^*r)=0$ for all $r\in\mcR$ and $1\leq i\leq k$. Since every $\{\gamma^x_a:a\in A\}$ is a PVM in the quotient $\PVM^{Y,B}\slash\ang{\mcR}$, every $\{\pi_i(\gamma^x_a)^T:a\in A\}$ is a PVM on $\C^{d_i}$. Let $M^x_a:=\bigoplus_{i=1}^k\pi_i(\gamma^x_a)^T,N^y_b:=\bigoplus_{i=1}^k\pi_i(n^y_b)$ for all $a,b,x,y$, and let $\ket{\psi}=\bigoplus_{i=1}^k\sqrt{\lambda_i}\ket{\varphi_{d_i}}$. Then
    \begin{align*}
        S:=\left(\bigoplus_{i=1}^k\C^{d_i},\bigoplus_{i=1}^k\C^{d_i},\{M^x_a\},\{N^y_b\},\ket{\psi}\right)
    \end{align*}
    is a full-rank projective quantum model such that
    \begin{align*}
        p_S(a,b|x,y)&=\sum_{i=1}^k \lambda_i \bra{\varphi_{d_i}}\pi_i(\gamma^x_a)^T\otimes \pi_i(n^y_b)\ket{\varphi_{d_i}}=\sum_{i=1}^k\lambda_i \bra{\varphi_{d_i}}\Id\otimes \pi_i(\gamma^x_a)\pi_i(n^y_b)\ket{\varphi_{d_i}}\\
        &= \sum_{i=1}^k\lambda_i \bra{\varphi_{d_i}}\Id\otimes \pi_i(\gamma^x_an^y_b)\ket{\varphi_{d_i}}= \sum_{i=1}^k\lambda_i \tau_i(\gamma^x_an^y_b)=\tau(\gamma^x_an^y_b)=p(a,b|x,y)
    \end{align*}
for all $a,b,x,y$. This completes the proof.
\end{proof}

In the proof of \Cref{lemma:fullrankprojective}, we see that the model $S$ employs a direct sum of maximally entangled states. When $p$ is an extreme point in $C_q$, we can strengthen \Cref{lemma:fullrankprojective}: $p$ has a projective model which employs a single maximally entangled state. The proof is based on the observation that $T_{\fin}^{(p)}(\Gamma,\mcR)$ is a convex face of $T_{\fin}(\PVM^{Y,B})$ whenever $p$ is extreme.

\begin{lemma}\label{lemma:extremeMES}
    Suppose $\mcG$ has a determine pair $(\Gamma,\mcR)$ and $p\in C_q$ is optimal for $\mcG$. If $p$ is an extreme point in $C_q$, then
    \begin{enumerate}[(a)]
        \item $ T_{\fin}^{(p)}(\Gamma,\mcR)\cap \partial_e T_{\fin}(\PVM^{Y,B})= \partial_e T_{\fin}^{(p)}(\Gamma,\mcR)$,
        \item $T_{\fin}^{(p)}(\Gamma,\mcR)=conv\big(\partial_e T_{\fin}^{(p)}(\Gamma,\mcR)\big)$, and
        \item $p$ has a projective quantum model that employs a maximally entangled state.
    \end{enumerate}
\end{lemma} 
\begin{proof}
    For part (a), it is clear that $T_{\fin}^{(p)}(\Gamma,\mcR)\cap \partial_e T_{\fin}(\PVM^{Y,B})\subseteq \partial_e T_{\fin}^{(p)}(\Gamma,\mcR)$. To see the converse inclusion, we assume for the sake of contradiction that there is a $\tau\in\partial_e T_{\fin}^{(p)}(\Gamma,\mcR)$ which is not in $\partial_e T_{\fin}(\PVM^{Y,B})$. Then there are distinct $\tau_1,\tau_2\in T_{\fin}(\PVM^{Y,B})$ and $0< \lambda <1$ such that $\tau=\lambda \tau_1+(1-\lambda)\tau_2$. $\tau$ is in $T^{(p)}(\Gamma,\mcR)$, so $\tau(r^*r)=0$ for all $r\in\mcR$. 
    Since $\lambda>0$, $1-\lambda>0$, and $\tau_1(r^*r),\tau_2(r^*r)\geq 0$, it follows that
    \begin{align}
        \tau_1(r^*r)=\tau_2(r^*r)=0\label{eq:all0}
    \end{align}
    for all $r\in\mcR$. By \Cref{lemma:extendtrace}, the vectors $p_1,p_2\in \R^{A\times B\times X\times Y}$ defined by $p_i(a,b|x,y)=\tau_i(\gamma^x_an^y_b),i\in\{1,2\}$ are optimal quantum correlations for $\mcG$ such that $p=\lambda p_1 +(1-\lambda)p_2$. Because $p$ is an extreme point in $C_q$, we have $p_1=p_2=p$. \Cref{eq:all0} implies $\tau_1,\tau_2\in T_{\fin}^{(p)}(\Gamma,\mcR)$, which contradicts the assumption that $\tau$ is an extreme point in $T_{\fin}^{(p)}(\Gamma,\mcR)$. We conclude that $ T_{\fin}^{(p)}(\Gamma,\mcR)\cap \partial_e T_{\fin}(\PVM^{Y,B})= \partial_e T_{\fin}^{(p)}(\Gamma,\mcR)$. 

    Now we prove part (b). For any $\tau\in T_{\fin}^{(p)}(\Gamma.\mcR)$, by \Cref{prop:fdtracial}, $\tau$ is a convex combination $\tau=\lambda_1\lambda_1+\cdots+\lambda_k\lambda_k$ of $\tau_1,\ldots,\tau_k\in T_{\fin}(\PVM^{Y,B})$. Since every $\lambda_1>0$, $\tau(r^*r)=0$ for all $r\in\mcR$ implies $\tau_i(r^*r)=0$ for all $1\leq i\leq k$ and $r\in \mcR$. Again, by \Cref{lemma:extendtrace} and the extremality of $p$, every $\tau_i$ is on $T_{\fin}^{(p)}(\Gamma,\mcR)$. Hence by part (a), every $\tau_i$ is in $\partial_e T_{\fin}^{(p)}(\Gamma,\mcR)$. This proves part (b).

    For (c), we first note that $T_{\fin}^{(p)}(\Gamma,\mcR)$ is non-empty, because $p$ is optimal for $\mcG$. Let $\tau\in \partial_e T_{\fin}^{(p)}(\Gamma,\mcR)$. Then by part (a), $\tau\in \partial_e T_{\fin}(\PVM^{Y,B})$. By \Cref{prop:tracialGNS}, $\tau$ has a GNS representation $ \big(\C^d\otimes\C^d,\Id\otimes\pi(\cdot),\ket{\varphi_d}  \big)$
    where $\pi:\PVM^{Y,B}\arr M_d(\C^d)$ is an irreducible $*$-representation and $\ket{\varphi_d}$ is a maximally entangled state. Let $M^x_a:=\pi(\gamma^x_a)^T$ and $N^y_b:=\pi(n^y_b)$ for all $a,b,x,y$. Then following the same reasoning as in the proof of \Cref{lemma:fullrankprojective}, 
    \begin{align*}
        S:=\big(\C^d,\C^d,\{M^x_a\},\{N^y_b\},\ket{\varphi_d} \big)
    \end{align*}
    is a projective quantum model for $p$ that employs a maximally entangled state.
\end{proof}

Now we are ready to proof \Cref{thm:correlationtrace}.

\begin{proof}[Proof of \Cref{thm:correlationtrace}]
    (a)$\Leftrightarrow$(b): By \Cref{lemma:fullrankprojective}, $p$ has a full-rank projective quantum model. Then by \cite[Corollary 3.6]{PSZZ23}, $p$ is a self-test for all quantum models if and only if $p$ is an abstract state self-test for projective finite-dimensional states. So (a)$\Leftrightarrow$(b) follows.

    (b)$\Leftrightarrow$(c) follows directly from the bijective correspondence between finite-dimensional states on $\PVM^{X,A}\otimes_{min}\PVM^{Y,B}$ for $p$ and tracial states in $T_{\fin}^{(p)}(\Gamma,\mcR)$, as shown in \Cref{prop:extendtrace}.

    (c)$\Leftrightarrow$(d): By \Cref{lemma:extremeMES}, $T_{\fin}^{(p)}(\Gamma,\mcR)$ contains a unique element if and only if $\partial_e T_{\fin}^{(p)}(\Gamma,\mcR)$ contains a unique element. So (c)$\Leftrightarrow$(d) follows from the correspondence between tracial states in $\partial_e T_{\fin}^{(p)}(\Gamma,\mcR)$ and representations in $\Irr_{\fin}\big(\PVM^{Y,B} \big)$ which satisfy part (d) of \Cref{thm:correlationtrace}, as shown in \Cref{cor:extremetrace-irrep}.
\end{proof}

\begin{corollary}\label{cor:traceideal}
    Let $\mcG$ be a nonlocal game with a determining pair $(\Gamma,\mcR)$. Suppose $p\in\partial_e C_q$ is optimal for $\mcG$ and is a self-test for quantum models. Let $f$, $\tau$, and $\pi:\PVM^{Y,B}\arr M_d(\C)$ be unique finite-dimensional state on $\PVM^{X,A}\otimes_{min}\PVM^{Y,B}$ for $p$, then unique tracial stat in $T_{\fin}(\Gamma,\mcR)$, and the unique irreducible representation given by part (d) of \Cref{thm:correlationtrace} respectively. Let $\wtd{M}^x_a:=\pi(\gamma^x_a)^T$ and $\wtd{N}^y_b:=\pi(n^y_b)$ for all $a,b,x,y$, and let $\ket{\varphi_d}:=\frac{1}{\sqrt{d}}\sum_{i=1}^d\ket{i}\otimes\ket{i}$. Then
\begin{enumerate}[(a)]
    \item the tuple $\big(\C^d\otimes\C^d,\Id\otimes\pi(\cdot),\ket{\varphi_d}  \big)$ is a GNS representation for $\tau$,
    \item $\wtd{S}:=\big(\C^d,\C^d,\{\wtd{M}^x_a\},\{\wtd{N}^y_b\},\ket{\varphi_d} \big)$ is an ideal model for $p$, and
    \item the tuple $\big(\C^d\otimes\C^d,\pi_A\otimes\pi_B,\ket{\varphi_d}  \big)$ is a GNS representation for $f$, where $\pi_A\otimes\pi_B$ is the associated representation of $\wtd{S}$.
\end{enumerate}
\end{corollary}

\begin{comment}
\begin{corollary}
    Suppose $\mcG$ is a nonlocal game with a determining pair $(\gamma,\mcR)$ and $\mcG$ has an optimal quantum correlation $p\in C_q(X,Y,A,B)$ which is an extreme point in $C_q$. If $p$ is a self-test for all quantum models, then $p$ has an ideal model
    \begin{align*}
        \wtd{S}=\big(\C^d,\C^d,\{\wtd{M}^x_a\},\{\wtd{N}^y_b\},\ket{\varphi_d}  \big)
    \end{align*}
    with associated representation $\wtd{\pi}_A\otimes\wtd{\pi}_B$ such that
    \begin{enumerate}[(a)]
        \item $\ket{\varphi_d}=\frac{1}{\sqrt{d}}\sum_{i=1}^d\ket{i}\otimes\ket{i}$ is a maximally entangled state, and
        \item $M^x_a=\wtd{\pi}_B(\gamma^x_a)^T$ for all $x\in X,a\in A$.
    \end{enumerate}
\end{corollary}
\end{comment}

\subsection{General results for nonlocal games}
We also prove the following general statement for nonlocal games.
\begin{theorem}\label{thm:gametrace}
   Let $\mcG$ be a nonlocal game with a determining pair $(\Gamma,\mcR)$, and let $C^*(\mcG)$ be the associated game algebra. Then the following statements are equivalent.
   \begin{enumerate}[(a)]
       \item $\mcG$ is a self-test for its optimal quantum strategies.
       \item There is a unique finite-dimensional optimal state on $\PVM^{X,A}\otimes_{min}\PVM^{Y,B}$ for $\mcG$.
       \item $C^*(\mcG)$ has a unique finite-dimensional tracial state.
       \item $C^*(\mcG)$ has a unique finite-dimensional irreducible $*$-representation.
   \end{enumerate}
\end{theorem}

For the proof, we need to establish a correspondence between optimal states for $\mcG$ and tracial states on $C^*(\mcG)$. Given a nonlocal game $\mcG$ with a determining pair $(\Gamma,\mcR)$, we say a tracial state $\tau$ on $\PVM^{Y,B}$ is optimal for $\mcG$ if $\tau(r^*r)=0$ for all $r\in \mcR$. By \Cref{lemma:extendtrace}, if $\tau$ is an optimal tracial state for $\mcG$, then the correlation defined by $p(a,b|x,y)=\tau(\gamma^x_an^y_b)$ is optimal for $\mcG$. Hence a tracial state $\tau$ on $\PVM^{Y,B}$ is optimal for $\mcG$ if and only if $\tau$ is in $T^{(p)}(\Gamma,\mcR)$ for some optimal $p\in C_{qc}$ for $\mcG$. The following correspondence between optimal states on $\PVM^{X,A}\otimes\PVM^{Y,B}$ for $\mcG$ and optimal tracial states on $\PVM^{X,A}\otimes\PVM^{Y,B}$ for $\mcG$ is an immediate consequence of  \Cref{prop:extendtrace}.

\begin{lemma}
    Suppose $\mcG$ has a determining pair $(\Gamma,\mcR)$. There is a bijective correspondence between optimal states on $\PVM^{X,A}\otimes\PVM^{Y,B}$ for $\mcG$ and optimal tracial states on $\PVM^{Y,B}$ for $\mcG$, sending an $f$ to the tracial state $\tau:=f|_{1\otimes\PVM^{Y,B}}$. Moreover,
    \begin{enumerate}[(a)]
        \item $f$ is finite-dimensional if and only if the corresponding $\tau$ is finite-dimensional, and 
        \item $f$ is min-continuous if the corresponding $\tau$ is amenable.
    \end{enumerate}
\end{lemma}

\begin{comment}
The following lemma establishes a correspondence between optimal tracial states on $\PVM^{Y,B}$ for $\mcG$ and tracial states on $\PVM^{Y,B}$ whose kernel contains the ideal generated by $\mcR$.

\begin{lemma}\label{lemma:optimaltrace}
    Suppose $\mcG$ has a determining pair $(\gamma,\mcR)$. A tracial state $\tau$ on $\PVM^{Y,B}$ is optimal for $\mcG$ with respect to $(\gamma,\mcR)$ if and only if $\tau(r^*r)=0$ for all $r\in\mcR$.
\end{lemma}
\begin{proof}
    The ``only if" direction follows from the definition of optimal tracial states and $T^{(p)}(\gamma,\mcR)$.

    To see the ``if" direction, suppose $\tau$ is a tracial state on $\PVM^{Y,B}$ that satisfies $\tau(r^*r)=0$ for all $r\in\mcR$. Then by \Cref{lemma:extendtrace}, $\tau$ extends to an optimal state $f$ on $\PVM^{X,A}\otimes\PVM^{Y,B}$ for $\mcG$ such that $\tau(\gamma^x_an^y_b)=f(m^x_a\otimes n^y_b)$ for all $a,b,x,y$. Let $p$ be the correlation achieved by $f$, then $p\in C_{qc}$ is optimal for $\mcG$ and $\tau\in T^{(p)}(\gamma,\mcR)$. So $\tau$ is an optimal tracial state on $\PVM^{Y,B}$ for $\mcG$ with respect to $(\gamma,\mcR)$.
\end{proof}
\end{comment}

Any optimal tracial state for a nonlocal game $\mcG$ with determining pair $(\Gamma,\mcR)$ drops to a tracial state on the quotient $C^*(\mcG)$ and vice versa. The following correspondence between optimal tracial states for $\mcG$ and tracial states on the associated game algebra $C^*(\mcG)$ follows straight from \Cref{prop:tracequotientfd,prop:tracequotientamenable}.

\begin{lemma}
    Suppose $\mcG$ has a determining pair $(\Gamma,\mcR)$, and let $C^*(\mcG)$ be the associated game algebra. Then there is a bijective correspondence between optimal tracial states on $\PVM^{Y,B}$ for $\mcG$ with respect to $(\Gamma,\mcR)$ and tracial states on $C^*(\mcG)$. Moreover, 
    \begin{enumerate}[(a)]
        \item an optimal tracial state on $\PVM^{Y,B}$ is finite-dimensional if and only if the corresponding tracial state on $C^*(\mcG)$ is finite-dimensional, and
        \item an optimal tracial state on $\PVM^{Y,B}$ is amenable if the corresponding tracial state on $C^*(\mcG)$ is amenable.
    \end{enumerate}
\end{lemma}

Put everything together, we conclude that:
\begin{proposition}\label{prop:extendtracegame}
    Suppose $\mcG$ has a determining pair $(\Gamma,\mcR)$. Then there is a bijective correspondence between optimal states on $\PVM^{X,A}\otimes \PVM^{Y,B}$ for $\mcG$ and tracial states on $C^*(\mcG)$. Moreover, 
    \begin{enumerate}[(a)]
        \item an optimal state is finite-dimensional if and only if the corresponding tracial state is finite-dimensional, and 
        \item an optimal state is min-continuous if the corresponding tracial state is amenable.
    \end{enumerate}

\end{proposition}

Now we are ready to prove \Cref{thm:gametrace}
\begin{proof}[Proof of \Cref{thm:gametrace}]
    (a)$\Leftrightarrow$(b): By \Cref{lemma:uniqueextreme,lemma:gameabstractstate}, both (a) and (b) imply that $\mcG$ has a unique optimal quantum correlation $p$ and $p$ is an extreme point in $C_q$. So (a)$\Leftrightarrow$(b) follows from (a)$\Leftrightarrow$(b) in \Cref{thm:correlationtrace}. 

    (b)$\Leftrightarrow$(c) follows directly from the correspondence between finite-dimensional optimal states on $\PVM^{X,A}\otimes_{min}\PVM^{Y,B}$ for $\mcG$ and finite-dimensional tracial states on $C^*(\mcG)$, as shown in \Cref{prop:extendtracegame}.

    (c)$\Leftrightarrow$(d): By \Cref{prop:fdtracial}, $T_{\fin}\big( C^*(\mcG)\big)$ contains a unique element if and only if $\partial_eT_{\fin}\big( C^*(\mcG)\big)$ contains a unique element. So (b)$\Leftrightarrow$(c) follows from \Cref{cor:extremetrace-irrep}.
\end{proof}

\begin{corollary}\label{cor:irreps}
    Let $\mcG$ be a nonlocal game with a determining pair $(\Gamma,\mcR)$. Let $C^*(\mcG)$ be the associated game algebra, and let $q:\PVM^{Y,B}\arr C^*(\mcG)$ be the quotient map. Suppose $\mcG$ is a self-test for its optimal quantum strategies. Let $p\in\partial_e C_q$ be the unique optimal quantum correlation for $\mcG$ given by \Cref{lemma:uniqueCq}. Then $p$ is a self-test for quantum models. Let $\wtd{\pi}$ be the unique $*$-representation in $\Irr_{\fin}\big(C^*(\mcG)\big)$ given by the part (d) of \Cref{thm:gametrace}, and let $\pi$ be the irreducible $*$-representation on $\PVM^{Y,B}$ given by part (d) of \Cref{thm:correlationtrace}. Then $\pi=\wtd{\pi}\circ q$. 
\end{corollary}

%\subsection{Self-testing for XOR and synchronous games}
\subsection{Self-testing for synchronous games}\label{subsec:syn}
Recall that a game $\mcG=(X,Y,A,B,\mu,V)$ is said to a synchronous game if $A=B,X=Y$, and $V(a,a'|x,x)=0$ for all $x\in X$ and $a\neq a'$ in $A$. We make the convention that the question distribution in a synchronous game is uniform on $X\times Y$, and we use $\mcG=(X,A,V)$ to denote a synchronous game. 

Given a synchronous game $\mcG=(X,A,V)$ and a projective commuting operator strategy $S=(\mcH,\{M^x_a\},\{N^y_b\},\ket{\psi})$ for $\mcG$. The strategy $S$ is perfect for $\mcG$ if and only if 
\begin{align}
    \bra{\psi}M^x_aN^y_b\ket{\psi}=0 \label{eq:MN0}
\end{align}
whenever $V(a,b|x,y)=0$. As shown in \cite[Theorem 5.5]{PSSTW16}, if $S$ is perfect for $\mcG$, then 
\begin{align}
    M^x_a\ket{\psi}=N^x_a\ket{\psi}\label{eq:M=N}
\end{align}
for all $a\in A$ and $b\in B$.

\begin{proposition}\label{prop:synpair}
    Suppose $\mcG=(X,A,V)$ is a synchronous game with $w_{qc}(\mcG)=1$. Then $\mcG$ has a determining pair $(\Gamma,\mcR)$ where
    \begin{itemize}
        \item $\gamma^x_a=n^x_a$ for all $a\in A$ and $x\in X$, and
        \item $\mcR=\{n^x_an^y_b:V(a,b|x,y)=0      \}$.
    \end{itemize}
\end{proposition}
\begin{proof}
    Since every $\{n^x_a:a\in A\}$ is a PVM in $\PVM^{X,A}$, every $\gamma^x_a$ is self-adjoint in $\PVM^{X,A}$ and every $\{\gamma^x_a:a\in A\}$ is a PVM in $C^*(\mcG):=\PVM^{X,A}\slash\ang{\mcR}$.  Condition (i) of \Cref{def:pair} holds. 
    
    Next we examine condition (ii). Let $f$ be a state on $\PVM^{X,A}\otimes\PVM^{Y,B}$, and let $(\mcH,\pi,\ket{\psi})$ be a GNS representation for $f$. Let $M^x_a:=\pi(m^x_a),N^x_a:=\pi(n^x_a)$ for all $a\in A,x\in X$, and let $S:=(\mcH,\{M^x_a\},\{N^y_b\},\ket{\psi})$ be a projective commuting operator strategy for $\mcG$

    Suppose $f$ is perfect for $\mcG$. Then $S$ is  perfect for $\mcG$. \Cref{eq:M=N}
    implies that
    \begin{align*}
        f\big((m^x_a\otimes 1-1\otimes \gamma^x_a)^2    \big)=\norm{\left(M^x_a-N^x_a\right)\ket{\psi}}^2=0
    \end{align*}
    for all $a\in A$ and $x\in X$. Hence $f$ satisfies part (b) of condition (ii) in \Cref{def:pair}. Let $w_1=n^{x_1}_{a_1}\cdots n^{x_k}_{a_k}$ and $w_2=n^{y_1}_{b_1}\cdots n^{y_\ell}_{b_\ell}$ be two monomials in $\PVM^{X,A}$. Note that every $M^x_a$ is self-adjoint and commutes with every $N^y_b$. By \Cref{eq:M=N} again, 
    \begin{align*}
        f(1\otimes w_2w_1)&=\bra{\psi}N^{y_1}_{b_1}\cdots N^{y_\ell}_{b_\ell}N^{x_1}_{a_1}\cdots N^{x_k}_{a_k}\ket{\psi}=\bra{\psi}M^{x_k}_{a_k}\cdots M^{x_1}_{a_1}N^{y_1}_{b_1}\cdots N^{y_\ell}_{b_\ell}  \ket{\psi}\\
        &=\bra{\psi}N^{x_1}_{a_1}\cdots N^{x_k}_{a_k}N^{y_1}_{b_1}\cdots N^{y_\ell}_{b_\ell}N^{x_1}_{a_1}\cdots N^{x_k}_{a_k}\ket{\psi}=f(1\otimes w_1w_2).
    \end{align*}
Since $w_1$ and $w_2$ were arbitrary and monomials over $\{n^y_b:y\in X,b\in A\}$ form a dense subset of $\PVM^{X,A}$, the linearly functional $\tau:=f|_{1\otimes \PVM^{X,A}}$ is a tracial state on $\PVM^{X,A}$. For any relation $r=n^x_an^y_b$ in $\mcR$, by \Cref{eq:MN0,eq:M=N},     
\begin{align*}
\tau(r^*r)=\tau(n^y_bn^x_an^x_an^y_b)=\tau(n^x_an^y_b)=\bra{\psi}N^x_aN^y_b\ket{\psi}=\bra{\psi}M^x_aN^y_b\ket{\psi}=0.
\end{align*}
We conclude that $f$ satisfies condition (ii) in \Cref{def:pair}.

Now suppose $f$ satisfies condition (ii) in \Cref{def:pair}. By part (a) of condition (ii), the strategy $S$ satisfies \Cref{eq:M=N} for all $a\in A$ and $x\in X$. Part (b) of condition (ii) implies that $\tau:=f|_{1\otimes\PVM^{X,A}}$ is a tracial state satisfying
\begin{align*}
    0=\tau\big((n^x_an^y_b)^*(n^x_an^y_b)\big)=\tau(n^x_an^y_b)=\bra{\psi}N^y_aN^y_b\ket{\psi}\ket{\psi}=\bra{\psi}M^x_aN^y_b\ket{\psi}.
\end{align*}
whenever $V(a,b|x,y)=0$. This implies $S$ is perfect for $\mcG$, and hence $f$ is perfect for $\mcG$.
\end{proof}
We refer to $C^*(\mcG):=\PVM^{X,A}\slash\ang{\mcR}$, with $\mcR$ defined in \Cref{prop:synpair}, as the \emph{synchronous algebra} for $\mcG$. 

Note that for any synchronous correlation $p\in C_{qc}^{syn}(X,A)$, one can define a synchronous game $\mcG=(X,A,V)$ such that $p$ is perfect for $\mcG$. So  \Cref{thm:correlationtrace} applies to all synchronous correlations.
\begin{theorem}\label{thm:syntrace}
    Suppose a synchronous quantum correlation $p\in C_q^{syn}(X,A)$ is an extreme point in $C_q$. Then $p$ is a self-test for quantum models if and only there is a unique finite-dimensional tracial state on $\PVM^{Y,B}$ such that $\tau(n^x_an^y_b)=p(a,b|x,y)$ for all $a,b,x,y$. 
\end{theorem}
\begin{proof}
    Let $\mcG=(X,A,V)$ be the nonlocal synchronous game defined by $V(a,b|x,y)=0$ if and only if $p(a,b|x,y)=0$. Then $p$ is a perfect quantum correlation for $\mcG$. Let $(\Gamma,\mcR)$ be the determining pair of $\mcG$ defined in \Cref{prop:synpair}. Then a finite-dimensional tracial state $\tau$ on $\PVM^{X,A}$ is in $T_{\fin}^{p}(\Gamma,\mcR)$ if and only if 
    \begin{enumerate}[(i)]
        \item $\tau(n^x_an^y_b)=p(a,b|x,y)$ for all $a,b,x,y$, and 
        \item $\tau\big((n^x_an^y_b)^*(n^x_an^y_b)\big)=\tau(n^x_an^y_b)=0$ whenever $V(a,b|x,y)=0$.
    \end{enumerate}
    Note that (ii) is equivalent to $\tau(n^x_an^y_b)=0$ whenever $p(a,b|x,y)=0$, and this is already implied by (i). So $T_{\fin}^{(p)}(\Gamma,\mcR)$ consists of all finite-dimensional tracial states that satisfy (i). The theorem follows from (a)$\Leftrightarrow$(c) in \Cref{thm:correlationtrace}.
\end{proof}

 In terms of self-testing for synchronous games, \Cref{thm:gametrace} immediately implies the following characterization.
    \begin{theorem}\label{syngameselftest}
        Let $\mcG=(X,A,V)$ be a synchronous game, and let $C^*(\mcG)$ be the associated synchronous algebra. The following statements are equivalent.
        \begin{enumerate}[(a)]
            \item $\mcG$ is a self-test for its perfect quantum strategies.
            \item $C^*(\mcG)$ has a unique finite-dimensional tracial states.
            \item $C^*(\mcG)$ has a unique finite-dimensional irreducible $*$-representation.
        \end{enumerate}
    \end{theorem}
    
\subsection{Self-testing for XOR games}\label{subsec:XOR}
Recall that a nonlocal game $\mcG=(I,J,A,B,\mu,V)$ is an XOR game if $A=B=\{0,1\}$, and there is a matrix $(t_{ij})\in \{0,1\}^{I\times J}$ such that
\begin{align*}
    V(a,b|i,j)=\begin{cases}
        1 & \text{ if } a\oplus b=t_{ij}\\
        0 & \text { otherwise}.
    \end{cases}
\end{align*}
We also define the cost matrix $(\omega_{ij})$ for $\mcG$ where $\omega_{ij}:=(-1)^{t_{ij}}\mu(i,j)$. We always assume $\mu$ is non-degenerate, meaning that $\mu(i,j)\neq 0$ for all $i\in X,j\in Y$. So the cost matrix completely determines $(t_{ij})$ and $\mu$. We use $\mcG=\big(I,J,(\omega_{ij})\big)$ to denote an XOR game. When working with XOR games, we usually use binary observable generators $x_i:=m^i_0-m^i_1,i\in I$ for $\PVM^{I,\{0,1\}}$ and $y_j:=n^j_0-n^j_1,j\in J$ for $\PVM^{J,\{0,1\}}$.

Every XOR game $\mcG=\big(I,J,(\omega_{ij})\big)$ has associated marginal row biases $\{r_i;i\in I\}$ and associated column biases $\{c_j:j\in J\}$, where $r_i> 0$ and $c_j>0$\footnote{In general, if $\mu$ is degenerate, then the cost matrix $(\omega_{ij})$ may have a all-zero row or column and there could be a zero row or column bias.} for all $i,j\in I$.  Let $f$ be state on $\PVM^{I,\{0,1\}}\otimes\PVM^{J,\{0,1\}}$, let $(\mcH,\pi,\ket{\psi})$ be a GNS representation for $f$, and let
\begin{align*}
    X_i:=\pi\big(x_i\otimes 1\big),i\in I \text{ and }Y_j:=\pi\big(1\otimes y_j\big),j\in J
\end{align*}
be the binary observables employed by Alice and Bob. Then $\{X_i\ket{\psi}:i\in I\},\{Y_j\ket{\psi}:j\in J\}$ is a vector strategy for $\mcG$. As shown in \cite[Corollary 3.2]{Slof11},  $f$ is optimal for $\mcG$ if and only if
\begin{align}
    \frac{1}{r_i}\sum_{j\in J} \omega_{ij} Y_j\ket{\psi}=X_i\ket{\psi}\label{eq:row}
\end{align}
for all $i\in I$. \Cref{eq:row} can also be replaced by
\begin{align}
    \frac{1}{c_j}\sum_{i\in I}\omega_{ij}X_i\ket{\psi}=Y_j\ket{\psi}\label{eq:column}
\end{align}
for all $j\in J$.

\begin{proposition}\label{prop:XORpair}
    Let $\mcG=\big(I,J,(\omega_{ij})\big)$ be an XOR game, and let $\{r_i:i\in I\}$ be the associated row biases. Then $\mcG$ has a determining pair $(\Gamma,\mcR)$, where $\Gamma=\{\gamma^i_a:i\in I,a\in\{0,1\}\}$ is given by
    \begin{align}
        \gamma^i_a=\frac{1}{2}\Big(1-\frac{1}{r_i}\sum_{j\in J} \omega_{ij}(1-2n^j_a) \Big),\label{eq:gamma1}
    \end{align}
    and the relations $\mcR$ consists of
    \begin{align}
        \big(\sum_{j\in J} \omega_{ij}(n^j_0-n^j_1)\big)^2-r_i^2
    \end{align}
    for all $i\in I$.
 \end{proposition}

\begin{comment}
Note that by replacing the PVM generators with corresponding binary observable generators, we see that \Cref{eq:gamma0,eq:gamma1} are equivalent to
    \begin{align}
       \widehat{x}_i=\frac{1}{r_i}\sum_{j\in J} \omega_{ij}y_j, i\in I
    \end{align}
    where $\widehat{x}_i:= \gamma^i_0-\gamma^i_1$.
    %, and $\mcR$ consists of relations
\end{comment}

\begin{proof}
    Since every $n^j_a$ is self-adjoint in $\PVM^{J,\{0,1\}}$, every $\gamma^i_a$ is self-adjoint in  $\PVM^{J,\{0,1\}}$. \Cref{eq:gamma1} implies that 
    \begin{align*}
        \gamma^i_0-\gamma^i_1=\frac{1}{r_i}\sum_{j\in J} \omega_{ij}(n^j_0-n^j_1),
    \end{align*}
    so the relations in $\mcR$ can be expressed as $(\gamma^i_0-\gamma^i_1)^2-1,i\in I$. Hence every $\gamma^i_0-\gamma^i_1$ is a binary observable in $C^*(\mcG):=\PVM^{J,\{0,1\}}\slash\ang{\mcR}$. This means every $\{\gamma^i_0,\gamma^i_1\}$ is a PVM in $C^*(\mcG)$. Condition (i) of \Cref{def:pair} holds. 
    
    Now we examine 
    Condition (ii) of \Cref{def:pair}. Let $f$ be state on $\PVM^{I,\{0,1\}}\otimes\PVM^{J,\{0,1\}}$, let $(\mcH,\pi,\ket{\psi})$ be a GNS representation for $f$, and let $ X_i:=\pi\big(x_i\otimes 1\big),Y_j:=\pi\big(1\otimes y_j\big)$, and $\widehat{Y}_j:=\frac{1}{c_j}\sum_{i\in I}\omega_{ij}X_i$
    for all $i\in I,j\in J$, where $\{c_j:j\in J\}$ is the column biases.
    
    Suppose $f$ is optimal. Then $X_i,Y_j$ satisfy \Cref{eq:row,eq:column}. This implies
    \begin{align*}
        \pi\big(1\otimes (\gamma^i_0-\gamma^i_1)\big)\ket{\psi}=\frac{1}{r_j}\sum_{j\in J}w_{ij}Y_j\ket{\psi}=X_i\ket{\psi}
    \end{align*}
    for all $i\in I$. Replacing the binary observables with the corresponding PVM elements in the above equation, we have
\begin{align*}
    \pi(1\otimes \gamma^i_a)\ket{\psi}=\pi(m^i_a\otimes 1)\ket{\psi}.
\end{align*}
It follows that
\begin{align*}
    f\left( \left(m^i_a\otimes 1- 1\otimes\gamma^i_a\right)^2  \right)=0
\end{align*}
for all $i\in I$ and $a\in \{0,1\}$. Hence $f$ satisfies part (a) of condition (ii). Note that $\widehat{Y}_j,Y_\ell$ are self-adjoint and  $[\widehat{Y}_j,Y_\ell]=0$ for all $j,\ell\in J$. Then for any monomials  $W_1:=y_{i_1}\cdots y_{i_k}$ and $W_2:=y_{j_1}\cdots y_{j_\ell}$ in $\PVM^{J,\{0,1\}}$, \Cref{eq:column} implies that 
\begin{align*}
f(1\otimes W_1W_2)&=\bra{\psi}Y_{i_1}\cdots Y_{i_k}Y_{j_1}\cdots Y_{j_\ell}\ket{\psi}=\bra{\psi}\widehat{Y}_{j_\ell}\cdots\widehat{Y}_{j_1}Y_{i_1}\cdots Y_{i_k}\ket{\psi}\\
&=\bra{\psi}Y_{j_1}\cdots Y_{j_\ell}Y_{i_1}\cdots Y_{i_k}\ket{\psi}=f(1\otimes W_2W_1).
\end{align*}
It follows that $\tau:=f|_{1\otimes \PVM^{J,\{0,1\}}}$ is a tracial state. For any relation $s=\big(\sum_{j\in J}\omega_{ij}y_i\big)^2-r_i^2\in \mcR$, by \Cref{eq:row},
\begin{align*}
    \pi(1\otimes s)\ket{\psi}&=r_i^2\big(\widehat{Y}_j^2\ket{\psi} -\ket{\psi}\big)=r_i^2\big(\widehat{Y}_jY_j\ket{\psi} -\ket{\psi}\big)\\
&=r_i^2\big(Y_j\widehat{Y}_j\ket{\psi} -\ket{\psi}\big)=r_i^2\big(Y_j^2\ket{\psi} -\ket{\psi}\big) =0.
\end{align*}
The last equality uses the fact that $Y_j^2=\Id$. Hence $\tau(s^*s)=\bra{\psi}\pi(1\otimes s^*s)\ket{\psi}=0$ for all $s\in \mcR$. We conclude that $f$ satisfies part (b) condition (ii) in \Cref{def:pair}.

On the other hand, suppose $f$ satisfies parts (a) and (b) in condition (ii). Then in part (a), by replacing the PVM elements with the corresponding binary observables, we see that \Cref{eq:row} holds. This means $f$ is optimal. We conclude that $(\Gamma,\mcR)$ satisfies condition (ii) in \Cref{def:pair}. Hence $(\Gamma,\mcR)$ is a determining pair for $\mcG$.
\end{proof}

We refer to $C^*(\mcG):=\PVM^{J,\{0,1\}}\slash\ang{\mcR}$, with $\mcR$ defined in \Cref{prop:XORpair}, as the \emph{solution algebra} for $\mcG$. In \cite[Proposition 2.8]{Slof11}, Slofstra shows that $\mcG$ has a unique optimal quantum correlation $p$ if and only if the solution algebra for $\mcG$ is Clifford. Furthermore, the associated XOR correlation $c$ with $p$ has rank $r$ if and only if the corresponding Clifford algebra is of rank $r$. Here the Clifford algebra $Cl_r$ of rank $r$ refers to the universal $C^*$-algebra generated by indeterminates $x_1,\cdots,x_r$ and subject to the relations $x_ix_j+x_jx_i=2\delta_{ij},1\leq i,j\leq n$. This algebra has either one (when $r$ is even) or two (when $r$ is odd) irreducible $*$-representations of dimension $2^{\lfloor r/2 \rfloor}$.

\begin{theorem}
    An XOR game $\mcG$ is a self-test for its optimal quantum strategies if and only if $\mcG$ has a unique optimal quantum correlation $p$ and the associated XOR correlation with $p$ has even rank.
\end{theorem}
\begin{proof}
For the ``only if" direction, suppose $\mcG$ is a self-test for its optimal quantum strategies. Then by \Cref{lemma:uniqueCq}, $\mcG$ has a unique optimal quantum correlation $p$. Let $r$ be the rank of the associated XOR correlation with $p$. By \Cref{thm:gametrace}, the solution algebra $Cl_r$ for $\mcG$ has a unique finite-dimensional irreducible $*$-representation. So $r$ must be even.

For the ``if" direction, suppose $\mcG$ has a unique optimal quantum correlation $p$ and the associated XOR correlation with $p$ has even rank. This implies the solution algebra $Cl_r$ for $\mcG$ has a unique finite-dimensional irreducible $*$-representation. By \Cref{thm:gametrace} again, $\mcG$ is a self-test for its optimal quantum strategies.
\end{proof}

\section{Tracial states and robust self-tests}\label{sec:tracerobustselftest}

Suppose a nonlocal game $\mcG$ has a determining pair $(\Gamma,\mcR)$, and let $C^*(\mcG)$ be the associated game algebra. In this section, we show that if $(\Gamma,\mcR)$ is robust in a certain way, then the robust self-testing property of $\mcG$ can be characterized by amenable tracial states on $C^*(\mcG)$. 

\subsection{Robustness of determining pairs}
\begin{definition}\label{def:stablepair}
A determining pair $(\Gamma,\mcR)$ for a nonlocal game $\mcG$ is $\nu$-robust (or simply called robust) if there exists a function $\nu:\R_{\geq 0}\arr\R_{\geq 0}$ with $\nu(\eps)\arr 0$ as $\eps\arr 0$ such that for any $\eps\geq 0$ and any projective $\eps$-optimal quantum strategy
    \begin{align*}
        \models
    \end{align*}
   for $\mcG$,
    \begin{enumerate}[(i)]
        \item $\norm{\big(\pi_A(m^x_a)\otimes\Id-\Id\otimes\pi_B(\gamma^x_a)    \big)\ket{\psi}}\leq \nu(\eps)$ for all $x\in X,a\in A$, where $\pi_A\otimes\pi_B$ is the associated representation of $S$,
        \item $\norm{\pi_B(r)}_{\rho_B}\leq \nu(\eps)$ for all $r\in\mcR$, where $\rho_B:=\Tr_{\mcH_B}(\ket{\psi}\bra{\psi})$, and
        \item $\norm{\pi_B(n^y_b)\sqrt{\rho_B}-\sqrt{\rho_B}\pi_B(n^y_b)}_F\leq \nu(\eps)$ for all $y\in Y,b\in B$.
    \end{enumerate}
\end{definition}

Recall that the big Frobenius norm $\norm{\cdot}_F$ on $M_d(\C)$ is defined by $\norm{T}_F=\sqrt{\Tr(T^*T)}$, and the $\rho$-norm $\norm{\cdot}_\rho$ is defined by $\norm{T}_{\rho}=\sqrt{\Tr(T^*T\rho)}$. Intuitively, a determining pair $(\Gamma,\mcR)$ is robust for a game $\mcG$ if in any near-optimal strategy for $\mcG$, the action of Alice's measurement $m^x_a$ is approximately equal to the action of Bob's measurement $\gamma^x_a$, and Bob's measurements respect the relations in $\mcR$ approximately.

\begin{lemma}\label{lemma:approxrep}
    Suppose $\mcG$ is a nonlocal game with a $\nu$-robust determining pair $(\Gamma,\mcR)$. Let 
    \begin{align*}
        \models
    \end{align*}
    be a projective $\eps$-optimal quantum strategy for $\mcG$, and let $\phi:=f_S|_{1\otimes\PVM^{Y,B}}$ where $f_S$ is the state on $\PVM^{X,A}\otimes_{min}\PVM^{Y,B}$ induced by $S$.
    Then
    \begin{enumerate}[(a)]
        \item $\sqrt{\phi(r^*r)}\leq \nu(\eps)$ for all $r\in\mcR$, 
        \item $\norm{\pi_B(W)\sqrt{\rho_B}-\sqrt{\rho_B}\pi_B(W)}_F\leq \deg(W)\nu(\eps)$ for any monomial $W\in\C^*\ang{n^y_b:b\in B,y\in Y }$,
        \item $\abs{\phi(W_1W_2)-\phi(W_2W_1)}\leq \big(\deg(W_1)+\deg(W_2)\big)\nu(\eps)$ for any monomials $W_1,W_2\in\C^*\ang{n^y_b:b\in B,y\in Y }$, and
        \item $\abs{\phi(\gamma^x_an^y_b)-p_S(a,b|x,y)}\leq \nu(\eps)$ for all $a,b,x,y$.
    \end{enumerate}
\end{lemma}
\begin{proof}
    Let $\pi_A\otimes\pi_B$ be the associated representation of $S$, and let $\rho_B:=\Tr_{\mcH_B}(\ket{\psi}\bra{\psi})$. For any $r\in \mcR$,
    \begin{align*}
\phi(r^*r)=\bra{\psi}\Id\otimes\pi_B(r^*r)\ket{\psi}=\norm{\pi_B(r)}^2_{\rho_B}\leq \nu(\eps)^2.
    \end{align*}
So (a) follows.

We prove (b) by induction on the degree $k$ of $W$. $k=1$ follows from condition (iii) of \Cref{def:stablepair}. Now suppose (b) holds for any monomial of degree $k$ for some $k\geq 1$. For any monomial $W=\alpha_1\cdots \alpha_{k}\alpha_{k+1}$ in $\PVM^{Y,B}$ of degree $k+1$, let $W':=\alpha_1\cdots \alpha_{k}$ be a monomial of degree $k$. Then by the induction hypothesis,
\begin{align*}
    \norm{\pi_B(W')\sqrt{\rho_B}-\sqrt{\rho_B}\pi_B(W')}_F\leq k\nu(\eps).
\end{align*}
Note that $\norm{\pi_B(W')}_{op}\leq 1$ and $\norm{\pi_B(\alpha_{k+1})}_{op}\leq 1$, it follows that
\begin{align*}
    \norm{\pi_B(W)\sqrt{\rho_B}-\sqrt{\rho_B}\pi_B(W)}_F  =& \norm{\pi_B(W')\pi_B(\alpha_{k+1})\sqrt{\rho_B}-\sqrt{\rho_B}\pi_B(W')\pi_B(\alpha_{k+1})}_F\\
    \leq& \norm{\pi_B(W')\big( \pi_B(\alpha_{k+1})\sqrt{\rho_B}-\sqrt{\rho_B}\pi_B(\alpha_{k+1}) \big)}_F\\
    &+ \norm{\big(\sqrt{\rho_B}\pi_B(W')-\pi_B(W')\sqrt{\rho_B}  \big)\pi_B(\alpha_{k+1})}_F\\
    \leq & \norm{\pi_B(W')}_{op}\norm{\pi_B(\alpha_{k+1})\sqrt{\rho_B}-\sqrt{\rho_B}\pi_B(\alpha_{k+1})}_F\\
    & + \norm{\pi_B(\alpha_{k+1})}_{op}\norm{\sqrt{\rho_B}\pi_B(W')-\pi_B(W')\sqrt{\rho_B} }_F\\
    \leq & \nu(\eps)+k\nu(\eps)= (k+1)\nu(\eps).
\end{align*}
So (b) holds for monomials of degree $k+1$. We conclude that (b) holds.

To prove (c), we first recall that for any matrices $S$ and $T$,
\begin{align*}
    \abs{\Tr(ST)}\leq \norm{S}_{op}\cdot \norm{T}_F.
\end{align*}
Note that $\norm{\sqrt{\rho_B}}_{op}\leq 1$, $\norm{\pi_B(W_1)}_{op}\leq 1$, and $\norm{\pi_B(W_2)}_{op}\leq 1$. By part (b) we have 
\begin{align*}
    \abs{\phi(W_1W_2)-\phi(W_2W_1)}=&\abs{\Tr\Big(\sqrt{\rho_B}\pi_B(W_1)\big(\pi_B(W_2)\sqrt{\rho_B}-\sqrt{\rho_B}\pi_B(W_2)\big)       \Big)\\
    &+ \Tr\Big(\sqrt{\rho_B}\pi_B(W_2)\big(\pi_B(W_1)\sqrt{\rho_B}-\sqrt{\rho_B}\pi_B(W_1)\big)       \Big)   }\\
    \leq & \norm{\pi_B(W_1)\sqrt{\rho_B}-\sqrt{\rho_B}\pi_B(W_1)}_F + \norm{\pi_B(W_2)\sqrt{\rho_B}-\sqrt{\rho_B}\pi_B(W_2)}_F\\
    \leq & \deg(W_1)\nu(\eps) + \deg(W_2)\nu(\eps) = \big(\deg(W_1)+\deg(W_2) \big) \nu(\eps).
\end{align*}
This proves (c).

Now we prove (d). Since $\norm{\pi_B(n^y_B)}_{op}\leq 1$ and $\norm{\big(\pi_A(m^x_a)\otimes\Id-\Id\otimes\pi_B(\gamma^x_a)    \big)\ket{\psi}}\leq \nu(\eps)$,
\begin{align*}
    \abs{\phi(\gamma^x_an^y_b)-p_S(a,b|x,y)}&=\abs{\bra{\psi}\Id\otimes\pi_B(\gamma^x_an^y_b)\ket{\psi}-\bra{\psi}\pi_A(m^x_a)\otimes\pi_B(n^y_b)\ket{\psi} } \\
    &= \abs{\bra{\psi}\big(\pi_A(m^x_a)\otimes\Id-\Id\otimes\pi_B(\gamma^x_a)\big)\cdot\big(\Id\otimes\pi_B(n^y_b) \big)\ket{\psi} }\\
    &\leq \norm{\big(\pi_A(m^x_a)\otimes\Id-\Id\otimes\pi_B(\gamma^x_a)    \big)\ket{\psi}}\leq \nu(\eps)
\end{align*}
for all $a,b,x,y$. This completes the proof.
\end{proof}

Given a quantum model 
\begin{align*}
    \models,
\end{align*}
by picking the Schmidt basis for $\ket{\psi}$ as bases for $\mcH_A$ and $\mcH_B$, and dilating Alice or Bob's system so that they have the same dimension, we can always write
\begin{align}
    S=\big(\C^d,\C^d,\{M^x_a\},\{N^y_b\},\ket{\psi}=\sum_{i=1}^d\lambda_i\ket{i}\otimes\ket{i}   \big),\label{eq:balance}
\end{align}
where $d\geq 1$ and $\lambda_1,\ldots,\lambda_d\geq 0$, without changing the state induced by $S$. We refer to \Cref{eq:balance} as a \emph{balanced form} of $S$. Indeed, if $s:=\dim(\mcH_B)-\dim(\mcH_A)>0$, then we take $\mcH_A\oplus\C^s$ as the new space for Alice and define her new measurement operators to be $M^x_a\oplus \Id$ for one fixed $a\in A$ and $M^x_{a'}\oplus 0$ for the rest $a'\in A$. We can process similarly when $\dim(\mcH_A)>\dim(\mcH_B)$. So we can always assume $\dim(\mcH_A)=\dim(\mcH_B)=d$ for some $d\geq 1$. Consider the Schmidt decomposition $\ket{\psi}=\sum_{i=1}^k\lambda_i\ket{\alpha_i}\otimes\ket{\beta_i}$ where $k\leq d$, $\lambda_1,\ldots,\lambda_k>0$, and $\{\ket{\alpha_i}:1\leq i\leq k\}$ and $\{\ket{\beta_i}:1\leq i\leq k\}$ are orthonormal sets in $\mcH_A$ and $\mcH_B$ respectively. We can expand these two orthonormal sets to orthonormal bases $\{\ket{\alpha_i}:1\leq i\leq d\}$ and $\{\ket{\beta_i}:1\leq i\leq d\}$ for $\mcH_A$ and $\mcH_B$ respectively, and write $\ket{\psi}=\sum_{i=1}^d\lambda_i\ket{\alpha_i}\otimes\ket{\beta_i}$ where $\lambda_{k+1}=\cdots=\lambda_d=0$. After fixing the bases $\{\ket{\alpha_i}:1\leq i\leq d\}$ and $\{\ket{\beta_i}:1\leq i\leq d\}$, we can identify both $\mcH_A$ and $\mcH_B$ with $\C^d$, and identify $\ket{\psi}$ with $\sum_{i=1}^d\lambda_i\ket{i}\otimes\ket{i}$. One of the advantages of the balanced form \Cref{eq:balance} is that the reduced density matrices of $\ket{\psi}$ on Alice and Bob's sides are both $\rho:=\sum_{i=1}^d\lambda_i^2\ket{i}\bra{i}$. So we have $\ket{\psi}=\sqrt{\rho}\otimes\Id\ket{\widehat{\varphi}_d}=\Id\otimes\sqrt{\rho}\ket{\widehat{\varphi}_d}$, where $\ket{\widehat{\varphi}_d}:=\sum_{i=1}^d\ket{i}\otimes\ket{i}$.

\begin{proposition}\label{prop:limitamenable}
   Let $\mcG=(X,Y,A,B,\mu,V)$ be a nonlocal game with a $\nu$-robust determining pair $(\Gamma,\mcR)$. Suppose $p\in C_q$ is an optimal quantum correlation for $\mcG$. Let $S_n,n\in \N$ be a sequence of projective quantum models such that $\norm{p_{S_n}-p}_1\arr 0$ as $n\arr \infty$, and let $\phi_n:=f_{S_n}|_{1\otimes\PVM^{Y,B}}$ for all $n\in\N$, where $f_{S_n}$ is the state on $\PVM^{X,A}\otimes_{min}\PVM^{Y,B}$ induced by $S_n$. Then any weak*-accumulation point $\phi$ of $\{\phi_n,n\in \N\}$ is an amenable tracial state on $\PVM^{Y,B}$ satisfying
   \begin{enumerate}[(a) ]
       \item $\phi(\gamma^x_an^y_b)=p(a,b|x,y)$ for all $a,b,x,y$, and
       \item $\phi(r^*r)=0$ for all $r\in \mcR$.
   \end{enumerate}
\end{proposition}
We remind the reader that using the notation before, there is another way to phrase this proposition: any weak*-accumulation point of $\{\phi_n:n\in\N\}$ is an amenable tracial state in $T^{(p)}(\Gamma,\mcR)$.
\begin{proof}
    By passing a subsequence, we may assume that $\lim\limits_{n\arr\infty}\phi_n=\phi$ in the weak*-topology. Note that 
    $\norm{p_{S_n}-p}_1\arr 0$ implies $\eps_n:=w_q(\mcG)-w(\mcG;S_n)\arr 0$ as $n\arr\infty$. By part (d) of \Cref{lemma:approxrep},  
    \begin{align*}
        \abs{\phi(\gamma^x_an^y_b)&-p(a,b|x,y)}\\
        &\leq \abs{\phi(\gamma^x_an^y_b)-\phi_n(\gamma^x_an^y_b)}+\abs{\phi_n(\gamma^x_an^y_b)-p_{S_n}(a,b|x,y)}+\abs{p_{S_n}(a,b|x,y)-p(a,b|x,y) }\\
        &\leq \abs{\phi(\gamma^x_an^y_b)-\phi_n(\gamma^x_an^y_b)}+\abs{p_{S_n}(a,b|x,y)-p(a,b|x,y) } + \nu(\eps_n) \arr 0
    \end{align*}
    as $n\arr \infty$. So (a) follows. By part (a) of \Cref{lemma:approxrep}, $\phi_n(r^*r)\leq \nu(\eps_n)^2\arr 0$ as $n\arr \infty$ for all $r\in\mcR$. This proves (b). 
    
    It follows from part (c) of \Cref{lemma:approxrep} that $\phi$ is a tracial state. Now we only need to show  $\phi$ is amenable. Let $\wtd{f}$ be the state on $(\PVM^{Y,B})^{\op}\otimes\PVM^{Y,B}$ defined by $\wtd{f}(\alpha^{\op}\otimes\beta):=\phi(\alpha\beta),\alpha,\beta\in \PVM^{Y,B}$. We write every $S_n$ in their balanced form
    \begin{align*}
        S_n=\big(\C^{d_n},\C^{d_n},\{M^x_a(n)\},\{N^y_b(n)\},\ket{\psi_n} \big)
    \end{align*}
    and denote by $\rho_n$ the reduced density of $\ket{\psi_n}$. For every $n\in \N$, let $\pi_n^A\otimes\pi_n^B$ be the associated representation of $S_n$ (in the balanced form), let $\wtd{\pi}_n:(\PVM^{Y,B})^{\op}\otimes\PVM^{Y,B}\arr\msB(\C^d\otimes\C^d)$ be the $*$-representation sending 
    \begin{align*}
        1\otimes n^y_b\mapsto \Id\otimes \pi_n^B(n^y_b) \text{ and }n^y_b\otimes 1\mapsto \pi_n^B(n^y_b)^T\otimes \Id,
    \end{align*}
    and let $\wtd{f}_n$ be the finite-dimensional state on $(\PVM^{Y,B})^{\op}\otimes\PVM^{Y,B}$ defined by 
\begin{align*}
\wtd{f}_n(\alpha)=\bra{\psi_n}\wtd{\pi}_n(\alpha)\ket{\psi_n},\alpha\in (\PVM^{Y,B})^{\op}\otimes\PVM^{Y,B}.
\end{align*}
 For any monomials $W_1$ and $W_2$ over $\{n^y_b,y\in Y,b\in B\}$,
    \begin{align*}
        \wtd{f}_n(W_1^{\op}\otimes W_2)&=\bra{\psi_n}\wtd{\pi}_n(W_1^{op}\otimes W_2)\ket{\psi_n}=\bra{\psi_n}\pi_n^B(W_1)^T\otimes \pi_n^B(W_2)\ket{\psi_n}\\
        &= \bra{\widehat{\varphi}_{d_n}}\pi_n^B(W_1)^T\otimes \sqrt{\rho_n}\pi_n^B(W_2)\sqrt{\rho_n}\ket{\widehat{\varphi}_{d_n}}\\
        &=\bra{\widehat{\varphi}_{d_n}}\Id\otimes \sqrt{\rho_n}\pi_n^B(W_2)\sqrt{\rho_n}\pi_n^B(W_1)\ket{\widehat{\varphi}_{d_n}}\\
&=\Tr\big(\sqrt{\rho_n}\pi_n^B(W_2)\sqrt{\rho_n}\pi_n^B(W_1)\  \big),
    \end{align*}
 where $\ket{\widehat{\varphi}_{d_n}}:=\sum_{i=1}^{d_n}\ket{i}\otimes\ket{i}$. So by part (b) of \Cref{lemma:approxrep},
 \begin{align*}
     \abs{\wtd{f}_n(W_1^{\op}\otimes W_2)- \phi_n(W_1W_2)  }&= \abs{\Tr\Big(\sqrt{\rho_n}\pi_n(W_1)\big(\sqrt{\rho_n}\pi_n(W_2)-\pi_n(W_2)\sqrt{\rho_n}\big)  \Big)   }\\
     &\leq \norm{\sqrt{\rho_n}\pi_n(W_2)-\pi_n(W_2)\sqrt{\rho_n}}_F\leq \deg(W_2)\nu(\eps_n).
 \end{align*}
 It follows that
\begin{align*}
    \abs{\wtd{f}_n&(W_1^{\op}\otimes W_2)-\wtd{f}(W_1^{\op}\otimes W_2)}\\
    &\leq \abs{\wtd{f}_n(W_1^{\op}\otimes W_2)-\phi_n(W_1W_2) } + \abs{\phi_n(W_1W_2)-\phi(W_1W_2) } + \abs{\phi(W_1W_2)-\wtd{f}(W_1^{\op}\otimes W_2) }\\
    &\leq \deg(W_2) \nu(\eps_n)+\abs{\phi_n(W_1W_2)-\phi(W_1W_2) } \arr 0
\end{align*}
 as $n\arr \infty$. Since $W_1$ and $W_2$ were arbitrary, $\lim\limits_{n\arr\infty}\wtd{f}_n=\wtd{f}$ in the weak*-topology. This means $\wtd{f}$ is the weak*-limit of a sequence of finite-dimensional states on $\PVM^{X,A}\otimes\PVM^{Y,B}$. Hence $\wtd{f}$ is min-continuous, and thus $\phi$ is amenable.
\end{proof}

\subsection{Stability of game algebras}
Let $\mcG$ be a nonlocal game with a robust determining pair $(\Gamma,\mcR)$. Suppose that there is a unique amenable tracial state $\tau$ on $\PVM^{Y,B}$ respecting all the relations in $\mcR$ and $\tau$ is finite-dimensional. In this section, we show that $\mcG$ must be a robust self-test for its optimal quantum strategies. The proof relies on the stability of the game algebra $C^*(\mcG)=\PVM^{Y,B}\slash\ang{\mcR}$. We outline some key ideas in the following two paragraphs. 

We first demonstrate that there exists a ucp map $\theta:C^*(\mcG)\arr \PVM^{Y,B}$ inverting the quotient map $q:\PVM^{Y,B}\arr\C^*(\mcG)$ in the sense that every $\theta(n^y_b)-n^y_b\in\ang{\mcR}$. Now suppose $\pi:\PVM^{Y,B}\arr\msB(\mcH_B)$ is the representation given by a near-optimal strategy $\models$. Since the determining pair $(\Gamma,\mcR)$ is robust, the representation $\pi$ respects the relations in $\mcR$ approximately. Let $\widehat{\theta}$ be the ucp map from $C^*(\mcG)\arr\msB(\mcH)$ defined by $\widehat{\theta}:=\pi\circ \theta$, and let $V^*\wtd{\pi}(\cdot)V$ be the Stinespring dilation of $\widehat{\theta}$, where $\wtd{\pi}:C^*(\mcG)\arr\msB(\mcK)$ is a $*$-representation and $V:\mcH\arr\mcK$ is an isometry. Then $\pi(n^y_b)-V^*\wtd{\pi}(n^y_b)V$ is small because
\begin{align*}
    \pi(n^y_b)-V^*\wtd{\pi}(n^y_b)V=\pi(n^y_b)-\widehat{\theta}(n^y_b)=\pi(n^y_b)-\pi\big(\theta(n^y_b)\big)=\pi\big(n^y_b-\theta(n^y_b) \big),
\end{align*}
$n^y_b-\theta(n^y_b)$ is in $\mcR$, and $\pi$ respects the relations in $\mcR$ approximately. In other words, $\pi$ is close to the exact representation $\wtd{\pi}$ under dilation.

One can even derive a constructive bound on the difference between $\pi(n^y_b)$ and $V^*\wtd{\pi}(n^y_b)V$ if $n^y_b-\theta(n^y_b)$ can be expressed by relations in $\mcR$ explicitly (see the notion of an $\mcR$-decomposition in \Cref{defn:rde}). This forms the core of the quantitative Gowers-Hatami theorem for $C^*$-algebras that we will introduce in \Cref{sec:GH}. In this section, however, we work with a slightly weaker statement. The following proposition illustrates how to construct a ucp map $\theta:C^*(\mcG)\arr\PVM^{Y,B}$ such that every $\theta(n^y_b)-n^y_b$ is in the kernel of $\tau$.

\begin{proposition}\label{prop:lifting}
    Suppose a nonlocal game $\mcG$ has a determining pair $(\Gamma,\mcR)$. Let $C^*(\mcG)$ be the associated game algebra. For every finite-dimensional tracial state $\tau$ on $\PVM^{Y,B}$ that satisfies $\tau(r^*r)=0$ for all $r\in\mcR$, there is a ucp map $\theta:C^*(\mcG)\arr\PVM^{Y,B}$ such that
    \begin{align*}
        \tau\Big( \big(\theta(\alpha)-\alpha \big)^*\big(\theta(\alpha)-\alpha\big)  \Big)=0
    \end{align*}
    for any $*$-polynomial $\alpha\in \C^*\ang{n^y_b,y\in Y,b\in B}$.
\end{proposition}

Note that $\PVM^{Y,B},C^*(\mcG)$, and $\PVM^{Y,B}\slash\mcI_\tau$ all have generating set $\{n^y_b:y\in Y,b\in B\}$. In the expression $\theta(\alpha)-\alpha$, the first $\alpha$ is viewed as an element of $C^*(\mcG)$, while the second $\alpha$ is an element of $\PVM^{Y,B}$. The proof uses a lifting theorem due to Choi and Effros. 

\begin{lemma}[\cite{CE76}, see also \cite{Arv77}]\label{lemma:lifting}
    Suppose $\mcA$ and $\mcB$ are two unital separable $C^*$-algebras. Suppose $\mcJ$ is a closed two-sided ideal in $\mcB$. Let $q:\mcB\arr\mcB\slash\mcJ$ be the quotient map, and let $\varphi:\mcA\arr\mcB\slash\mcJ$ be a ucp map. If $\varphi$ is nuclear, then there is a ucp map $\theta:\mcA\arr\mcB$ such that $\varphi=q\circ \theta$.
\end{lemma}

\begin{proof}[Proof of \Cref{prop:lifting}]
Suppose $\tau$ is a finite-dimensional tracial state on $\PVM^{Y,B}$ satisfying $\tau(r^*r)=0$ for all $r\in\mcR$. Let $\ang{\mcR}$ be the closed two-sided ideal generated by $\mcR$. Since $\ang{\mcR}\subset \mcI_{\tau}$, there is a surjective
$*$-homomorphism $\varphi:C^*(\mcG)\arr \PVM^{Y,B}\slash\mcI_\tau$ sending $\alpha\mapsto \alpha$ for any $*$-polynomial $\alpha\in \C^*\ang{n^y_b,y\in Y,b\in B}$. Since $\tau$ is finite-dimensional, $\PVM^{X,A}\slash\mcI_\tau$ is a finite-dimensional $C^*$-algebra, and hence nuclear. It follows that $\varphi$ is nuclear.  By \Cref{lemma:lifting}, there is a ucp map $\theta:\mcA(\mcG)\arr \PVM^{Y,B}$ such that $\varphi=q_\tau\circ\theta$, where $q_\tau:\PVM^{Y,B}\arr \PVM^{Y,B}\slash\mcI_\tau$ is the quotient map. For any $*$-polynomial $\alpha\in \C^*\ang{n^y_b,y\in Y,b\in B}$,
\begin{align*}
q_\tau\big(\theta(\alpha)\big)=\varphi(\alpha)=\alpha=q_\tau(\alpha) 
\end{align*}
in $\PVM^{Y,B}\slash\mcI_\tau$.
Hence $\theta(\alpha)-\alpha\in \ker q_\tau=\mcI_{\tau}$.
\end{proof}

\Cref{prop:lifting} is used to prove the following stability result.

\begin{proposition}\label{prop:uniquetraceimplyrobust}
Let $\mcG=(X,Y,A,B,\mu,V)$ be a nonlocal game with a $\nu$-robust determine pair $(\Gamma,\mcR)$. Suppose $p\in C_q$ is the unique optimal quantum correlation for $\mcG$ and there is a unique amenable tracial state $\tau$ in $T^{(p)}(\Gamma,\mcR)$. Then there exists a full-rank projective quantum model $\wtd{S}$ for $p$ such that the following holds. For any sequence of projective quantum models
    \begin{align*}
    S_n=\big(\mcH_n^A,\mcH_n^B,\{M^x_a(n) \},\{N^y_b(n) \},\ket{\psi_n}\big), n\in\N
    \end{align*}
     with $\lim\limits_{n\arr\infty} \norm{p_{S_n}-p}_1=0$,
     there is a function $\eta:\N\arr\R_{\geq 0}$ with $\lim\limits_{n\arr\infty} \eta(n)=0$ such that $S_n\succeq_{\eta(n)}\wtd{S}$ for all $n\in\N$.
\end{proposition}

\begin{proof}
By \Cref{prop:extendtrace}, $\tau$ extends to a state $\wtd{f}$ which is the unique state on $\PVM^{X,A}\otimes_{min}\PVM^{Y,B}$ that can achieve $p$.
Note that $p$ is a quantum correlation. This means $\wtd{f}$ and $\tau$ must be finite-dimensional, and hence $\wtd{f}$ is the unique finite-dimensional state on $\PVM^{X,A}\otimes_{min}\PVM^{Y,B}$ for $p$, and $\tau$ is the unique tracial state in $T_{\fin}^{(p)}(\Gamma,\mcR)$.
Since $p$ is the unique optimal quantum correlation for $\mcG$, by \Cref{lemma:uniqueextreme}, $p$ is an extreme point in $C_q$. Hence by \Cref{thm:correlationtrace}, $p$ is a self-test for all quantum models, and $\mcG$ is a self-test for its optimal quantum strategies. Let $\pi:\PVM^{Y,B}\arr M_d(\C)$ be the finite-dimensional irreducible $*$-representation given by part (d) of \Cref{thm:correlationtrace}, and let $\wtd{\pi}$ be the unique irreducible $*$-representation of $C^*(\mcG)$ given by part (d) of \Cref{thm:gametrace}. Then $\pi=\wtd{\pi}\circ q$ by \Cref{cor:irreps}, where  $q:\PVM^{Y,B}\arr C^*(\mcG)$ is the quotient map. Let $\wtd{M}^x_a:=\pi(\gamma^x_a)^T$ and $\wtd{N}^y_b:=\pi(n^y_b)$ for all $a,b,x,y$, and let $\ket{\varphi_d}:=\frac{1}{\sqrt{d}}\sum_{i=1}^d\ket{i}\otimes\ket{i}$.
 Then by \Cref{cor:traceideal},
\begin{align*}
    \wtd{S}:=\left(\C^d,\C^d,\{\wtd{M}^x_a\},\{\wtd{N}^y_b\},\ket{\varphi_d} \right)
\end{align*}
is a quantum model for $\mcG$, and the triple $(\C^d\otimes\C^d,\wtd{\pi}_A\otimes\wtd{\pi}_B,\ket{\varphi_d})$ is a GNS representation of $\wtd{f}$, where $\wtd{\pi}_A\otimes\wtd{\pi}_B$ is the associated representation of $\wtd{S}$.

Let $\theta:C^*(\mcG)\arr \PVM^{Y,B}$ be the ucp map given by \Cref{prop:lifting} that satisfies 
 \begin{align*}
        \tau\Big( \big(\theta(\alpha)-\alpha \big)^*\big(\theta(\alpha)-\alpha \big)  \Big)=0
    \end{align*}
for any $*$-polynomial $\alpha$ over $\{n^y_b:y\in Y,b\in B\}$. For every $n\in \N$, let $\pi^A_n\otimes\pi^B_n$ be the associated representation of $S_n$, and let $\rho_n^A:=\Tr_{\mcH_n^B}(\ket{\psi_n}\bra{\psi_n})$ and $\rho_n^B:=\Tr_{\mcH_n^A}(\ket{\psi_n}\bra{\psi_n})$.

\begin{claim}\label{claim1}
There are integers $r_n,s_n\in \N$, and isometries $I_{A,n}:\mcH_n^A\arr \C^d\otimes \C^{r_n}$ and $I_{B,n}:\mcH_n^B\arr \C^d\otimes \C^{s_n}$ such that 
\begin{align}
&\delta^A_{x,a}(n):=\norm{I_{A,n}^*(\wtd{M}^x_a\otimes \Id_{\C^{r_n}}) I_{A,n}-M^x_a(n)}_{\rho^A_n} \arr 0, \text{ and}\label{eq:M0}\\
&\delta^B_{y,b}(n):=\norm{I_{B,n}^*(\wtd{N}^y_b\otimes \Id_{\C^{s_n}}) I_{B,n}-N^y_b(n)}_{\rho^B_n} \arr 0\label{eq:N0}
\end{align}
as $n\arr \infty$ for all $a,b,x,y$.   
\end{claim}

\begin{proof}[Proof of \Cref{claim1}]
Without loss of generality, we can write every $S_n$ in its balanced form
\begin{align*}
    S_n=\big(\C^{d_n},\C^{d_n},\{M^x_a(n)\},\{N^y_b(n)\},\ket{\psi_n}\big).
\end{align*}
Let $\rho_n$ be the reduced density of $\ket{\psi_n}$ on $\C^{d_n}$, and let $\phi_n:=f_{S_n}|_{1\otimes \PVM^{Y,B}}$ where $f_{S_n}$ is the state on $\PVM^{X,A}\otimes\PVM^{Y,B}$ induced by $S_n$. By \Cref{prop:limitamenable}, any weak*-accumulation point $\phi$ of $\{\phi_n:n\in\N\}$ is an amenable tracial state in $T^{(p)}(\Gamma,\mcR)$. By the hypothesis, $T^{(p)}(\Gamma,\mcR)$ has a unique amenable tracial state $\tau$. So we must have $\phi=\tau$. By passing a subsequence, we conclude that $\lim\limits_{n\arr\infty}\phi_n=\tau$ is the weak*-topology.

Now consider the ucp maps $\theta^B_n:=\pi^B_n\circ\theta,n\in \N$ from $C^*(\mcG)\arr \msB(\C^{d_n})$. Since $\wtd{\pi}$ is the unique finite-dimensional irreducible $*$-representation of $C^*(\mcG)$, any finite-dimensional $*$-representation of $C^*(\mcG)$ is a direct sum of $\wtd{\pi}$'s. By the Stinespring dilation theorem, there are $s_n\in \N$ and isometries $I_{B,n}: \C^{d_n}\arr\C^d\otimes \C^{s_n},n\in \N$ such that 
\begin{align*}
 \theta^B_n(\alpha)=I_{B,n}^*\big(\wtd{\pi}(\alpha)\otimes \Id_{\C^{s_n}}\big)I_{B,n}   
\end{align*}
for all $\alpha\in C^*(\mcG)$. It follows that
\begin{align*}
    \lim_{n\arr\infty} \norm{I_{B,n}^*\big(\wtd{\pi}(\alpha)\otimes \Id_{\C^{s_n}}\big) I_{B,n}-\pi_n^B(\alpha)}_{\rho_n}^2
    &=\lim_{n\arr\infty} \norm{\theta^B_n(\alpha)-\pi^B_n(\alpha)}_{\rho_n}^2\\
    &=\lim_{n\arr\infty} \Tr_{\rho_n}\circ \pi^B_n\Big(\big(\theta(\alpha)-\alpha\big)^* \big(\theta(\alpha)-\alpha\big)     \Big)\\
    &=\lim_{n\arr\infty} \phi_n\Big(\big(\theta(\alpha)-\alpha\big)^* \big(\theta(\alpha)-\alpha\big)     \Big)\\
    &= \tau\Big(\big(\theta(\alpha)-\alpha\big)^* \big(\theta(\alpha)-\alpha\big)\Big)=0 
\end{align*}
for any $*$-polynomial $\alpha$ over $\{n^y_b:y\in Y,b\in B\}$. Taking $\alpha=n^y_b$, \Cref{eq:N0} follows.

Note that every $\gamma^x_a$ is a $*$-polynomial over $\{n^y_b:y\in Y,b\in B\}$ and is self-adjoint in $\PVM^{Y,B}$. So $\wtd{\pi}\big(\gamma^x_a \big)$ and $\pi_n^B\big(\gamma^x_a \big)$ are self-adjoint, and the above calculation also implies
\begin{align*}
     \norm{I_{B,n}^*\big(\wtd{\pi}\big(\gamma^x_a \big)^T\otimes \Id_{\C^{s_n}}\big) I_{B,n}-\pi_n^B\big(\gamma^x_a \big)^T}_{\rho_n}^2
    &= \norm{I_{B,n}^*\big(\wtd{\pi}\big(\gamma^x_a \big)\otimes \Id_{\C^{s_n}}\big) I_{B,n}-\pi_n^B\big(\gamma^x_a \big)}_{\rho_n}^2\\
    &\arr \tau\Big(\big(\theta(\gamma^x_a)-\gamma^x_a\big)^* \big(\theta(\gamma^x_a)-\gamma^x_a\big)\Big)=0
\end{align*}
as $n\arr\infty$. Again, since $(\Gamma,\mcR)$ is $\nu$-robust, by part (b) of \Cref{lemma:approxrep},
\begin{align*}
    \norm{\pi_n^B(\gamma^x_a)\sqrt{\rho_n}-\sqrt{\rho_n}\pi_n^B(\gamma^x_a)}_F\leq \norm{\gamma^x_a}_{1,1}\nu(\eps_n),
\end{align*}
where $\eps_n:=w_q(\mcG)-w(\mcG;S_n)\arr 0$ as $n\arr\infty$. Here $\norm{\gamma^x_a}_{1,1}$ is the first Sobolev $1$-seminorm of the $*$-polynomial $\gamma^x_a$ (see \cite{MSZ23} for a definition). Hence
\begin{align*}
    \norm{M^x_a(n)-\pi_n^B\big(\gamma^x_a &\big)^T}_{\rho_n}=\norm{\pi^A_n(m^x_a)-\pi_n^B(\gamma^x_a )^T}_{\rho_n}
    =\norm{\big(\pi_n^A(m^x_a)-\pi_n^B(\gamma^x_a )^T\big)\otimes\Id\ket{\psi_n}}\\
    &\leq \norm{\big(\pi_n^A(m^x_a)\otimes\Id-\Id\otimes\pi_n^B(\gamma^x_a)\big)\ket{\psi}}+ \norm{\big(\pi_n^B(\gamma^x_a)^T\otimes\Id-\Id\otimes\pi_n^B(\gamma^x_a)\big)\ket{\psi}}\\
    &\leq \nu(\eps_n) + \norm{\pi_n^B(\gamma^x_a)\sqrt{\rho_n}-\sqrt{\rho_n}\pi_n^B(\gamma^x_a)}_F\leq \big(\norm{\gamma^x_a}_{1,1}+1\big)\nu(\eps_n),
\end{align*}
$\arr 0$ as $n\arr\infty$.
Note that $\big(\wtd{\pi}(\gamma^x_a)\big)^T=\wtd{M}^x_a$. Taking $r_n=s_n$ and $I_{A,n}=I_{B,n}$ for all $n\in \N$, it follows that
\begin{multline*}
\norm{I_{A,n}^*\big(\wtd{M}^x_a\otimes\Id_{\C^{r_n}}\big)I_{A,n}-M^x_a(n)}_{\rho_n}\\ \leq \norm{I_{B,n}^*\big(\wtd{\pi}\big(\gamma^x_a \big)^T\otimes \Id_{\C^{s_n}}\big) I_{B,n}-\pi_n^B\big(\gamma^x_a \big)^T}_{\rho_n}+\norm{M^x_a(n)-\pi_n^B\big(\gamma^x_a \big)^T}_{\rho_n}
\end{multline*}
$\arr 0$ as $n\arr\infty$.
So \Cref{eq:M0} follows.
\end{proof}

\begin{claim}\label{claim2}
    There are unit vectors $\ket{\kappa_n}\in \C^{r_n}\otimes\C^{s_n},n\in\N$ such that
    \begin{align}
        \norm{ I_{A,n}\otimes I_{B,n}\big( M^x_a(n)\otimes N^y_b(n)\ket{\psi}\big) -     \big(\wtd{M}^x_a \otimes \wtd{N}^y_n\ket{\varphi_d}\big)  \otimes \ket{\kappa_n}     } \arr 0 \text{ as } n\arr \infty\label{eq:limit}
    \end{align}
 for all $a,b,x,y$. 
\end{claim}
\begin{proof}[Proof of \Cref{claim2}]
Since $\mcG$ is a self-test, we let $\Delta$ be its spectral gap, and let 
\begin{align*}
    \delta(n):=\sum_{a,b,x,y}\mu(x,y) V(a,b|x,y)\left( \delta^A_{x,a}(n)+\delta^B_{y,b}(n)  \right),n\in\N.
\end{align*}
Since $\wtd{f}$ is the unique finite-dimensional optimal state on $\PVM^{X,A}\otimes_{min}\PVM^{X,A}$ for $\mcG$, by \Cref{prop:gameidealrep} and \Cref{rmk:gameidealrep}, there is a unit vector $\ket{\kappa_n}\in \C^{r_n}\otimes \C^{s_n}$ for every $n\in\N$ such that
\begin{multline*}
    \norm{ I_{A,n}\otimes I_{B,n}\big( M^x_a(n)\otimes N^y_b(n)\ket{\psi}\big) -     \big(\wtd{M}^x_a \otimes \wtd{N}^y_n\ket{\varphi_d}\big)  \otimes \ket{\kappa_n} }\\
    \leq \frac{\sqrt{2}\big(\delta(n)+\sqrt{\epsilon_n}\big)}{\Delta}+\delta^A_{x,a}(n)+\delta^B_{y,b}(n)
\end{multline*}
for all $a,b,x,y$. By \Cref{claim1}, $\lim\limits_{n\arr\infty}\delta^A_{x,a}(n)=\lim\limits_{n\arr\infty}\delta^B_{y,b}(n)=0$, and hence $\lim\limits_{n\arr\infty}\delta(n)=0$. So \Cref{eq:limit} follows.
\end{proof}

 For every $n\in\N$, let
    \begin{align*}
        \eta(n):=\max\limits_{x,y,a,b}\left\{\norm{ I_{A,n}\otimes I_{B,n}\big( M^x_a(n)\otimes N^y_b(n)\ket{\psi}\big) -     \big(\wtd{M}^x_a \otimes \wtd{N}^y_n\ket{\varphi_d}\big)  \otimes \ket{\kappa_n} }   \right\}.
    \end{align*}
Then $S_n\succeq_{\eta(n)}\wtd{S}$ for all $n\in \N$, and $\lim\limits_{n\arr\infty}\eta(n)=0$ by \Cref{claim2}.  This completes the proof.
\end{proof}

\subsection{General results}
\begin{theorem}\label{thm:robustcorr}
Let $\mcG=(X,Y,A,B,\mu,V)$ be a nonlocal game with a robust determining pair $(\Gamma,\mcR)$. Suppose $p\in C_q$ is the unique optimal quantum correlation for $\mcG$. Then the following statements are equivalent.
    \begin{enumerate}[(a)]
        \item $p$ is a robust self-test for all quantum models.
        \item $p$ is an abstract state self-test for all states on $\POVM^{X,A}\otimes_{min}\POVM^{Y,B}$.
        \item $p$ is an abstract state self-test for all states on $\PVM^{X,A}\otimes_{min}\PVM^{Y,B}$.
        \item There is a unique amenable tracial state $\tau$ on $\PVM^{Y,B}$ such that
        \begin{enumerate}[(i)]
            \item $\tau(r^*r)=0$ for all $r\in \mcR$ and
            \item $\tau(\gamma^x_an^y_b)=p(a,b|x,y)$ for all $a,b,x,y$.
        \end{enumerate}
    \end{enumerate}
\end{theorem}
\begin{proof}
    (a)$\Rightarrow$(b) follows from \Cref{thm:robustuniquestate}.

    (b)$\Rightarrow$(c): Suppose $p$ is an abstract state self-test for all states on $\POVM^{X,A}\otimes_{min}\POVM^{Y,B}$. That is, there is a unique state $f$ on $\POVM^{X,A}\otimes_{min}\POVM^{Y,B}$ for $p$. By \Cref{lemma:fullrankprojective}, $p$ has a full-rank projective quantum model, so $f$ must be finite-dimensional and projective. Hence $f$ drops to a finite-dimensional state $\wtd{f}$ on $\PVM^{X,A}\otimes_{min}\PVM^{Y,B}$ and $\wtd{f}$ is a state for $p$. Now assume for the sake of contradiction that there is another state $f'\neq \wtd{f}$ on $\PVM^{X,A}\otimes\PVM^{Y,B}$ that can achieve $p$. Then the pull-back of $f'$ on $\POVM^{X,A}\otimes_{min}\POVM^{Y,B}$ is a state for $p$ but is distinct to $f$, a contradiction. Hence $\wtd{f}$ is the unique state on $\PVM^{X,A}\otimes_{min}\PVM^{Y,B}$ for $p$.

    (c)$\Rightarrow$(d): By \Cref{prop:extendtrace}, any amenable tracial state in $T^{(p)}(\Gamma,\mcR)$ extends uniquely to a state on $\PVM^{X,A}\otimes_{min}\PVM^{Y,B}$ for $p$. If there is a unique state $f$ on $\PVM^{X,A}\otimes_{min}\PVM^{Y,B}$ for $p$, then $\tau:=f|_{1\otimes\PVM^{Y,B}}$ must be the unique amenable tracial state in $T^{(p)}(\Gamma,\mcR)$.

    (d)$\Rightarrow$(a): Suppose there is a unique amenable tracial state $\tau$ in $T^{(p)}(\Gamma,\mcR)$. Let $\wtd{S}$ be the full-rank projective quantum model given in \Cref{prop:uniquetraceimplyrobust}. We first claim the $p$ is a robust self-test for projective quantum models with ideal model $\wtd{S}$. Assume for the sake of contradiction that there is a $\delta_0>0$ and a sequence of projective quantum models $S_n,n\in\N$ such that $\eps_n:=\norm{p_{S_n}-p}_1\arr 0 $ as $n\arr\infty$, but for every $n\in \N$, the relation $S_n\succeq_{\delta_0} \wtd{S}$ does not hold. By \Cref{prop:uniquetraceimplyrobust}, there is a function $\eta:\N\arr\R_{\geq 0}$ such that $\lim\limits_{n\arr\infty}\eta(n)=0$ and $S_n\succeq_{\eta(n)}\wtd{S}$ for all $n\in \N$. Take a big enough $N\in\N$ such that $\eta(N)\leq\delta_0$. Then $S_N\succeq_{\delta_0}\wtd{S}$, which contradicts the assumption. We conclude that $p$ is a robust self-test for projective quantum models. Note that the ideal model $\wtd{S}$ is full-rank and projective. By \cite[Theorem 4.1]{lifiting23}, $p$ is a robust self-test for all quantum models.
\end{proof}

In terms of robust self-testing for nonlocal games, we state the following theorem. 

\begin{theorem}\label{thm:robustgame}
For any nonlocal game $\mcG$ that has a robust determining pair $(\Gamma,\mcR)$, the following statements are equivalent.
    \begin{enumerate}[(a)]
        \item $\mcG$ is a robust self-test for its optimal quantum strategies.
        \item There is a unique optimal state $f$ on $\POVM^{X,A}\otimes_{min}\POVM^{Y,B}$ for $\mcG$ and $f$ is finite-dimensional. 
        \item There is a unique optimal state $f$ on $\PVM^{X,A} \otimes_{min} \PVM^{Y,B}$ for $\mcG$ and $f$ is finite-dimensional.
        \item There is a unique amenable tracial state $\tau$ on $\PVM^{Y,B}$ such that $\tau(r^*r)=0$ for all $r\in\mcR$ and $\tau$ is finite-dimensional.
    \end{enumerate}
\end{theorem}
\begin{proof}
    Any one of the above four statements implies $\mcG$ has a unique optimal quantum correlation. So the theorem follows directly from \Cref{thm:robustcorr}. 
\end{proof}

In \cref{thm:gametrace}, we have shown that a nonlocal game $\mcG$ is a self-test if and only if the associated game algebra $C^*(\mcG)$ has a unique finite-dimensional tracial state. Part (d) of \Cref{thm:robustgame} suggests one may have a similar characterization of robust self-testing: $\mcG$ is a robust self-test if and only if $C^*(\mcG)$ has a unique amenable tracial state $\tau$. However, as discussed in \Cref{rmk:tracequotientamenable}, it is possible that an amenable tracial state on $\PVM^{Y,B}$ whose kernel contains $\mcR$ drops to a non-amenable tracial state on the quotient $C^*(\mcG)$. So it is possible that $C^*(\mcG)$ has a unique amenable tracial state but there are multiple amenable tracial states on $\PVM^{Y,B}$ that are optimal for $\mcG$. Nonetheless, if $C^*(\mcG)$ has a unique tracial state $\wtd{\tau}$ and $\wtd{\tau}$ is finite-dimensional, then the tracial state $\tau$ on $\PVM^{Y,B}$ induced by $\wtd{\tau}$ must be the unique tracial state (and hence the unique amenable tracial state) that satisfies $\tau(r^*r)=0$ for all $r\in\mcR$. We conclude that:
\begin{corollary}\label{cor:uniquetrace}
    Suppose a nonlocal game $\mcG$ has a robust determining pair $(\Gamma,\mcR)$. Let $C^*(\mcG)$ be the associated game algebra. If $C^*(\mcG)$ has a unique tracial state $\tau$ and $\tau$ is finite-dimensional, then $\mcG$ is a robust self-test for its optimal quantum strategies.
\end{corollary}

\subsection{Robust self-testing for XOR and synchronous games}
For XOR games, the robustness of the $(\Gamma,\mcR)$ determining pair defined in \Cref{prop:XORpair} was first proved in \cite{Slof11}. For synchronous games, similar results were observed in \cite{MPS21}. Both cases were also proved in \cite{Pad22}. We summarize these results in the following proposition.
\begin{proposition}[Proposition 5.14 and Proposition 5.23 in \cite{Pad22}]\label{prop:stablepair}
$\quad$
    \begin{enumerate}[(1)]
        \item Let $\mcG=(X,A,V)$ be a synchronous game with $w_q(\mcG)=1$, and let $(\Gamma,\mcR)$ be as in \Cref{prop:synpair}. Then $(\Gamma,\mcR)$ is a $O(\eps^{1/4})$-robust determining pair for $\mcG$.
        \item Let $\mcG=(I,J,(\omega_{ij}))$ be an XOR game, and let $(\Gamma,\mcR)$ be as in \Cref{prop:XORpair}. Then $(\Gamma,\mcR)$ is a $O(\eps^{1/4})$-robust determining pair for $\mcG$.
    \end{enumerate}
\end{proposition}

Hence for synchronous games and XOR games, we have:

\begin{theorem}
    Let $\mcG$ be a synchronous game, and let $C^*(\mcG)$ be the associated synchronous algebra. If $C^*(\mcG)$ has a unique tracial state $\tau$ and $\tau$ is finite-dimensional, then $\mcG$ is a robust self-test for its perfect quantum strategies.
\end{theorem}

\begin{theorem}
    If $p$ is an extreme point of $C_q^{unbiased}$ and the associated XOR correlation $c$ has even rank, then $p$ is a robust self-test for all quantum models.
\end{theorem}

\section{Self-testing in parallel with synchronous games}\label{sec:parallel}
In this section, we apply our tracial-state characterization of self-testing to study self-testing in parallel with synchronous games. We first recall that given two nonlocal games $\mcG_1=(X_1,Y_1,A_1,B_1,V_1)$ and $\mcG_2=(X_2,Y_2,A_2,B_2,V_2)$, their product $\mcG_1\times\mcG_2$ is the nonlocal game $(X_1\times X_2,Y_1\times Y_2,A_1\times A_2,B_1\times B_2,V_1\times V_2)$ where 
\begin{equation*}
    V_1\times V_2\left((a_1,a_2),(b_1,b_2)|(x_1,x_2),(y_1,y_2)  \right)=V_1(a_1,b_1|x_1,y_1)V_2(a_2,b_2|x_2,y_2).
\end{equation*}
When $\mcG_1=\mcG_2$, their product is just a parallel repetition. It is easy to see that if both $\mcG_1$ and $\mcG_2$ are synchronous games, then $\mcG_1\times\mcG_2$ is also a synchronous game. In this case, the following characterization for $C^*(\mcG_1\times\mcG_2)$, the synchronous algebra of $\mcG_1\times \mcG_2$, was established in \cite{productsyn}. 

\begin{theorem}[Theorem 3.1 in \cite{productsyn}]
    Let $\mcG_1$ and $\mcG_2$ be synchronous games. The games $\mcG_1$ and $\mcG_2$ have perfect quantum strategies if and only if $\mcG_1\times \mcG_2$ does. In this case, the associated synchronous algebra $C^*(\mcG_1\times\mcG_2)$ for the product game $\mcG_1\times\mcG_2$ is $*$-isomorphic to $C^*(\mcG_1)\otimes_{max}C^*(\mcG_2)$.
\end{theorem}

\begin{lemma}\label{lemma:irreptensor}
    Let $\mcA$ and $\mcB$ be two $C^*$-algebras. The $C^*$-algebra $\mcA\otimes_{max}\mcB$ has a unique finite-dimensional irreducible representation if and only if both $\mcA$ and $\mcB$ have unique finite-dimensional irreducible representations.
\end{lemma}
\begin{proof}
    We first prove the contrapositive of the ``only if" direction. Suppose $\mcA$ has two different finite-dimensional irreducible representations $\pi_A$ and $\phi_A$. Let $\pi_B$ be a finite-dimensional irreducible representation of $\mcB$. Then $\pi_A\otimes\pi_B$ and $\phi_A\otimes \pi_B$ are different finite-dimensional irreducible representations for $\mcA\otimes_{max}\mcB$. 

    Now we prove the ``if" direction. Suppose $\mcA$ and $\mcB$ have unique finite-dimensional irreducible representations $\wtd{\pi}_A$ and $\wtd{\pi}_B$ on $\mcH_A$ and $\mcH_B$ respectively. Let $\pi:\mcA\otimes_{max}\mcB \arr\msB(\mcH)$ be a finite-dimensional irreducible representation of $\mcA\otimes_{max}\mcB$.\footnote{Such a representation always exists. For instance, $\wtd{\pi}_A\otimes\wtd{\pi}_B$ is one, and the unique one, as the proof shows.} Then there are finite-dimensional representations $\pi_A:\mcA\arr\msB(\mcH)$ and $\pi_B:\mcB\arr\msB(\mcH)$ with commuting ranges such that $\pi(a\otimes b)=\pi_A(a)\pi_B(b)$ for all $a\in\mcA$ and $b\in\mcB$. Note that $\pi_A$ can be decomposed into a direct sum of finite-dimensional irreducible representations of $\mcA$ and $\wtd{\pi}_A$ is the unique finite-dimensional irreducible representation of $\mcA$, so $\mcH\cong \mcH_A\otimes \mcH_B$ and $\pi_A\cong \wtd{\pi}_A\otimes\Id_{\mcH_B}$ for some Hilbert space $\mcH_B$. Since $\pi_B$ and $\pi_A$ have commuting ranges and $\wtd{\pi}_A$ is irreducible, $\pi_B(b)$ acting trivially on $\mcH_A$ for all $b\in\mcB$. So $\pi_B=\Id_{\mcH_A}\otimes\widehat{\pi}_B$ for some representation $\widehat{\pi}_B:\mcB\arr \msB(\mcH_B)$, and hence $\pi\cong\wtd{\pi}_A\otimes\widehat{\pi}_B$. By \cite[Lemma 4.11]{PSZZ23}, $\widehat{\pi}_B$ must be irreducible. It follows that $\widehat{\pi}_B\cong\wtd{\pi}_B$ and $\pi\cong\wtd{\pi}_A\otimes\wtd{\pi}_B$.
\end{proof}

Recall from \Cref{syngameselftest} that a synchronous game $\mcG$ is a self-test for its perfect quantum strategies if and only if the associated synchronous algebra $C^*(\mcG)$ has a unique irreducible $*$-representation. As an immediate consequence of \Cref{lemma:irreptensor}, we have:

\begin{theorem}\label{thm:productselftest}
    Let $\mcG_1$ and $\mcG_2$ be two synchronous games. The product game $\mcG_1\times\mcG_2$ is a self-test for its perfect quantum strategies if and only if both $\mcG_1$ and $\mcG_2$ are self-tests for their perfect quantum strategies.
\end{theorem}

This theorem can be easily generalized to more products of copies (aka. parallel repetition) of a synchronous game.

\begin{corollary}\label{cor:parallelselftest}
    A synchronous game $\mcG$ is a self-test for its perfect quantum strategies if and only if the parallel repeated game $\mcG^{\times n}$ is a self-test for its perfect strategies.
\end{corollary}

It is natural to ask whether \Cref{thm:productselftest,cor:parallelselftest} also hold for robust self-testing.
\begin{conjecture}
    Let $\mcG_1$ and $\mcG_2$ be two synchronous games. The product game $\mcG_1\times\mcG_2$ is a robust self-test for its perfect strategies if and only if both $\mcG_1$ and $\mcG_2$ are robust self-tests for their perfect strategies.
\end{conjecture}

\begin{comment}
In \Cref{coro:uniqueirreptrace} we've shown that a $C^*$-algebra has a unique finite-dimensional irreducible $*$-representation if and only if it has a unique finite-dimensional tracial state. So \Cref{lemma:irreptensor} can be equivalently stated as the following.
\begin{lemma}
Let $\mcA$ and $\mcB$ be two $C^*$-algebras. The $C^*$-algebra $\mcA\otimes_{max}\mcB$ has a unique finite-dimensional tracial state $\tau$ if and only $\mcA$ has a unique finite-dimensional tracial state $\tau_A$ and $\mcB$ has a unique finite-dimensional $\tau_B$.
    
\end{lemma}

states that for any $C^*$-algebras $\mcA$ and $\mcB$, the $C^*$-algebra $\mcA\otimes_{max}\mcB$ has a unique finite-dimensional tracial state if and only if both $\mcA$ and $\mcB$ have unique tracial states. To study robust self-testing for product games or parallel repeated games, 
\end{comment}

\section{A quantitative Gowers-Hatami theorem for game algebras}\label{sec:GH}

In \cite{GH17}, Gowers and Hatami prove that any finite group $G$ is dilation-stable: if a function $f$ from $G$ to unitaries respects the multiplication table of $G$ approximately, then there must be a representation $\phi$ of $G$ and an isometry $I$ such that $f(g)$ and $I^*\phi(g)I$ are close in $\norm{\cdot}_{hs}$. A state-norm version of the Gowers-Hatami theorem was introduced by Vidick in \cite{Vid18} and has been widely used in proving robust self-testing results. In \cite{MPS21}, Man\v{c}inska, Praksh, and Schafhauser introduce an analog of the Gowers-Hatami theorem that applies to some $C^*$-algebras. Using this result, they construct a family of constant-sized nonlocal games that can robustly self-test for maximally entangled states of unbounded dimension. However, the Gowers-Hatami theorem in \cite{MPS21} is not quantitative, in the sense that there is no explicit function to characterize the ``distance" between the approximate representations and exact representations. Consequently, the robustness of their self-tests is non-constructive. In this section, we state and prove the first quantitative Gowers-Hatami theorem for $C^*$-algebres and show how it can be used to derive the explicit robustness function of a self-test.

We first need to recall some notions from \cite{MSZ23}.

\begin{definition}[Definition 3.2 in \cite{MSZ23}]\label{defn:rde}
    Let $\mcA$ be a $*$-algebra generated by a set of unitaries $\mcX$. Let $\mcR\subseteq\C^*\ang{\mcX}$ be a set of $*$-polynomials over
    $\mcX$. For any $*$-polynomial $f\in\C^*\ang{\mcX}$ that is trivial in
    $\mcA/\ang{\mcR}$, we say that $\sum_{i=1}^n \lambda_iu_ir_iv_i$ is an
    $\mcR$-decomposition for $f$ in $\mcA$ if
    \begin{enumerate}
        \item $u_i,v_i$ are $*$-monomials in $\C^*\ang{\mcX}$ for all $1\leq i\leq n$, 
        \item $r_i\in\mcR\cup\mcR^*$ for all $1\leq i\leq n$,
        \item $\lambda_i \in \C$ for all $1\leq i\leq n$, and
        \item $f=\sum_{i=1}^n \lambda_iu_ir_iv_i$ in $\mcA$.
    \end{enumerate}
    The size of an $\mcR$-decomposition $\sum_{i=1}^n \lambda_i u_i r_i v_i$ is $\sum_{i=1}^{n} |\lambda_i| (1+ \norm{r_i}_{\mcA} \deg(v_i))$,
    where $\norm{\cdot}_{\mcA}$ is the operator norm in $\mcA$.
 \end{definition}
Here a $*$-monomial in $\C^*\ang{\mcX}$ is a product $a_1a_2\cdots a_k$ where
$k \geq 0$ and $a_1,\ldots, a_k\in\mcX\cup\mcX^*$ (the integer $k$ is called
the degree). Since all the generators in $\mcX$ are unitary in $\mcA$, every $u_i$ and $v_i$ are unitary in $\mcA$. We see that the size of an $\mcR$-decomposition does not depend on $u_i$'s at all. This is because later we will evaluate $u_ir_iv_i$ in some state norm $\norm{\cdot}_\rho$. Every $u_i$ is unitary in $\mcA$ and $\norm{\cdot}_\rho$ is left unitarily invariant, so $\norm{u_ir_iv_i}_\rho=\norm{r_iv_i}_\rho$. The size of an $\mcR$-decomposition does depend on the degree of the monomials $v_i$ because we want to switch $r_i$ and $v_i$ in $\norm{\cdot}_\rho$ and $\deg(v_i)$ is the price we need to pay.

Given any $*$-algebra $\mcA$ and a set of relations $\mcR\subset \mcA$, we say a $*$-representation $\pi:\mcA\arr\msB(\mcH)$ is an $(\eps,\rho,\mcR)$-representation for some $\eps\geq 0$ and density operator $\rho$ on $\mcH$, if
\begin{align*}
 \norm{\pi(r)}_\rho\leq \eps
\end{align*}
for all $r\in\mcR$. For example, suppose a nonlocal game $\mcG$ has a $\nu$-robust determining pair $(\Gamma,\mcR)$, and let $\models$ be an $\eps$-optimal strategy with associated representation $\pi_A\otimes\pi_B$. Then $\pi_B$ is a $(\nu(\eps),\rho_B,\mcR)$-representation for $\PVM^{Y,B}$, where $\rho_B$ is the reduced density of $\ket{\psi}$ on $\mcH_B$.

Now we can state our quantitative Gowers-Hatami theorem. In the following, we work with unitary generators $\{b_y:y\in Y\}$ for $\PVM^{Y,B}$.

\begin{theorem}\label{thm:algGH}
 Let $\mcG=(X,Y,A,B,\mu,V)$ be a nonlocal game with a $\nu$-robust determining pair $(\Gamma,\mcR)$, and let $C^*(\mcG)=\PVM^{Y,B}\slash \ang{\mcR}$ be the associated game algebra. Suppose there is a ucp map $\theta:C^*(\mcG)\arr \PVM^{Y,B}$ and a positive real number $\Lambda$ such that every $\theta(b_y)-b_y$ has an $\mcR$-decomposition in $\PVM^{Y,B}$ with size at most $\Lambda$. Then for any $(\eps,\rho,\mcR)$-representation $\pi:\PVM^{Y,B}\arr\msB(\mcH)$, there is a $*$-representation $\wtd{\pi}:C^*(\mcG)\arr\msB(\wtd{\mcH})$ and an isometry $I:\mcH\arr\wtd{\mcH}$ such that
 \begin{equation}
     \norm{\pi(b_y)-I^*\wtd{\pi}(b_y)I}_\rho\leq \Lambda\cdot\nu(\eps)\label{eq:GH}
 \end{equation}
 for all $y\in Y$. If in addition, $\mcG$ is a self-test with spectral gap $\Delta$, then $\mcG$ is a robust self-test with robustness $O\left(\frac{\Lambda\nu(\eps)+\sqrt{\eps}}{\Delta}\right)$.
\end{theorem}
\begin{proof}
 Let $\pi:\PVM^{Y,B}\arr\msB(\mcH)$ be a $*$-representation such that $\norm{\pi(r)}_\rho\leq \nu(\eps)$ for all $r\in\mcR$. Then $\widehat{\theta}:=\pi\circ \theta$ defines a ucp map from $C^*(\mcG)\arr\mcB(\mcH)$. By the Stinespring dilation theorem, there is a $*$-representation $\wtd{\pi}:C^*(\mcG)\arr \msB(\wtd{H})$ and an isometry $I:\mcH\arr\wtd{\mcH}$ such that $\widehat{\theta}(\alpha)=I^*\wtd{\pi}(\alpha)I$ for all $\alpha\in C^*(\mcG)$.

 Fix a $y\in Y$, and let $\sum_{i=1}^n\lambda_iu_ir_iv_i$ be an $\mcR$-decomposition in $\PVM^{Y,B}$ for $\theta(b_y)-b_y$ with size $\leq \Lambda$. Since $(\Gamma,\mcR)$ is $\nu$-robust, for every $i$,
 \begin{align*}    \norm{\pi(r_i)\pi(v_i)}_\rho&=\norm{\pi(r_i)\pi(v_i)\sqrt{\rho}}_F= \norm{\pi(r_i)\big(\pi(v_i)\sqrt{\rho}-\sqrt{\rho}\pi(v_i)\big)+\pi(r_i)\sqrt{\rho}\pi(v_i)}_F\\
 &\leq \norm{r_i}_{\PVM^{Y,B}}\norm{\pi(v_i)\sqrt{\rho}-\sqrt{\rho}\pi(v_i)}_F + \norm{\pi(r_i)\sqrt{\rho}}_F\\
 &\leq  \norm{r_i}_{\PVM^{Y,B}} \deg(v_i)\cdot\nu(\eps)+\nu(\eps).
 \end{align*}
 Here we use the facts that every $\pi(v_i)$ is unitary and $\norm{\cdot}_F$ is right unitarily invariant.  It follows that
 \begin{align*}
     \norm{\pi(b_y)-I^*\wtd{\pi}(b_y)I}_\rho&= \norm{\pi(b_y)-\widehat{\theta}(b_y)}_\rho=\norm{\pi(b_y)-\pi\big(\theta(b_y)\big)}_\rho=\norm{\pi\big(b_y-\theta(b_y)\big)}_\rho\\
     &=\norm{\pi\big(  \sum_{i=1}^n\lambda_iu_ir_iv_i\big)}_\rho\leq \sum_{i=1}^n\abs{\lambda_i}\norm{\pi(r_i)\pi(v_i)}_\rho\\
     &\leq \left(\sum_{i=1}^n \abs{\lambda_i}\big(1+\norm{r_i}_{\PVM^{Y,B}}\deg(v_i)\big) \right)\nu(\eps)\leq\Lambda \nu(\eps).
 \end{align*}
This proves \Cref{eq:GH}. The rest follows from \Cref{prop:gameidealrep}.
\end{proof}

\begin{remark}\label{rmk:GH}
    Suppose a nonlocal game $\mcG$ is a self-test and has a robust determining pair.  From the above theorem, we clearly see that the robustness of this self-test is completely determined by three factors:
\begin{enumerate}[(1)]
    \item How efficiently can the ucp map from the game algebra to the PVM algebra be expressed by relations in $\mcR$.
    \item The robustness of the determining pair.
    \item The spectral gap of this self-test.
\end{enumerate}
Here, the spectral gap (as well as the upper bound of the size of $\mcR$-decompositions) is a constant that is independent of $\eps$, so it can be absorbed in the big-$O$ in \Cref{thm:algGH}. However, we often work with a family of nonlocal games that self-test for growing systems. In this case, the spectral gaps of this sequence of self-tests may depend on the dimension of the systems. For instance, in the family of low-weight Pauli braiding tests $\{\mcG_n\}$~\cite{BMZ23}, every $\mcG_n$ self-tests $n$ EPR pairs, and the spectral gap of $\mcG_n$ is $O(1/ poly(n))$.  For this reason, we keep $\Delta$ in the expression of the robustness function.
\end{remark}

\subsection{Example: CHSH and Clifford algebra}
In the CHSH game, the verifier samples a pair $(x,y)\in \Z_2\times\Z_2$ uniformly at random. Upon receiving $x$ and $y$, Alice and Bob return $a\in\Z_2$ and $b\in\Z_2$ respectively. They win if and only $a+b=x\cdot y$. So CHSH game is an XOR game. The game polynomial of  CHSH is given by
\begin{align*}
    \Phi_{CHSH}:=\frac{1}{2}+\frac{1}{4}(a_0\otimes b_0+a_0\otimes b_1+a_1\otimes b_0-a_1\otimes b_1).
\end{align*}
Here we use the binary observables $a_x:=m^x_0-m^x_1$ and $b_y:=n^y_0-n^y_1$ as generators. 

In \Cref{prop:XORpair}, we have shown that CHSH has a determining pair $(\Gamma,\mcR)$ where $\gamma_0=\tfrac{b_0+b_1}{\sqrt{2}},\gamma_1=\tfrac{b_0-b_1}{\sqrt{2}}$, and $\mcR=\{b_0b_1+b_1b_0\}$. The game algebra associated with CHSH game is $C^*(\mcG)=\PVM^{\Z_2,\Z_2}\slash\ang{\mcR}\cong Cl_2$, the Clifford algebra of rank 2. The mapping $\sigma:Cl_2\arr M_2(\C)$ sending $b_0\mapsto\sigma_X$ and $b_1\mapsto\sigma_Z$ defines a $*$-isomorphism, so it is the unique irreducible representation of $Cl_2$. By \Cref{thm:gametrace}, CHSH is a self-test. The ideal strategy $\wtd{S}:=\big(\wtd{\mcH}_A,\wtd{\mcH}_B,\{\wtd{A}_0,\wtd{A}_1\}.\{\wtd{B}_0,\wtd{B}_1\},\ket{\wtd{\psi}} \big)$ is given by $\wtd{\mcH}_A=\wtd{\mcH}_B=\C^2$, $\wtd{A}_0=\tfrac{\sigma_X+\sigma_Z}{\sqrt{2}},\wtd{A}_1=\tfrac{\sigma_X-\sigma_Z}{\sqrt{2}},\wtd{B}_0=\sigma_X,\wtd{B}_1=\sigma_Z$, and $\ket{\wtd{\psi}}=\tfrac{\ket{00}+\ket{11}}{\sqrt{2}}$. Let $\wtd{\pi}=\wtd{\pi}_A\otimes\wtd{\pi}_B$ be the associated representation. Then 
\begin{align*}
    \wtd{\pi}(\Phi_{CHSH})&=\frac{1}{2}+\frac{\sqrt{2}}{4}(\sigma_X\otimes\sigma_X+\sigma_Z\otimes\sigma_Z)\\
    &=\frac{1}{2}+\frac{\sqrt{2}}{4}\big(\ket{00}+\ket{11} \big)\big(\bra{00}+\bra{11}\big)-\frac{\sqrt{2}}{4}\big(\ket{01}-\ket{10} \big)\big(\bra{01}-\bra{10}\big).
\end{align*}
So the spectral gap $\Delta_{CHSH}$ of CHSH game is $\frac{\sqrt{2}}{2}$. It is well-known that the robustness of the determining pair $(\Gamma,\mcR)$ is $O(\sqrt{\eps})$ (see e.g., \cite{YZ} for a sum-of-squares approach). By \Cref{rmk:GH}, now computing the robustness of the CHSH self-test boils down to constructing the desired ucp map.
\begin{lemma}\label{lemma:ucp}
    There is a ucp map $\theta:Cl_2\arr \PVM^{Y,B}$ sending 
    \begin{align}
        b_0&\mapsto b_0-\frac{1}{2}b_1(b_0b_1+b_1b_0)\label{CHSH1}\\
        b_1&\mapsto b_1-\frac{1}{2}b_0(b_0b_1+b_1b_0),\text{ and}\label{CHSH2}\\
        b_0b_1&\mapsto b_0b_1-\frac{1}{2}(b_0b_1+b_1b_0).\label{CHSH3}
    \end{align}
\end{lemma}
Since $\{1,b_0,b_1,b_0b_1\}$ is a basis for $Cl_2$, \Cref{CHSH1,CHSH2,CHSH3} completely determine $\theta$. Note that $Cl_2$ is the full group $C^*$-algebra of the 1-qubit Pauli group $P_1$. This construction of $\theta$ is given by applying the enhanced Gower-Hatami theorem in \cite{BMZ23} to $P_1$ (see also \cite{YZ} for a proof that $\theta$ is ucp).

\begin{theorem}
    The CHSH game is a robust self-test for its optimal strategies with robustness $O(\sqrt{\eps})$.
\end{theorem}
\begin{proof}
    Let $\theta$ be the ucp map defined in \Cref{lemma:ucp}. Then $\frac{1}{2}b_1(b_0b_1+b_1b_0)$ is an $\mcR$-decomposition for $\theta(b_0)-b_0$ with size $\frac{1}{2}$ and $\frac{1}{2}b_0(b_0b_1+b_1b_0)$ is an $\mcR$-decomposition for $\theta(b_1)-b_1$ with size $\frac{1}{2}$. The spectral gap of CHSH is $\frac{\sqrt{2}}{2}$. The stability of the determining pair $(\Gamma,\mcR)$ is $O(\sqrt{\eps})$. By \Cref{thm:algGH}, the CHSH game is a robust self-test with robustness $O(\sqrt{\eps})$.
\end{proof}

\bigskip

\printbibliography

\appendix
% Add an un-numbered title page before the appendices and a line in the Table of Contents
\section{The closure of finite-dimensional states}\label{AppendixA}
In this appendix, we prove \Cref{prop:dense}:

\AppendixA*

Recall that a \emph{tensor product model} $\models$ is a model in which the Hilbert spaces $\mcH_A$ and $\mcH_B$ are not restricted to be finite-dimensional (but we still assume they are separable). A state $f$ on $\mintensor$ is said to be a \emph{tensor product state} if $f=f_S$ for some tensor product model $S$. To prove that any state on $\mintensor$ can be approximated by finite-dimensional states, we first show that any tensor product state can be approximated by finite-dimensional states.

\begin{lemma}\label{lemma:fdstates}
For every tensor product model $\models$ there is a sequence of finite-dimensional states $\{f_n\}_{n\in\N}$ on $\mintensor$ such that $\lim\limits_{n\arr\infty} f_n=f_S$ in the weak*-topology.
\end{lemma}

\begin{proof}
 If $\mcH_A$ or $\mcH_B$ is finite-dimensional, then $f_S$ is a finite-dimensional state and the lemma follows by taking $f_n=f_S$ for all $n\in \N$. Now suppose $\mcH_A$ and $\mcH_B$ are separable infinite-dimensional Hilbert spaces. The state $\ket{\psi}$ has decomposition
 \begin{align*}
     \ket{\psi}=\sum_{i=1}^\infty\lambda_i\ket{\xi_i}\otimes\ket{\eta_i}
 \end{align*}
 where $\{\lambda_i\}_{i\in\N}$ is a non-negative sequence in $\ell^2(\N)$ with $\sum_{i\in\N}\lambda_i^2=1$, and $\{\ket{\xi_i}:i\in\N\}$ and $\{\ket{\eta_i}:i\in \N\}$ are orthonormal bases for $\mcH_A$ and $\mcH_B$ respectively. For every $n\in\N$ we define projections $\Pi^A_n:=\sum_{i=1}^n\ket{\xi_i}\bra{\xi_i}\in\msB(\mcH_A)$ and $\Pi^B_n:=\sum_{i=1}^n\ket{\eta_i}\bra{\eta_i}\in\msB(\mcH_B)$, finite-dimensional spaces $\mcH^A_n:=\Pi^A_n\mcH_A$ and $\mcH^B_n:=\Pi^B_n\mcH_B$, and vector state $\ket{\psi_n}:=\frac{1}{\mu_n}\sum_{i=1}^n\lambda_i\ket{\xi_i}\otimes\ket{\eta_i}\in \mcH^A_n\otimes \mcH^B_n$ where $\mu_n:=\sqrt{\sum_{i=1}^n\lambda_i^2}$. Then for every $n\in \N$,
 \begin{align*}
    S_n:=(\mcH^A_n,\mcH^B_n,\{\Pi^A_nM^x_a\Pi^A_n \},\{\Pi^B_nN^y_b\Pi^B_n\},\ket{\psi_n}) 
 \end{align*}
is a quantum model. Let $\pi_A\otimes \pi_B$ be the associate representation of $S$, and let $\pi^A_n\otimes\pi^B_n$ be the associate representation of $S_n$ for every $n\in\N$.
\begin{claim}\label{claimA1}
    Given $\epsilon\geq 0$, for every monomial $\alpha$ in $\POVM^{X,A}$, there exists a $T_A\in\N$ such that 
    \begin{align}
\pi^A_n(\alpha)\otimes\Id\ket{\psi_n}&\approx_\epsilon\pi_A(\alpha)\otimes \Id\ket{\psi}\label{eq:n-alpha}
\end{align}
for all $n\geq T_A$.
\end{claim}
\begin{proof}[Proof of \Cref{claimA1}]
    We prove this claim by induction on the degree $k$ of the monomial $\alpha$. Since $\mu_n\arr 1$ as $n\arr\infty$, there is a $T\in \N$ such that $1-\mu_n\leq \eps/2$ and $\sqrt{1-\mu_n^2}\leq \epsilon/2$ for all $n\geq T$. So 
\begin{align*}
    \norm{\ket{\psi_n}-\ket{\psi}}\leq \norm{(1-\mu_n)\ket{\psi_n}}+\norm{\mu_n\ket{\psi_n}-\ket{\psi}}\leq 1-\mu_n+\sqrt{1-\mu_n^2}\leq \epsilon
\end{align*}
for all $n\geq T$. \Cref{eq:n-alpha} holds for $k=0$. Now assume it holds for all monomials of degree $k$ for some $k\geq 0$. Let $\alpha=\alpha_1\cdots\alpha_{k+1}$ be a monomial of degree $k+1$, and let $\alpha':=\alpha_{2}\cdots \alpha_{k+1}$. By the induction hypothesis, there exists a $T_1\in \N$ such that 
\begin{align*}
\pi^A_n(\alpha')\otimes\Id\ket{\psi_n}\approx_{\epsilon/2}\pi_A(\alpha')\otimes \Id\ket{\psi}
\end{align*}
for all $n\geq T_1$. For every $i,j\in \N$, let $s_{ij}:=\bra{\xi_i,\eta_j}\pi_A(\alpha)\otimes \Id\ket{\psi}$. Since $\sum_{i,j\geq 1}\abs{s_{ij}}^2=\norm{\pi_A(\alpha)\otimes \Id\ket{\psi}}^2\leq 1$, there is a $T_2\in\N$ such that 
\begin{align*}
    \sum_{i,j=1}^{n}\abs{s_{ij}}^2\geq \sum_{i,j\geq 1}\abs{s_{ij}}^2 -(\eps/2)^2
\end{align*}
for all $n\geq T_2$. Note that $\norm{\pi_A(\alpha_1)}_{op}\leq 1$. It follows that
\begin{align*}
    \pi_A(\alpha)\otimes\Id\ket{\psi}&\approx_{\eps/2} \Pi^A_n\otimes\Pi^B_n\big(\pi_A(\alpha)\otimes\Id\ket{\psi}\big)=\Pi^A_n\pi_A(\alpha_1)\otimes\Pi^B_n \big(\pi_A(\alpha')\otimes\Id\ket{\psi} \big)\\
    &\approx_{\eps/2}\Pi^A_n\pi_A(\alpha_1)\otimes\Pi^B_n \big(\pi^A_n(\alpha')\otimes\Id\ket{\psi_n} \big)= \Pi^A_n\pi_A(\alpha_1)\Pi^A_n\pi^A_n(\alpha')\otimes \Id\ket{\psi_n}\\
    &=\pi^A_n(\alpha_1)\pi^A_n(\alpha')\otimes \Id\ket{\psi_n}=\pi^A_n(\alpha)\otimes \Id\ket{\psi_n}
\end{align*}
for all $n\geq T:=\max\{T_1,T_2\}$. So \Cref{eq:n-alpha} holds for all monomials of degree $k+1$. This completes the proof.
\end{proof}
Similarly, for any $\epsilon\geq 0$ and monomial $\beta\in \POVM^{Y,B}$, there exists a $T_B\in\N$ such that
\begin{align}
\Id\otimes \pi^B_n(\beta^*)\ket{\psi_n}&\approx_\epsilon\Id\otimes\pi_B(\beta^*)\ket{\psi}.\label{eq:n-beta}
\end{align}
for all $n\geq T_B$. Then by \Cref{lemma:diff}, $\lim\limits_{n\arr \infty}f_{S_n}(\alpha\otimes\beta)=f(\alpha\otimes\beta)$ for all monomials $\alpha$ and $\beta$. So the lemma follows. 
 \end{proof}

The above lemma illustrates that the weak*-closure of the set of finite-dimensional states contains all tensor product states. To prove the set of finite-dimensional states is weak*-dense, we only need to show that the set of tensor product states is weak*-dense. Recall the following well-known fact (see e.g. \cite[Corollary 4.3.10]{KR97}).

\begin{lemma}\label{lemma:mixedstate}
    Suppose $\mcA\subseteq\msB(\mcH)$ is a concretely represented unital $C^*$-algebra, and let $f$ be a state on $\mcA$. Then for every $\epsilon>0$ and every finite set of elements $\alpha_1,\cdots,\alpha_n\in\mcA$, there are coefficients $\lambda_1,\cdots,\lambda_k\geq 0$ with $\sum_{i=1}^k\lambda_i=1$ and vectors $\ket{\psi_1},\cdots,\ket{\psi_k}\in \mcH$ such that 
    \begin{align*}
        \abs{f(\alpha_i)-\sum_{j=1}^k\lambda_j\bra{\psi_j}\alpha_i\ket{\psi_j}}\leq \epsilon
    \end{align*}
    for all $1\leq i\leq n$.
\end{lemma}

This lemma can be used to prove:
\begin{lemma}\label{lemma:qsstates}
    For every state $f$ on $\mintensor$ there is a sequence of tensor product models $\{S_n\}_{n\in N}$ such that $\lim\limits_{n\arr\infty}f_{S_n}=f$ in the weak*-topology.
\end{lemma}

\begin{proof}
    We first fix any faithful representations $\POVM^{X,A}\subseteq \msB(\mcH_A)$ and $\POVM^{Y,B}\subseteq \msB(\mcH_B)$ for some Hilbert spaces $\mcH_A$ and $\mcH_B$. Then $\mintensor$ is the $C^*$-algebra generated by the concrete operators $m^x_a\otimes n^y_b$, $(a,b,x,y)\in A\times B\times X\times Y$, where each $m^x_a\otimes n^y_b$ acts on $\mcH_A\otimes \mcH_B$. For every $n\in\N$, let 
    \begin{align*}
        C(n):=\{\alpha\otimes\beta: \alpha \text{ monomial in }\POVM^{X,A}, \beta \text{ monomial in }\POVM^{Y,B}, \deg(\alpha)+\deg(\beta)\leq n \}
    \end{align*}
    be the set of monomials in $\mintensor$ of degree $\leq n$. Then $C(n)$ is a finite set and $C(n)\subset C(n+1)$ for all $n\in\N$. For every $n\in\N$, by \Cref{lemma:mixedstate} there are coefficients $\lambda_1^{(n)},\cdots,\lambda^{(n)}_{k_n}\geq 0$ with $\sum_{i=1}^{k_n}\lambda^{(n)}_i=1$ and vectors $\ket{\psi^{(n)}_1},\cdots,\ket{\psi^{(n)}_{k_n}}\in \mcH_A\otimes \mcH_B$ such that 
    \begin{align}
        \abs{f(\alpha\otimes \beta)-\sum_{i=1}^{k_n}\lambda^{(n)}_i\bra{\psi^{(n)}_i}\alpha\otimes\beta\ket{\psi^{(n)}_i}}\leq \frac{1}{n} \text{ for all }\alpha\otimes\beta\in C(n).\label{eq:mixedstate}
    \end{align}
    Then $S_n:=\big(\mcH_A,\mcH_B,\{m^x_a\},\{n^y_b\},\rho_n  \big)$ is a tensor product model with a mixed state $\rho_n:=\sum_{i=1}^{k_n}\lambda_i\ket{\psi_i}\bra{\psi_i}$. \Cref{eq:mixedstate} implies 
    \begin{align*}
        \abs{f(\alpha\otimes \beta)-f_{S_n}(\alpha\otimes\beta)}\leq \frac{1}{n} \text{ for all }\alpha\otimes\beta\in C(n).
    \end{align*}
   Then for every $\epsilon>0$ and every monomial $\alpha\otimes\beta\in \mintensor$, 
   \begin{align*}
       \abs{f(\alpha\otimes \beta)-f_{S_n}(\alpha\otimes \beta)}\leq \epsilon
   \end{align*}
   for all $n\geq\max\{\lceil \frac{1}{\epsilon}\rceil,\deg(\alpha)+\deg(\beta) \}$. So $\lim\limits_{n\arr\infty}f_{S_n}(\alpha\otimes\beta)=f(\alpha\otimes\beta)$ for any monomial $\alpha\otimes\beta$. The lemma follows.
\end{proof}

\begin{proof}[Proof of \Cref{prop:dense}] Let $f$ be a state on $\mintensor$. By \cref{lemma:qsstates} there is a sequence of tensor product models $\{S_n\}_{n\in\N}$ such that $f=\lim\limits_{n\arr\infty}f_{S_n}$ in the weak*-topology. Then for every $n\in\N$, by \Cref{lemma:fdstates} there is a sequence of finite-dimensional states $\{f^{(n)}_m\}_{m\in \N}$ such that $f_{S_n}=\lim\limits_{m\arr\infty}f^{(n)}_{m}$ in the weak*-topology. Hence $\{f^{(n)}_n\}_{n\in\N}$ is a sequence of finite-dimensional states such that 
$f=\lim\limits_{n\arr\infty}f^{(n)}_n$ in the weak*-topology. We conclude that the set of finite-dimensional states on $\mintensor$ is a weak*-dense subset of the state space of $\mintensor$.  
\end{proof}

\end{document}